\documentclass[twocolumn,prb,superscriptaddress,10pt]{revtex4-2}
\usepackage{array}[=2016-10-06]
\usepackage{booktabs,tabularx}
\usepackage{graphicx}
\usepackage{amsmath, amsfonts, amssymb, bm, amsthm}
\usepackage{bbm}
\usepackage{tcolorbox}
\usepackage{physics}
\usepackage{blkarray}
\usepackage[dvipsnames]{xcolor}
\usepackage[colorlinks, linkcolor=BrickRed, citecolor=RoyalBlue,urlcolor=NavyBlue, breaklinks]{hyperref}
\usepackage[all]{hypcap} 
\usepackage{multirow}
\usepackage{overpic}
\numberwithin{equation}{section}
\renewcommand{\theequation}{\arabic{section}.\arabic{equation}}
\newcommand{\nn}{\nonumber \\}

\tcbuselibrary{breakable}

\newtcolorbox{emphbox}[1][]{
  breakable,
  colback=blue!15,    
  boxrule=0pt,     
  arc=3mm,           
  left=2mm,          
  right=2mm,         
  top=2mm,           
  bottom=2mm,        
  auto outer arc,
  #1
}

\usepackage{blkarray}
\usepackage{tikz}
\usetikzlibrary{positioning}

\definecolor{BSorange}{RGB}{140,50,0}

\theoremstyle{plain}
\newtheorem{definition}{Definition}[section]
\newtheorem{theorem}{Theorem}[section]
\newtheorem{lemma}[theorem]{Lemma}
\newtheorem{proposition}[theorem]{Proposition}
\newtheorem{corollary}{Corollary}[theorem]
\newtheorem{porism}{Porism}[theorem]

\renewcommand{\thedefinition}{\arabic{section}.\arabic{definition}}
\renewcommand{\thetheorem}{\arabic{section}.\arabic{theorem}}

\begin{document}

\title{Transparent Domain Walls through Information Convex Sets}

\author{Jintae \surname{Kim}}
\email[Electronic address:$~~$]{jintae@illinois.edu}
\affiliation{Physics Department and Institute of Condensed Matter Theory, University of Illinois at Urbana-Champaign, Urbana, Illinois 61801, USA}
\affiliation{Institute of Basic Science, Sungkyunkwan University, Suwon 16419, South Korea}

\author{Amanda \surname{Gatto Lamas}}
\affiliation{Physics Department and Institute of Condensed Matter Theory, University of Illinois at Urbana-Champaign, Urbana, Illinois 61801, USA}
\author{Jacopo \surname{Gliozzi}}
\affiliation{Physics Department and Institute of Condensed Matter Theory, University of Illinois at Urbana-Champaign, Urbana, Illinois 61801, USA}
\affiliation{Department of Physics and Center for Theory of Emergent Quantum Matter, Pennsylvania State University, University Park, Pennsylvania 16802, USA}
\author{Bowen \surname{Shi}}
\affiliation{Physics Department and Institute of Condensed Matter Theory, University of Illinois at Urbana-Champaign, Urbana, Illinois 61801, USA}
\affiliation{ Beijing Institute of Mathematical Sciences and Applications, Beijing 101408, China }
\author{Taylor L. \surname{Hughes}}
\affiliation{Physics Department and Institute of Condensed Matter Theory, University of Illinois at Urbana-Champaign, Urbana, Illinois 61801, USA}
\author{Jong Yeon \surname{Lee}}
\email[Electronic address:$~~$]{jongyeon@illinois.edu}
\affiliation{Physics Department and Institute of Condensed Matter Theory, University of Illinois at Urbana-Champaign, Urbana, Illinois 61801, USA}
\affiliation{Korea Institute for Advanced Study, Seoul 02455, South Korea}

\begin{abstract}

In $(2+1)$-dimensional topologically ordered many-body states, transparent (topologically deformable) domain walls are invisible to local topological probes, yet can modify the ground state degeneracy (GSD) and transmute anyons transported across them. Here, we develop an entanglement-bootstrap framework using information convex sets (ICSs) on local and noncontractible annuli to extract information about transparent domain walls directly from ground state wavefunctions at fixed points of Abelian topological phases on a torus, without taking categorical defect data as input. We derive fusion rules governing the action of anyons on extreme points of ICSs on noncontractible annuli and determine their quantum dimensions. Extreme points invariant under transport around the complementary cycle correspond one-to-one to minimum entropy states (MESs), and their number equals the GSD. The maximal topological entanglement entropy (TEE), $\gamma_{\rm LW}^{\max}=\log(\mathcal{D}/d_\alpha)$, probes the net effect of walls crossing the chosen annulus, where $\mathcal{D}$ is the total quantum dimension and $d_\alpha$ is the quantum dimension of an extreme point of its ICS. In contrast, the maximal entanglement asymmetry is $\Delta S_X^{\max}=\log\mathrm{GSD}$. This value is the same for both annulus orientations and reflects the combined effect of the transparent domain walls. We further show that transparent domain walls can give rise to symmetries supported jointly on the two chosen fundamental cycles that cannot be decomposed into a product of two $1$-form symmetry operators, one supported on each cycle. By relating their action on MESs to anyon tunneling and the fusion rules, we clarify how these symmetries connect distinct ground states. Additionally, we apply the framework to Wen's plaquette model, the anisotropic dipolar toric code, and the rank-2 toric code.

\end{abstract}

\date{\today}
\maketitle


\section{Introduction}

Many-body topological phases in $(2+1)$ dimensions exhibit distinctive phenomena, including anyonic excitations~\cite{kitaev_anyons_2006, kitaev03} and universal signatures such as topological ground state degeneracy (GSD) and topological entanglement entropy (TEE)~\cite{kitaev_tee_2006,levin_tee_2006}. Due to their intrinsic topological properties, without any underlying 0-form symmetry, such phases can host \emph{topological defects}, ranging from point defects to gapped domain walls. These defects act nontrivially on anyons, potentially modifying global topological sectors~\cite{bombin_topological_2010, PhysRevB.100.115147, bais_theory_2009, kitaev_models_2012,barkeshli2013,barkeshli2013theory,barkeshli2013classification,khan2014,teo2015,lan_gapped_2015, shi_domain_2021, li_domain_2024, buican_algebraic_2025, Kong2014, Fuchs2013,Cong2016,shi_characterizing_2019, shi_entanglement_2021, williamson_symmetry-enriched_2026}.

\begin{figure}[!t]
\centering
\includegraphics[width=0.6\linewidth]{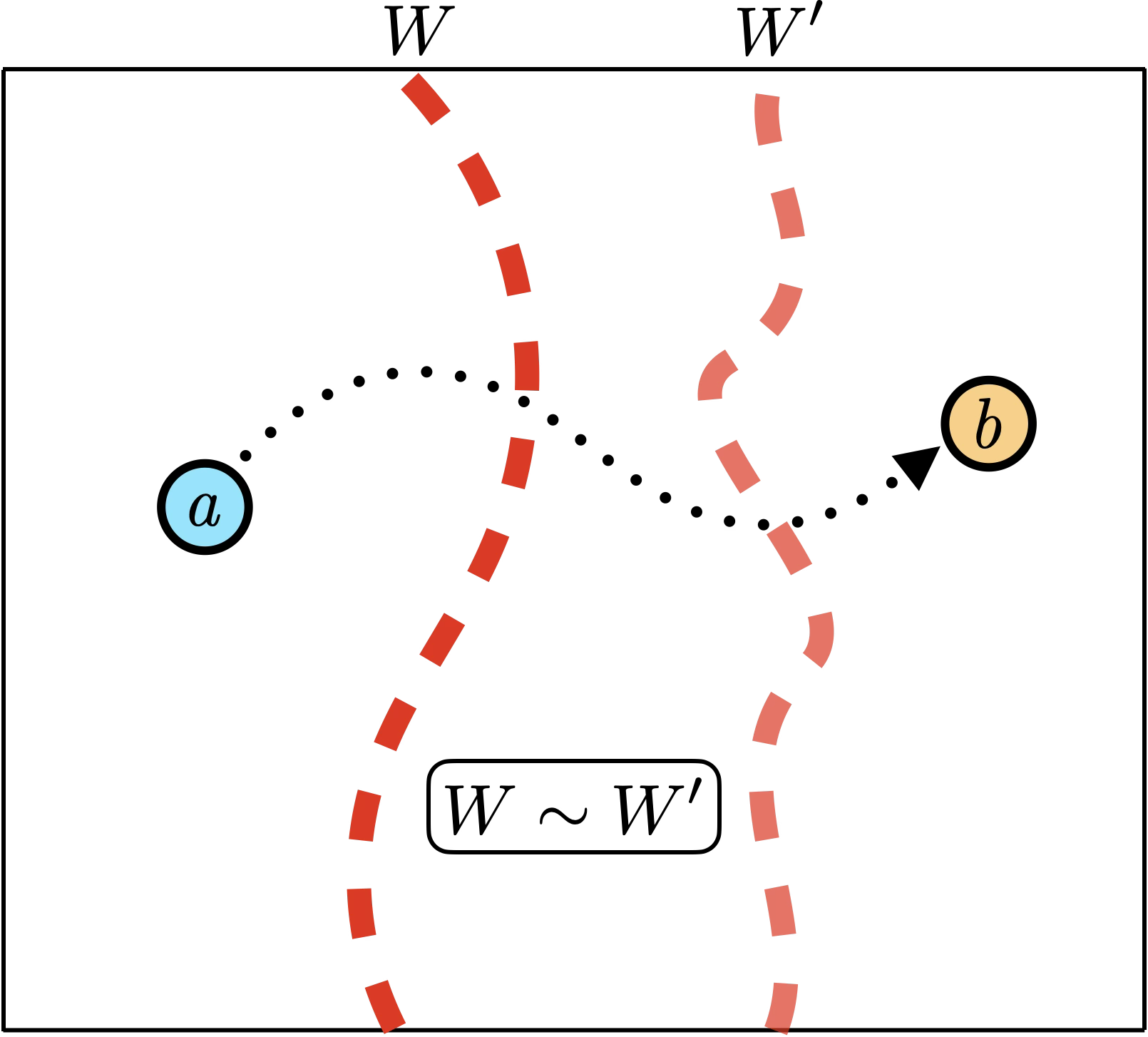}
\caption{ {\bf Transparent domain walls.} Transporting an anyon $a$ along the black dotted line across a transparent domain wall transmutes $a$ into $b$, i.e., $a \rightarrow b$. The two walls $W$ and $W'$ represent physically equivalent choices of the wall position, illustrating that the position of a transparent domain wall is not a physical observable. Thus, a local topological probe cannot detect the wall: any segment of $W$ can be deformed to $W'$ and moved away from the probed region.
}
\label{fig:tdw}
\end{figure}

While the properties and diagnostics of these defects have been extensively studied, transparent domain walls~\cite{buican_algebraic_2025}, a special class of gapped domain walls, pose a distinct challenge. They are intrinsically topological, avoiding any local topological diagnostics; in fact, it is not possible to locate these domain walls. In some sense, they are \emph{topological boundary condition}s, dictating how anyons permute when they are transported around a noncontractible loop of the system. Accordingly, the presence of transparent domain walls can modify the GSD~\cite{kitaev_models_2012, lan_gapped_2015, buican_algebraic_2025}. A schematic illustration of transparent domain walls is shown in Fig.~\ref{fig:tdw}.

Gapped domain walls are topological interfaces between potentially distinct topological phases. Across such an interface, anyons need not pass through unchanged: they may tunnel into anyons on the other side, terminate or condense on the wall, or become confined, while the wall itself supports gapped point-like topological excitations, dubbed \emph{wall superselection sectors}. Accordingly, gapped walls can be characterized in several complementary ways, including by anyon-tunneling data~\cite{kitaev_models_2012,lan_gapped_2015,Kong2014}, by wall superselection sectors~\cite{Fuchs2013,Cong2016,shi_characterizing_2019,shi_entanglement_2021,buican_algebraic_2025}, and by module rules describing how anyons from the two sides act on these sectors. Moreover, wall-local diagnostics can be formulated based on suitable regions contained in a disk intersecting the wall. 
Such diagnostics characterize a range of anyon data, including which anyons condense into the wall vacuum, the wall superselection sectors and their fusion rules, and the attachment of anyons to these sectors~\cite{shi_characterizing_2019, shi_entanglement_2021}.

Transparent domain walls form a special class for which the usual wall-local diagnostics are largely absent. Since their position can be topologically deformed away from any chosen local topological diagnostics, their characteristic information is instead global: transporting an anyon around a noncontractible cycle can return it as a different anyon type. Transparent domain walls have been studied through microscopic models and categorical theory~\cite{ PhysRevB.100.115147, kitaev_models_2012, Fuchs2013, bombin_topological_2010, barkeshli_topological_2012, kim_unveiling_2025}. 
Within these approaches, anyon permutation, GSD, and TEE have been characterized. 
These approaches, however, typically take the defect structure or its categorical description as an input, rather than deriving its topological structure directly from the ground state wavefunctions.

\begin{emphbox}
    \,\, This leads to the central question: what information associated with transparent domain walls can be derived directly from ground state wavefunctions, without assuming a microscopic realization or an a priori categorical description? We address this question through the entanglement bootstrap, which derives emergent topological structures from local entropic conditions and the consistency of reduced density matrices. Within this approach, we develop an entanglement-based framework for characterizing transparent domain walls at fixed points of Abelian topological phases on a torus.     
\end{emphbox}

The key object is the \emph{information convex set} (ICS): the set of density matrices on a manifold ${\cal M}$ that are locally consistent with a reference state $\sigma$~\cite{shi_characterizing_2019,shi_fusion_2020}. 
The ICS faithfully captures the topological structure of the underlying state on ${\cal M}$ and has been widely used in the entanglement bootstrap to reconstruct emergent topological data from wavefunctions~\cite{shi_fusion_2020,shi_immersed_2024,shi_deciphering_2018,buican_algebraic_2025,shi_seeing_2019,shi_entanglement_2021,shi_domain_2021,shi_characterizing_2019,shi_verlinde_2020, kim_modular_2022, huang_knots_2023, shi_remote_2025}, as well as to distinguish inequivalent mixed-state phases~\cite{yang_topological_2025}. The extreme points of the ICS on a local annulus, i.e., an annulus contained in a disk, correspond to anyon types~\cite{shi_fusion_2020}. 
However, local annuli cannot distinguish different transparent domain walls; doing so requires noncontractible annuli, whose boundary components are themselves noncontractible. We therefore consider the ICS on both local and noncontractible annuli, whose interplay exposes the global structure associated with transparent domain walls.

The paper is organized as follows. In Sec.~\ref{sec:2}, we summarize the main results of this work. In Sec.~\ref{sec:3}, we introduce ICSs on local and noncontractible annuli, define transparent domain walls and their associated fusion rules, and derive their basic properties. Then, in Sec.~\ref{sec:4}, we specialize to Abelian topological phases and derive further results for the TEE, symmetries, and entanglement asymmetry. We provide concrete examples in Sec.~\ref{sec:5}, where we apply our framework to Wen's plaquette model~\cite{wen_quantum_2003, you_projective_2012}, the anisotropic dipolar toric code~\cite{ebisu_anisotropic_2023}, and the rank-2 toric code~\cite{oh22a, pace-wen, oh22b, oh23, kim_unveiling_2025,kim_gauging_2026}, and summarize the resulting properties of the corresponding ICSs. We conclude with a discussion in Sec.~\ref{sec:6}.

\section{Summary of main results}
\label{sec:2}

\begin{figure}[!t]
\centering
\begin{overpic}[width=0.45\linewidth]{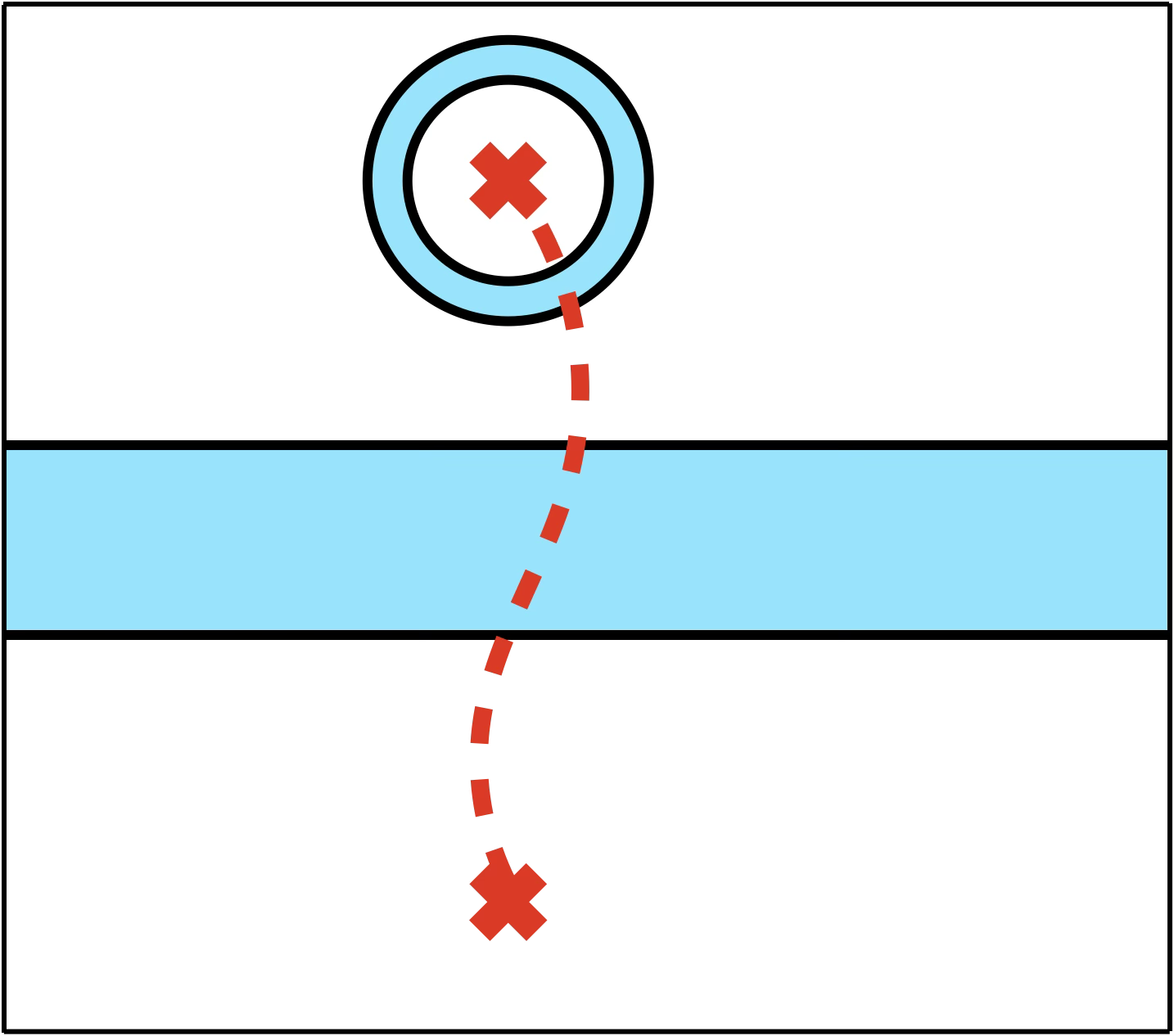}
\put(-10,90){(a)}
\end{overpic}
~~
\begin{overpic}[width=0.45\linewidth]{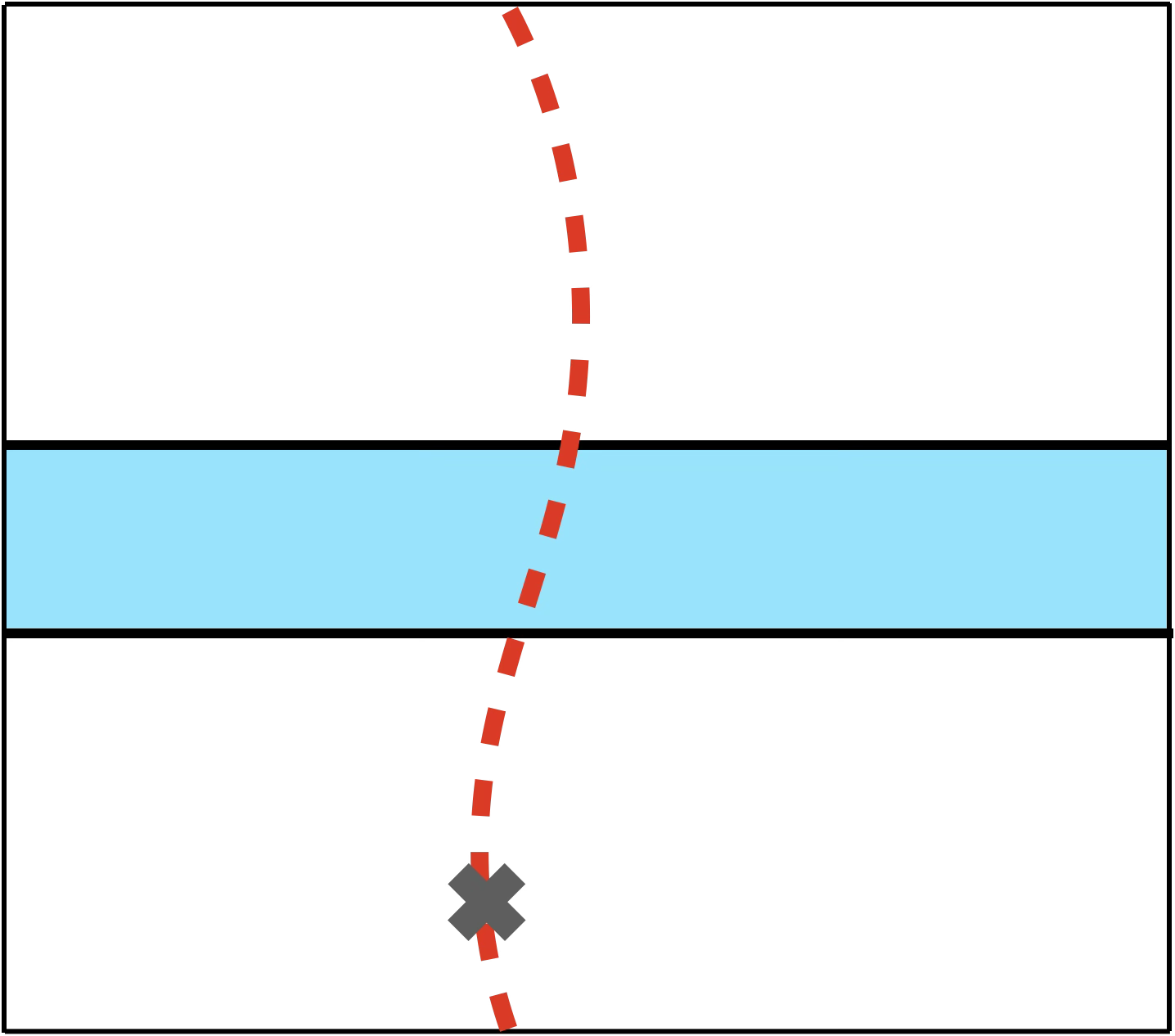}
\put(-10,90){(b)}
\end{overpic}      
\caption{(a) Illustration of two point defects (red $\times$ symbols) connected by a domain wall (red dotted line) on a torus. The local and noncontractible annuli used in the ICS analysis are also shown. (b) Closed domain wall winding around a noncontractible cycle, obtained by transporting one defect around the torus along a noncontractible cycle. The gray $\times$ symbol indicates the disappearance of the two point defects.}
\label{fig:pd}
\end{figure}

\begin{figure*}
    \centering
\begin{overpic}[width=0.8\linewidth]{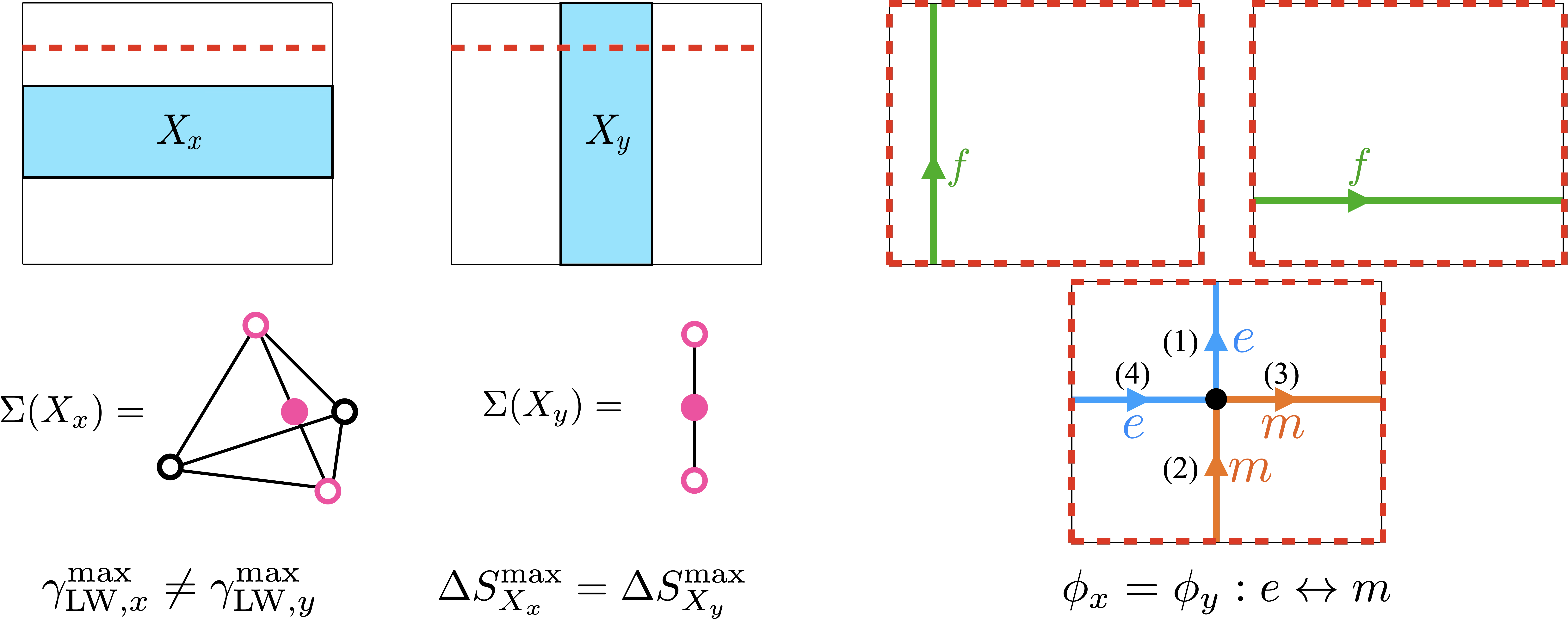}
\put(-2,40){(a)}
\put(53,40){(b)}
\end{overpic} 
\caption{
\textbf{(a) Probing transparent domain walls.}
A transparent domain wall (red dashed line) winds around the $x$-cycle. The noncontractible annulus $X_y$ crosses the wall, whereas $X_x$ runs parallel to it. The wall modifies $\Sigma(X_y)$ but leaves $\Sigma(X_x)$ unchanged, leading to different maximal TEEs for the two annuli. Open pink circles denote extreme points obtained as reductions of MESs, while open black circles denote the remaining ICS extreme points. The filled pink circle marks the equal-weight barycenter of the MES-associated extreme points, representing the symmetrized state. The maximal entanglement asymmetry is the entropy of this barycenter minus that of an MES-associated extreme point. It takes the same value for both annuli despite their different ICS structures and remains sensitive to the wall even when the wall runs parallel to the annulus.
\textbf{(b) Symmetries in the presence of transparent domain walls.}
Symmetry generators of Wen's plaquette model with transparent domain walls along both fundamental cycles, inducing anyon automorphisms $\phi_x$ and $\phi_y$. Transport of $f$ continues to generate $1$-form symmetries supported on individual cycles. The displayed generator associated with $e$ and $m$ is implemented by the anyon-transport sequence shown in steps (1)--(4) and has joint support on both chosen fundamental cycles. Arrows indicate anyon transport, the black circle marks an anchoring point, and red dashed lines denote transparent domain walls. Further details are provided in Sec.~\ref{sec:5a}.  
}
    \label{fig:sym}
\end{figure*}

Here, we summarize our main results for the reader's convenience. Table~\ref{tab:notation} lists the principal notation used throughout the paper.

A heuristic way to understand why the ICS framework is useful for characterizing transparent domain walls is that an ICS resolves the topological sectors associated with a region. In the absence of defects, the ICS of a local annulus forms a simplex whose extreme points are labeled by the possible topological charges enclosed by the annulus~\cite{shi_fusion_2020}. A transparent domain wall, meanwhile, can act nontrivially on these sectors by permuting anyon types transported across it. This suggests that its topological action can be probed using ICSs associated with suitably chosen noncontractible annuli that cross the domain wall.

The connection between these two settings can be understood through a defect--antidefect pair. Consider a state on a torus containing such a pair connected by an anyon-permuting domain wall, as shown in Fig.~\ref{fig:pd}(a). We compare the ICS of a local annulus enclosing one of the point defects with that of a noncontractible annulus crossing the connecting domain wall. As proved in Appendix~\ref{asec:pd}, these two ICSs are isomorphic. Thus, the fusion rules associated with a point defect can equivalently be captured by a noncontractible annulus probing its associated domain wall.


This observation naturally leads to the closed transparent domain walls considered in this work. Starting from a defect--antidefect pair, one may transport one defect around a noncontractible cycle of the torus and then annihilate it with the antidefect, as shown in Fig.~\ref{fig:pd}(b)~\cite{barkeshli_topological_2012}. The point-defect endpoints disappear, leaving behind a closed domain wall winding around the cycle. There is then no point defect for a local annulus to enclose, whereas the ICS of a noncontractible annulus crossing the remaining domain wall continues to provide a well-defined probe of its topological action. Accordingly, throughout this work we characterize transparent domain walls using ICSs associated with noncontractible annuli. Note that this does not mean that the global effects of point defects and transparent domain walls are the same. One simple way to distinguish between point defects and transparent domain walls is to examine their effects on the GSD: point defects may increase the GSD, but transparent domain walls may decrease it.

Our analysis of transparent domain walls addresses two main aspects: probing transparent domain walls and symmetries in the presence of transparent domain walls. Both aspects are illustrated in Fig.~\ref{fig:sym}. It is sufficient to consider the two chosen fundamental-cycle classes of noncontractible annuli, winding around the $x$ and $y$ cycles, respectively.


\vspace{5pt}\noindent {\bf 1. Probing transparent domain walls.}

Our framework extracts information about transparent domain walls from the structure and transport properties of ICSs on noncontractible annuli. These properties determine the GSD and underlie two complementary entanglement diagnostics: TEE and entanglement asymmetry.

First, the structure of the ICS on a noncontractible annulus probes transparent domain walls crossing the annulus. By Theorem~\ref{th:cc}, the number of extreme points equals the number of anyon types invariant under transport along the noncontractible direction of the annulus. A number of extreme points strictly smaller than the total number of anyon types therefore signals the presence of a transparent domain wall crossing the annulus.

We further determine the fusion rules governing the action of anyons on these extreme points and assign a quantum dimension to each extreme point. For Abelian topological phases, the quantum dimensions $d_\alpha$ of all extreme points $\alpha$ of the ICS on a noncontractible annulus are equal and satisfy
\[
d_\alpha=\sqrt{
\frac{\text{number of anyon types}}
{\text{number of extreme points}}
}.\]
Consequently, $d_\alpha$ can exceed unity even though every anyon has unit quantum dimension. These dimensions encode the effects of crossing transparent domain walls on the ICS on the noncontractible annulus.

This structure has a direct entanglement signature in the TEE, defined as one-half of the conditional mutual information for an appropriate tripartition of the noncontractible annulus. This combination cancels the leading contributions proportional to the boundary length~\cite{kitaev_tee_2006,levin_tee_2006}. The maximal TEE on a noncontractible annulus is
\[
\gamma_{\rm LW}^{\rm max}=\log (\mathcal{D}/d_\alpha),
\]
where $\mathcal{D}$ is the total quantum dimension [Theorem~\ref{th:TEE}]. Its dependence on $d_\alpha$ makes the maximal TEE sensitive to the walls crossing the chosen annulus.

Second, transporting the annulus around the complementary cycle probes walls parallel to it. If the number of transport-invariant extreme points is strictly smaller than the total number of extreme points, a transparent domain wall that does not cross the annulus must be present. The transport-invariant extreme points are precisely the reduced density matrices of MESs and are in one-to-one correspondence with them [Propositions~\ref{pro:Ctau} and \ref{pro:oto}]. Their number equals the GSD [Proposition~\ref{pro:GSD}]. Annulus transport thus provides both a diagnostic of parallel walls and a means of determining the GSD.

These transport properties, together with the symmetry action, also determine the maximal entanglement asymmetry. For a reduced density matrix on a noncontractible annulus, entanglement asymmetry measures the increase in entanglement entropy under averaging over this effective symmetry action, \(S(\rho_X^{\mathrm{sym}})-S(\rho_X)\). Its maximal value on a noncontractible annulus is
\[
\Delta S_X^{\rm max}=\log\mathrm{GSD}
\]
[Theorem~\ref{th:ea} and Proposition~\ref{pro:maximal-entanglement-asymmetry}]. This maximal value is the same for both annulus orientations and can respond to walls along either fundamental cycle through their effect on the GSD. The two entanglement diagnostics therefore probe distinct information: the maximal TEE reflects the net effect of walls crossing the chosen annulus, whereas the maximal entanglement asymmetry reflects the restriction on the global ground state space imposed by the full wall configuration.

\begin{table*}[t]
\caption{Principal notation used throughout the paper. The paired indices $x(y)$ and $y(x)$ mean either $(x,y)$ or $(y,x)$.}
\label{tab:notation}
\centering
\small
\renewcommand{\arraystretch}{1.05}
\setlength{\tabcolsep}{4pt}
\begin{tabularx}{\textwidth}{@{}
  >{\raggedright\arraybackslash}p{0.23\textwidth}
  >{\raggedright\arraybackslash}X@{}}
\toprule
Symbol & Meaning \\
\midrule

\multicolumn{2}{@{}l}{\textit{Geometry and information convex sets}} \\
\addlinespace[2pt]

$\mathtt{R}$
& Full system, topologically equivalent to a torus. \\

$l_x,l_y$; $l_{(n,m)}$
& Oriented fundamental-cycle loops; a loop of homology class $(n,m)$. \\

$\Delta(B,C,D)_\rho$
& $S(\rho_{BC})+S(\rho_{CD})-S(\rho_B)-S(\rho_D)$, with $S$ the
von Neumann entropy. \\


$\Sigma(X)$; $\operatorname{ext}\Sigma(X)$
& Information convex set on $X$; its set of extreme points. \\

$\mathtt{A}_{x(y)}$
& Family of noncontractible annuli containing loops $l_{x(y)}$. \\

\addlinespace[3pt]
\multicolumn{2}{@{}l}{\textit{Sector labels, transport, and fusion}} \\
\addlinespace[2pt]

$\mathcal{C}$; $1,\bar a$
& The set of anyon types; the vacuum and the antiparticle of $a$. \\

$\phi_{(n,m)}$; $\phi_x,\phi_y$
& Anyon-label automorphisms induced by local-annulus transport around
$l_{(n,m)}$; $\phi_x=\phi_{(1,0)}$ and $\phi_y=\phi_{(0,1)}$. \\

$\mathcal{C}^{\phi_{(n,m)}}$
& Invariant anyon types:
$\{a\in\mathcal{C}\mid\phi_{(n,m)}(a)=a\}$. \\

$\mathfrak{C}_{x(y)}$
& The set of the labels of noncontractible-annulus extreme points for
$X\in\mathtt{A}_{x(y)}$. \\

$\rho_X^a$; $\rho_X^\alpha$
& Extreme states on local and noncontractible annuli, respectively;
also denoted by $\sigma_X^a$ and $\sigma_X^\alpha$ in parts of the text. \\

$\tau_{y(x)}$
& Automorphism of $\mathfrak{C}_{x(y)}$ induced by transporting a
noncontractible annulus around the complementary cycle $l_{y(x)}$. \\

$\mathfrak{C}_{x(y)}^{\tau_{y(x)}}$
& Sectors fixed by $\tau_{y(x)}$; precisely those realized by MES reductions. \\

$N_{ab}^{c}$; $N_{a\alpha}^{\beta}$
& Fusion multiplicities for $a\times b\to c$ and
$a\times\alpha\to\beta$. \\

$\Sigma_{a\alpha}^{\beta}(Y)$
& Subset of the punctured-annulus ICS with boundary sectors
$a,\alpha,\beta$ fixed. \\

$\mathbb{V}_{a\alpha}^{\beta}$;
$\mathcal{S}(\mathbb{V}_{a\alpha}^{\beta})$
& Fusion Hilbert space of dimension $N_{a\alpha}^{\beta}$;
its density-matrix state space. \\

$d_a$; $d_\alpha$
& Quantum dimensions of an anyon and a noncontractible-annulus
sector, respectively. \\

$\mathcal{D}$
& Total quantum dimension:
$\mathcal{D}=\sqrt{\sum_{a\in\mathcal{C}}d_a^2}$. \\

$E_{x(y)}$ 
& Anyon equivalence classes: $a\overset{x(y)}{\sim}b$ iff
$a\times\alpha=b\times\alpha$ for all $\alpha\in\mathfrak{C}_{x(y)}$. \\

\addlinespace[3pt]
\multicolumn{2}{@{}l}{\textit{Ground states, symmetries, and entanglement}} \\
\addlinespace[2pt]

$|\alpha\rangle$
& MES associated with $X\in\mathtt{A}_{x(y)}$, labeled by
$\alpha\in\mathfrak{C}_{x(y)}^{\tau_{y(x)}}$. \\

$N(\mathtt{R})$
& Maximum number of mutually orthogonal extreme points of
$\Sigma(\mathtt{R})$; equal to the ground state degeneracy. \\

$\mathcal{G}(\mathtt{R})$; $\mathcal{P}_U$
& Group of the prescribed anyon-transport symmetry unitaries;
the fixed charge-labeled process realizing $U$. \\

$F_{\mathrm{total}}$
& Symmetry classes under equality of conjugation actions on every
$\rho\in\Sigma(\mathtt{R})$. \\

$\pi_{x(y)}$; $F_{x(y)}$
& Induced permutation action on $\mathfrak{C}_{x(y)}^{\tau_{y(x)}}$;
its image $F_{x(y)}=\operatorname{Im}\pi_{x(y)}$. \\

$\gamma_{\mathrm{LW},x(y)}$;
$\gamma_{\mathrm{LW},x(y)}^{\max}$
& TEE $\tfrac12 I(A:C\mid B)_{\rho_X}$ for the specified
partition; its maximum over $\Sigma(X)$. \\


$g\cdot\rho_X$
& Induced sector-permutation action
$\sum_\alpha p_\alpha\rho_X^{g(\alpha)}$, for $g\in F_{x(y)}$. \\

$\rho_X^{\mathrm{sym}}$
& Symmetrized state:
$|F_{x(y)}|^{-1}\sum_{g\in F_{x(y)}}g\cdot\rho_X$. \\

$\Delta S_{X_{x(y)}}[\rho]$; $\Delta S_{X_{x(y)}}^{\max}$
& Entanglement asymmetry $S(\rho_X^{\mathrm{sym}})-S(\rho_X)$; maximum entanglement asymmetry. \\

\bottomrule
\end{tabularx}
\end{table*}

\vspace{5pt}\noindent {\bf 2. Symmetries in the presence of transparent domain walls.}

At fixed points of Abelian topological phases without nontrivial transparent domain walls, operators implementing the closed-loop transport of Abelian anyons around noncontractible loops realize conventional 1-form symmetries. These are the symmetries that underlie the spontaneous 1-form symmetry breaking perspective on topological phases.

In the presence of transparent domain walls, however, we establish that there could exist symmetries supported jointly on two fundamental cycles that cannot be decomposed into a product of two $1$-form symmetries supported individually on these cycles [Theorem~\ref{th:uni} and Proposition~\ref{pro:symd}]. The support of these symmetries along both fundamental cycles arises from anyon transmutation induced by transparent domain walls. An anyon transported around a noncontractible loop may return as a different type, leaving a residual topological charge. Annihilating this charge may then require a second transport process along a topologically independent noncontractible loop. 

This modifies the symmetry structure underlying the spontaneous symmetry breaking description of topological order, making the relation between symmetries and ground state degeneracy less direct. We clarify this relation by connecting the symmetry action on MESs associated with a noncontractible annulus to anyon tunneling through the annulus. This tunneling process is, in turn, described by the fusion rules governing the action of anyons on the extreme points of the ICS on the noncontractible annulus.

\begin{figure}[!t]
\centering
\begin{overpic}[width=0.75\linewidth]{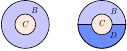}
\put(-7,40){(a)}
\put(54,40){(b)}
\end{overpic}      
\caption{Configurations used in axioms (a) ${\bf A0}$ and (b) ${\bf A1}$.}
\label{fig:axioms}
\end{figure}

\section{Information Convex Sets}
\label{sec:3}

This section first reviews the properties of the ICSs on local annuli for pure states satisfying the two axioms introduced below~\cite{shi_fusion_2020,huang_knots_2023,shi_immersed_2024}. Such states, referred to as \textit{reference states}, represent fixed points of $(2+1)$-dimensional many-body topological phases. We then discuss transparent domain walls, which are invisible to these axioms. Moreover, we introduce ICSs on noncontractible annuli, define fusion rules between the extreme points of ICSs on local and noncontractible annuli, and derive several useful properties. We emphasize that the results presented in this section apply to fixed points of both Abelian and non-Abelian topological phases.

In this work, the reference state is defined on a region $\mathtt{R}$ that is topologically equivalent to a torus. In the figures, the white square background represents $\mathtt{R}$. On $\mathtt{R}$, we primarily consider two {\it topologically independent} oriented noncontractible loops, denoted by $l_x$ and $l_y$, as illustrated in Fig.~\ref{fig:two-disks}. These loops wind once around the $x$ and $y$ cycles, respectively, and correspond to the homology classes $(1,0)$ and $(0,1)$. More generally, every topologically distinct oriented noncontractible loop can be represented by $l_{(n,m)}$, where $n,m \in \mathbb{Z}$ and $\gcd(n,m)=1$. We use the same notation for all noncontractible loops belonging to the same homology class, since the results considered here depend only on the homology class and not on the choice of representative loop.

To state the axioms, we first define the set of reduced density matrices
\begin{equation}
\mu \equiv \{\rho_b \mid b \in \mathtt{B}(r)\},
\end{equation}
where $\mathtt{B}(r)$ is the set of balls of radius no greater than $r$, with $r$ taken to be a constant independent of system size. Here, $\rho_b$ denotes the reduced density matrix of the reference state $\rho$ on $b$. We refer to the regions $b \in \mathtt{B}(r)$ as $\mu$-disks.

Any reference state must satisfy the two axioms ${\bf A0}$ and ${\bf A1}$ on every $\rho_b \in \mu$~\cite{shi_fusion_2020,shi_immersed_2024}:
\begin{align}
{\bf A0}:&\quad \Delta(B,C,\emptyset)_{\rho_b} = 0, \nonumber\\
{\bf A1}:&\quad \Delta(B,C,D)_{\rho_b} = 0,
\label{eq:axioms}
\end{align}
with
\begin{equation}
\Delta(B,C,D)_\rho \equiv S(\rho_{BC}) + S(\rho_{CD}) - S(\rho_{B}) - S(\rho_{D}),
\end{equation} and where $\emptyset$ is the empty set.
Here, $S(\rho)\equiv -\Tr(\rho \log \rho)$ denotes the von Neumann entropy of $\rho,$ and the partitions for ${\bf A0}$ and ${\bf A1}$ are illustrated in Fig.~\ref{fig:axioms}(a) and Fig.~\ref{fig:axioms}(b), respectively. We now turn to the properties of these states as encoded in the ICSs.

\subsection{ICSs on local annuli}
\label{sec:3a}

For a subregion $X \subseteq \mathtt{R}$, the ICS $\Sigma(X)$ consists of reduced density matrices on $X$ that are locally indistinguishable from the reference state and satisfy certain information-theoretic constraints that remove boundary effects~\cite{shi_fusion_2020}. In the following, we restrict our attention to the properties of the ICS relevant to the present discussion. The formal definition can be found in Ref.~\onlinecite{shi_fusion_2020}.

Let $X$ be either a local or noncontractible annular subregion, and let $\Sigma(X)$ denote the associated ICS. The set $\Sigma(X)$ satisfies the simplex theorem, whose proof is given in Ref.~\onlinecite{shi_entanglement_2021}.

\begin{theorem}[Simplex theorem]\label{th:simplex}
For a local or noncontractible annular subregion $X$, $\Sigma(X)$ is the convex hull of a finite set of mutually orthogonal extreme points, denoted by ${\rm ext}\Sigma(X)$. Explicitly,
\begin{equation}
\Sigma(X)= \left\{
\rho_X \,\middle|\, 
\rho_X = \sum_{a} p_a  \sigma_X^{a},\ \sigma_X^{a} \in {\rm ext}\Sigma(X)
\right\},
\label{eq:info-convex-annulus}
\end{equation}
where $\{a\}$ is a finite label set, and $\{p_a\}$ is a probability distribution.
\end{theorem}

\noindent In this subsection, we focus on local annuli; noncontractible annuli will be considered in Sec.~\ref{sec:3c}. From the ICS on a local annulus, we define the set of anyons (superselection sectors).

\begin{definition}[Anyon types]
The set of anyons $\mathcal{C}$ is defined as the set of labels of the extreme points of $\Sigma(X)$ when $X$ is a local annulus.
\end{definition}
\noindent Note that $\mathcal{C}$ always contains the trivial anyon $1$, corresponding to the extreme point of the ICS on a local annulus obtained by partially tracing the reference state~\cite{shi_fusion_2020, huang_knots_2023}. We use lowercase letters to denote the anyons in $\mathcal{C}$.

After defining the anyons, we can further extract the fusion rules from two-hole disk subregions; a detailed discussion can be found in Refs.~\onlinecite{shi_fusion_2020, huang_knots_2023}. The fusion rules take the form
\begin{align}
a \times b = \sum_{c \in \mathcal{C}} N_{ab}^c c, \label{eq:fusion}
\end{align}
and we denote by $\bar{a}$ the unique antiparticle of $a$, defined by the condition $N_{a\bar{a}}^{1}=1$. We may also introduce the corresponding quantum dimensions~\cite{shi_immersed_2024}.

\begin{definition}[Quantum dimension $d_a$]\label{def:qd}
The quantum dimension $d_a$ associated with $\rho_X^{a} \in {\rm ext}\Sigma(X)$ can be expressed as
\begin{equation*}
d_a \equiv \exp\left( \frac{1}{4} \Delta(A,B,C)_{\rho^{a}_X} \right),\label{eq:QD}
\end{equation*}
where $X = A B C$ as illustrated in Fig.~\ref{fig:QD}, with $A C = \partial X$ denoting the thickened boundary and $B = X \setminus \partial X$ the interior. Each of $A$ and $C$ is a union of two disks.
\end{definition}

\begin{figure}[ht]
\centering
\includegraphics[width=0.40\linewidth]{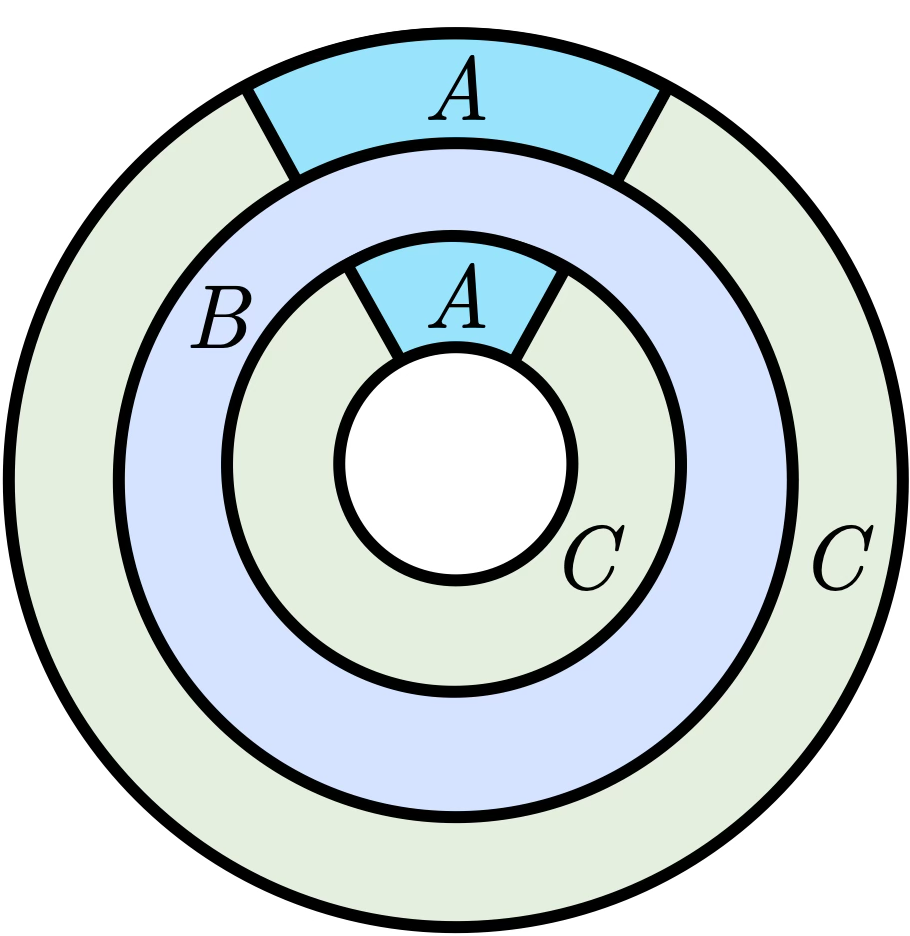}
\caption{Partition of $X = ABC$ used in Definition~\ref{def:qd}.}
\label{fig:QD}
\end{figure}

We then obtain the following proposition and theorem, whose proofs can be found in Refs.~\onlinecite{shi_immersed_2024,shi_fusion_2020}.

\begin{proposition}
\begin{align}
2\log (d_a/d_b)=S(\rho_X^a)-S(\rho_X^b),
\end{align}
where $a, b \in \mathcal{C},$ and $X$ is a local annulus.
\end{proposition}

\begin{theorem}\label{th:qda}

The quantum dimensions $d_a$ associated with $a \in \mathcal{C}$ satisfy the following set of equations:
\begin{equation}
d_a d_b = \sum_{c \in \mathcal{C}} N^{c}_{ab} d_c ,
\end{equation}
where $N^{c}_{ab}$ is obtained in \eqref{eq:fusion}.     
\end{theorem}
\noindent We will see how the quantum dimension can be defined in a similar way for noncontractible annuli and what this implies for the TEE. We further define the total quantum dimension $\mathcal{D}$.
\begin{definition}[Total quantum dimension]
The total quantum dimension $\mathcal{D}$ is defined as
\end{definition}
\begin{align*}
\mathcal{D} = \sqrt{\sum_{a \in \mathcal{C}} d_a^2}.
\end{align*}

\subsection{Transparent domain walls}

We can transport local annuli along a ``path,'' which, roughly speaking, specifies a way of moving a subregion while preserving its topology. A rigorous definition of the path is given in Appendix~\ref{asec:1} and Ref.~\onlinecite{shi_fusion_2020}. According to the isomorphism theorem~\cite{shi_fusion_2020}, the ICS of the transported local annulus is isomorphic to that of the original local annulus, with the isomorphism determined by the path of transport. Importantly, the isomorphism theorem applies only to states satisfying the axioms ${\bf A0}$ and ${\bf A1}$, namely, reference states.

In general, transporting a local annulus along two different paths that both lie entirely within a disk induces the same isomorphism of the ICS~\footnote{This statement holds under the definition of a path as a sequence of embedded regions related by elementary deformation steps that preserve the embedded-region condition, following Ref.~\onlinecite{shi_fusion_2020}. Whether it extends to immersed regions remains an open problem~\cite{shi_immersed_2024}.}. This then provides a natural way to label the extreme points of the local annulus after the transport (see the representation of the extreme points in Theorem~\ref{th:simplex}). In contrast, two paths not fully contained within \emph{the same disk} may induce different isomorphisms.

Similarly, transporting a local annulus along a path that lies entirely within a disk and returning it to its original position yields a trivial automorphism on the associated ICS~\cite{shi_fusion_2020}. However, a nontrivial automorphism may arise when the transport path winds around a noncontractible loop of the underlying manifold, which in the present case is a torus.

We then define automorphisms of $\mathcal{C}$ as follows.

\begin{definition}[Automorphism $\phi_{(n,m)}$]
The automorphisms of $\mathcal{C}$ associated with the noncontractible loop $l_{(n,m)}$ are defined as
\begin{align*}
\phi_{(n,m)}: \mathcal{C} \rightarrow \mathcal{C},
\end{align*}
where $n,m \in \mathbb{Z}$. Thus, $\phi_{(1,0)}$ and $\phi_{(0,1)}$ can be denoted as $\phi_x$ and $\phi_y$.
\end{definition}

\noindent Moreover, we have the following propositions.

\begin{proposition}\label{pro:trivial}
The trivial anyon $1$ is invariant under $\phi_{(n,m)}$, i.e., $\phi_{(n,m)}(1)=1$ for all $n,m \in \mathbb{Z}$.
\begin{proof}

Let us refer to the extreme point of a local annulus labeled by the trivial anyon $1$ as the vacuum state. Consider first transporting a local annulus within a larger disk, $X_0 \to X_1$. By Lemma~4.3 of Ref.~\onlinecite{shi_fusion_2020}, the vacuum state on $X_0$ is mapped to the vacuum state on $X_1$.

The same reasoning applies to transporting a local annulus on a torus, except that the full transport must be decomposed into a sequence of local steps supported in overlapping disks. Consider a path $X_0 \to X_1 \to X_2 \to \cdots \to X_N=X_0$ representing the homology class $(n,m)$. Applying the above local argument at each step shows that the vacuum state is mapped back to itself after completing the path. Therefore, the trivial anyon is invariant under $\phi_{(n,m)}$, namely, $\phi_{(n,m)}(1)=1 .$
\end{proof}
\end{proposition}

\begin{proposition}\label{pro:phi}
The maps $\phi_{(n,m)}$ are elements of the permutation group $S_k$, with $k=|\mathcal{C}|-1$, where the permutations act on the elements of $\mathcal{C}$ excluding the trivial anyon $1$. Furthermore,
\begin{align} \label{eq: phi commute}
\phi_{(n,m)}&=\phi_x^n \circ \phi_y^m= \phi_y^m \circ \phi_x^n.
\end{align}

\begin{proof}

Since $\phi_{(n,m)}$ is an automorphism of the ICS of a local annulus, it permutes its extreme points and therefore defines an element of $S_{|\mathcal{C}|}$. By
Proposition~\ref{pro:trivial}, the trivial anyon is fixed by every $\phi_{(n,m)}$. Hence, $\phi_{(n,m)}$ acts as an element of $S_k$, where $k=|\mathcal{C}|-1$.

The path corresponding to $\phi_{(n,m)}$ is homotopic to the concatenation of $n$ $(1,0)$ transports and $m$ $(0,1)$ transports. Therefore, it remains to prove that
\begin{equation}
\label{eq:xy}
\phi_y \circ \phi_x = \phi_x \circ \phi_y .
\end{equation}
The general claim then follows by applying \eqref{eq:xy} a finite number of times. 

Choose an initial local annulus $\widetilde{X}_0$. The transport corresponding to $\phi_y\circ\phi_x$ is a path of local annuli from $\widetilde{X}_0$ to its translate by $(L_x,L_y)$, obtained by first translating in the $x$ direction and then in the $y$ direction. Similarly, the transport corresponding to $\phi_x\circ\phi_y$ is a path with the same initial and final annuli, but with the two translations performed in the opposite order.
The two paths are contained within a common immersed disk $\mathfrak{D}$ of size $(L_x+d)\times(L_y+d)$, where $d$ is the diameter of the local annulus, and share the same initial and final annuli, where the immersion is defined in Ref.~\onlinecite{huang_knots_2023}. Importantly, every intermediate annulus along both paths remains embedded in $\mathfrak{D}$, and each consecutive pair is related by an elementary deformation supported within an embedded disk. Thus, no immersed annulus is required during the transport process. Applying path independence successively to these elementary deformations shows that the two paths induce the same isomorphism between the ICSs of their common initial and final annuli; this is by Lemma~4.3 of Ref.~\onlinecite{shi_fusion_2020}. Therefore, this proves \eqref{eq:xy}.
\end{proof}
\end{proposition}

\noindent In other words, all distinct automorphisms can be described in terms of commuting automorphisms $\phi_x$ and $\phi_y$. To provide a more intuitive interpretation of these automorphisms and Proposition~\ref{pro:phi}, we introduce the notion of transparent domain walls for $\mathcal{C}$ along $l_{x(y)}$. This picture provides an interpretation in which an automorphism arises when an element of $\mathcal{C}$ is transported across the wall, while also reflecting the fact that $\phi_x$ and $\phi_y$ generate all distinct automorphisms.

\begin{definition}[Transparent domain walls for $\mathcal{C}$]
A transparent domain wall for $\mathcal{C}$ along $l_{x(y)}$ is a wall in $\mathtt{R}$, topologically equivalent to $l_{x(y)}$, such that transporting an annulus across the wall induces a nontrivial automorphism $\phi_{y(x)}$ on the elements of $\mathcal{C}$.
\end{definition}

\noindent Here, ``transparent'' refers to the fact that the presence of the transparent domain wall is invisible to the axioms ${\bf A0}$ and ${\bf A1}$~\cite{shi_entanglement_2021, buican_algebraic_2025}. If multiple such domain walls are present, we regard them as being moved to the same position. It is therefore sufficient to work with a single transparent domain wall for $\mathcal{C}$ along $l_{x(y)}$, and to identify its induced automorphism on the labels with that obtained after transporting the annulus back to its original position. We adopt this convention throughout the paper.

The commutativity of the automorphisms $\phi_x$ and $\phi_y$ established in Proposition \ref{pro:phi} can also be understood in the domain-wall picture. General defect theories allow branch lines corresponding to domain walls labeled by noncommuting symmetry elements. When a defect labeled by $h$ is transported across a $g$-branch line, its group label changes to $ghg^{-1}$, and its attached branch line changes accordingly \cite{PhysRevB.100.115147}. However, closing branch lines around both fundamental cycles of a torus imposes an additional consistency condition. In the construction of Ref.~\onlinecite{PhysRevB.100.115147}, the two branch lines can be closed without leaving point-defect endpoints only when $gh=hg$. At the level of their induced permutations of anyon labels, this is consistent with the relation $\phi_x\circ\phi_y=\phi_y\circ\phi_x$ in \eqref{eq:xy}. Thus, Proposition \ref{pro:phi} restricts which wall actions can be simultaneously realized as the two cycle automorphisms of a reference state satisfying \textbf{A0} and \textbf{A1}.

Several remarks concerning transparent domain walls for $\mathcal{C}$ are in order. First, transparent domain walls for $\mathcal{C}$ are topologically deformable. This property stems from the fact that the set $\mathcal{C}$ is not physically fixed, but is instead conventional. For example, for the $e$-$m$ exchange transparent domain wall of the toric code phase, which we discuss in more detail in Sec.~\ref{sec:5}, the corresponding noncontractible loop cannot be deformed away from both disks shown in Fig.~\ref{fig:two-disks}. It can, however, be chosen to avoid any one disk individually.

\begin{figure}
\centering
\includegraphics[width=0.6\linewidth]{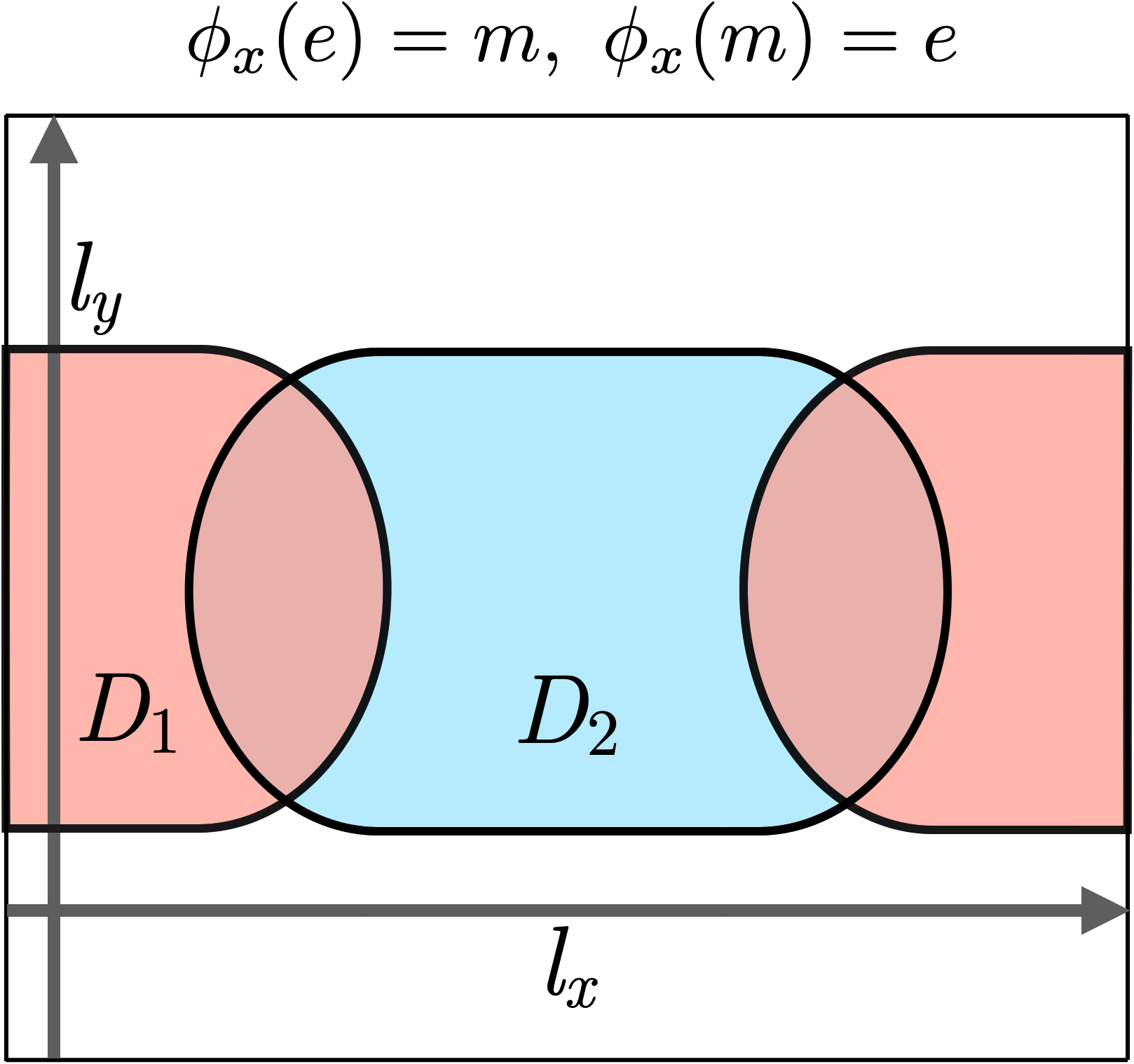}
\caption{A transparent domain wall along $l_y$ permutes the anyons as $e \leftrightarrow m$ when they are transported along $l_x$, i.e., $\phi_x(e)=m$ and $\phi_x(m)=e$. Whether the transparent domain wall is regarded as being away from the disk $D_1$ or $D_2$ is a matter of convention; however, it cannot be deformed away from both disks simultaneously. The two topologically independent noncontractible loops on $\mathtt{R}$, denoted by $l_x$ and $l_y$, respectively, are depicted. }
\label{fig:two-disks}
\end{figure}

Second, a transparent domain wall for $\mathcal{C}$ acts solely by inducing a nontrivial automorphism on the ICSs defined on local annuli, while leaving the ICSs themselves invariant. Consequently, the quantum information extractable from ICSs on local annuli remains unchanged in the presence of a transparent domain wall.

Third, these automorphisms must preserve the fusion rules. This can be shown using the following theorem.

\begin{theorem}
If
\begin{equation}
a \times b = \sum_{c \in \mathcal{C}} N_{ab}^c c
\end{equation}
is a fusion rule, then
\begin{align}
\phi_{(n,m)}(a) \times \phi_{(n,m)}(b) =\sum_{c \in \mathcal{C}}
N_{ab}^c \phi_{(n,m)}(c)\label{eq:phiabc}
\end{align}
must also hold. Equivalently, the fusion coefficients are invariant under $\phi_{(n,m)}$, i.e.,
\begin{equation}
N_{ab}^c=N_{\phi_{(n,m)}(a),\phi_{(n,m)}(b)}^{\phi_{(n,m)}(c)} .\label{eq:NNphi}
\end{equation}

\begin{proof}
By the definition of the fusion rules in Ref.~\onlinecite{shi_fusion_2020}, there are $N_{ab}^{c}$ mutually orthogonal extreme points of the ICS of a two-hole disk for which the local annuli encircling the two holes are labeled by $a$ and $b$, while the local annulus encircling the entire two-hole disk is labeled by $c$.

Now consider these mutually orthogonal extreme points of the ICS of the two-hole disk, and transport the two-hole disk along a path winding around the noncontractible loop $l_{(n,m)}$ before returning it to its original region. Under the induced automorphism, the labels $a$, $b$, and $c$ are mapped to $\phi_{(n,m)}(a)$, $\phi_{(n,m)}(b)$, and $\phi_{(n,m)}(c)$, respectively. Therefore, the transported states are extreme points of the ICS of a two-hole disk specified by the labels $\phi_{(n,m)}(a)$, $\phi_{(n,m)}(b)$, and $\phi_{(n,m)}(c)$. This gives
\begin{equation}
N_{ab}^{c}
\leq
N_{\phi_{(n,m)}(a)\phi_{(n,m)}(b)}^{\phi_{(n,m)}(c)}.
\end{equation}
Applying the same argument to the inverse transport gives the opposite inequality,
\begin{equation}
N_{ab}^{c}
\geq
N_{\phi_{(n,m)}(a)\phi_{(n,m)}(b)}^{\phi_{(n,m)}(c)}.
\end{equation}
Hence \eqref{eq:NNphi} follows, which directly implies \eqref{eq:phiabc}.

\end{proof}
\end{theorem}
 
\noindent We also note that, in determining the fusion rules of $\mathcal{C}$ using the two-hole disk~\cite{shi_fusion_2020}, we implicitly assume that transparent domain walls are located away from the two-hole disk, since topological deformations of the two-hole disk—restricted to a sufficiently large disk containing it—do not change the labels of the annular boundaries.

In addition, we define the following subsets of $\mathcal{C}$; these subsets will be useful for the subsequent analysis.
\begin{definition}
$\mathcal{C}^{\phi_{(n,m)}}$ is defined as the subset of $\mathcal{C}$ given by
\begin{align*}
\mathcal{C}^{\phi_{(n,m)}}=\{ a \in \mathcal{C} \mid \phi_{(n,m)}(a)=a \},
\end{align*}
where $n,m \in \mathbb{Z}$.
\end{definition}

\subsection{ICSs on noncontractible annuli}
\label{sec:3c}

We now consider ICSs associated with noncontractible annuli. As briefly discussed in Sec.~\ref{sec:2}, such ICSs can encode various types of information associated with transparent domain walls. In this subsection, we make this connection explicit.

Let $\mathtt{A}_x$ and $\mathtt{A}_y$ denote the following sets of noncontractible annuli.

\begin{definition}
$\mathtt{A}_{x(y)}$ is defined as the set of noncontractible annuli $X \subset \mathtt{R}$ that contain noncontractible loops $l_{x(y)}$.
\end{definition}
\noindent Even for noncontractible annuli, Theorem~\ref{th:simplex} implies that the ICS consists of finitely many mutually orthogonal extreme points and therefore forms a simplex.

\begin{definition}
Let $X \in \mathtt{A}_{x(y)}$. 
$\mathfrak{C}_{x(y)}$ is defined as the set of labels of the extreme points of $\Sigma(X)$.
\end{definition}
\noindent We use Greek letters to denote the elements in $\mathfrak{C}_{x(y)}$.

In what follows, we extract and analyze various properties associated with $\mathfrak{C}_{x(y)}$. We begin by determining, following Ref.~\onlinecite{shi_immersed_2024}, the number of extreme points of the ICSs defined on noncontractible annuli. This counting directly reveals how the structure of the ICS reflects the presence of transparent domain walls intersecting the annuli.

\begin{theorem}\label{th:cc}
For $X \in \mathtt{A}_{x(y)}$, the corresponding set $\mathfrak{C}_{x(y)}$ has cardinality
\begin{align}
|\mathfrak{C}_{x(y)}|=|\mathcal{C}^{\phi_{x(y)}}|.
\end{align}

\begin{proof}
We note that the noncontractible annulus $X$ is the thickening of the anyon-transport loop, which induces the automorphism $\phi_{x(y)}$ of the set of anyons $\mathcal{C}$. Thus, the claim in the theorem can be rephrased as stating that the number of extreme points of the ICS on the noncontractible annulus thickening the transport path equals the number of anyons preserved by the transport. The theorem is therefore a special case of Proposition~12 of Ref.~\onlinecite{shi_immersed_2024}. This completes the proof.
\end{proof}
\end{theorem}

Next, noncontractible annuli also satisfy the isomorphism theorem~\cite{shi_fusion_2020}. In this setting, one can show that transporting a noncontractible annulus to another noncontractible annulus along two different paths induces the same isomorphism, provided that both paths, as well as both annuli, are fully contained in a larger noncontractible annulus. In contrast, if the two paths are not contained in a common noncontractible annulus, they may induce distinct isomorphisms. The proof is given in Appendix~\ref{asec:1}. This further implies that a nontrivial automorphism of $\mathfrak{C}_{x(y)}$ can arise from transport along only the noncontractible loop $l_{y(x)}$, and not along $l_{x(y)}$.

We now define the automorphisms of $\mathfrak{C}_{x(y)}$ and the corresponding transparent domain walls.
\begin{definition}[Automorphism $\tau_{y(x)}$]
The automorphisms of $\mathfrak{C}_{x(y)}$ associated with the noncontractible loop $l_{y(x)}$ are defined as $\tau_{y(x)}: \mathfrak{C}_{x(y)} \rightarrow \mathfrak{C}_{x(y)}$.
\end{definition}

\begin{definition}[Transparent domain walls for $\mathfrak{C}_{x(y)}$]
A transparent domain wall for $\mathfrak{C}_{x(y)}$ along $l_{x(y)}$ is a wall in $\mathtt{R}$, topologically equivalent to $l_{x(y)}$, such that transporting an annulus across the wall induces a nontrivial automorphism $\tau_{y(x)}$ on the elements of $\mathfrak{C}_{x(y)}$.
\end{definition}

\noindent These transparent domain walls are likewise topologically deformable, and the deformability leads to several propositions concerning the fusion rules of $\mathcal{C}$ and $\mathfrak{C}_{x(y)}$, as explained below.

\begin{figure*}
\centering
\begin{overpic}[width=0.9\linewidth]{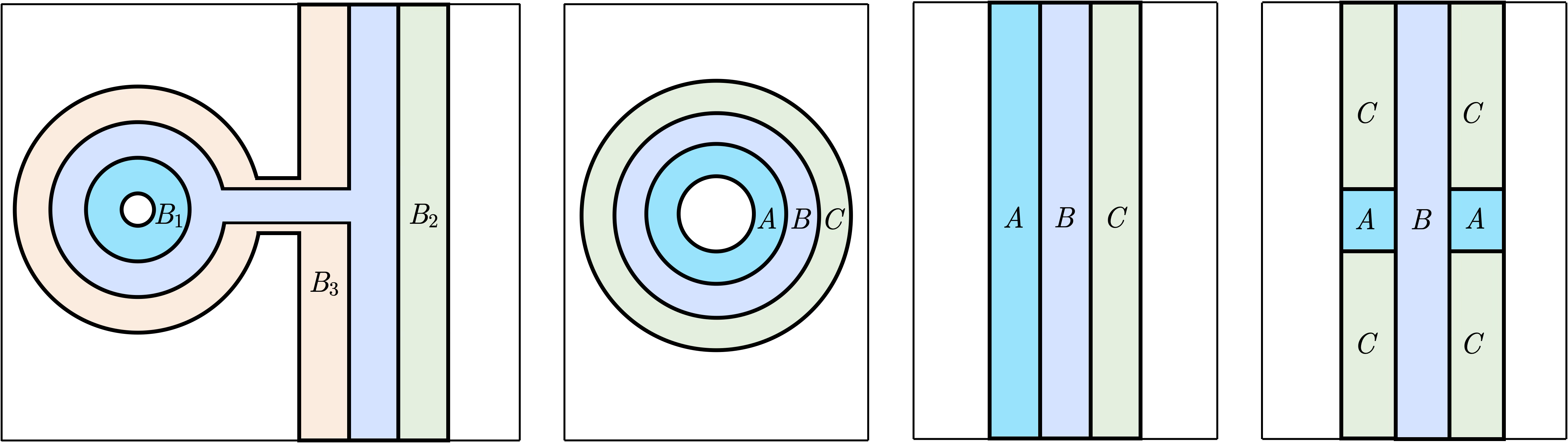}
\put(-2.2,29){(a)}
\put(33.5,29){(b)}
\put(56,29){(c)}
\put(78,29){(d)}
\end{overpic}
\caption{(a) Partitions of the punctured noncontractible annulus $Y$. Partition of $Z=ABC$ for (b) the local annulus case and (c) the noncontractible annulus case, used in Theorem~\ref{th:ext}. (d) Partition of $X = A BC$ used in Definition~\ref{def:qdalpha}.}
\label{fig:fusion}
\end{figure*}
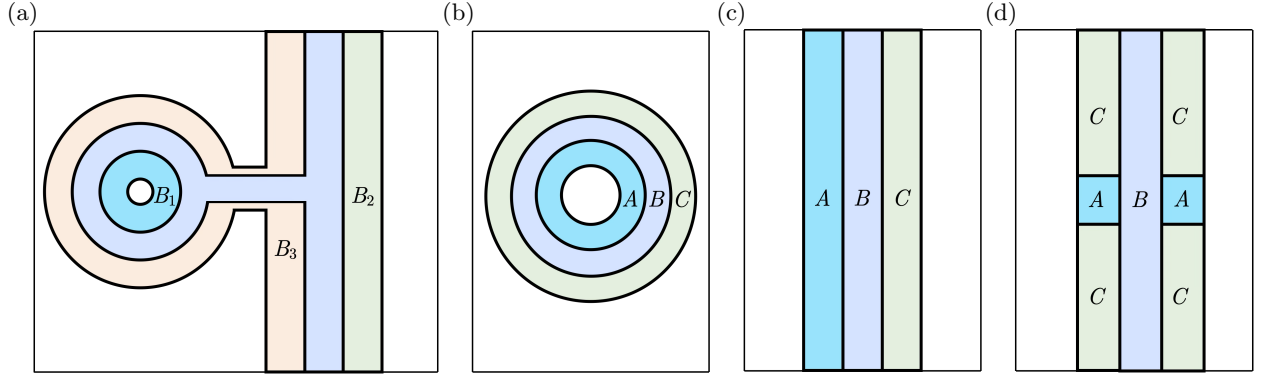

We now define the fusion rules of $\mathcal{C}$ and $\mathfrak{C}_{x(y)}$. For simplicity, we mainly focus on $\mathfrak{C}_y$ throughout this section. The discussion applies straightforwardly to $\mathfrak{C}_x$ as well. To define the fusion rules between $\mathcal{C}$ and $\mathfrak{C}_{y}$, we consider the punctured noncontractible annulus $Y$ shown in Fig.~\ref{fig:fusion}(a). We also assume that a transparent domain wall along $l_y$ is located sufficiently far from $Y$. This ensures that topological deformations of $Y$, restricted within a sufficiently large noncontractible annulus containing $Y$, do not change the labels of extreme points of the ICS on the local or noncontractible annuli, thereby maintaining the consistency of the fusion rules.

The subregion $Y$ has three thickened boundaries, $B_1$, $B_2$, and $B_3$, which are, respectively, a local annulus and two noncontractible annuli. Accordingly, the extreme points of the ICS on $B_1$ are labeled by elements of $\mathcal{C}$, while those on $B_2$ and $B_3$ are labeled by elements of $\mathfrak{C}_y$.

For the region shown in Fig.~\ref{fig:fusion}(a), we define $\Sigma_{a\alpha}^{\beta}(Y)$ as the convex subset of $\Sigma(Y)$ consisting of states whose reductions to the three boundary annuli have fixed extreme-point labels \(a\) on \(B_1\), \(\alpha\) on \(B_2\), and \(\beta\) on \(B_3\):

\begin{definition}
For the subregion $Y$ shown in Fig.~\ref{fig:fusion}(a), we define $\Sigma_{a\alpha}^{\beta}(Y)$ as
\begin{align*}
\Sigma_{a\alpha}^{\beta}(Y)\equiv&\{\rho_{Y}\in \Sigma(Y)|{\rm Tr}_{Y \setminus B_1 } \rho_Y=\sigma_{B_1}^{a},\nn
&{\rm Tr}_{Y \setminus B_2 } \rho_Y=\sigma_{B_2}^{\alpha}, {\rm Tr}_{Y \setminus B_3 } \rho_Y=\sigma_{B_3}^{\beta}  \},
\end{align*}
where $\sigma_{B_1}^{a} \in {\rm ext}\Sigma(B_1)$, $\sigma_{B_2}^{\alpha} \in {\rm ext}\Sigma(B_2)$, $\sigma_{B_3}^{\beta} \in {\rm ext}\Sigma(B_3)$, $a \in \mathcal{C}$, and $\alpha,\beta \in \mathfrak{C}_y$.
\end{definition}
\noindent In addition, using these convex subsets, we can express the corresponding ICSs as follows.

\begin{theorem}\label{th:ext}
For the subregion $Y$ shown in Fig.~\ref{fig:fusion}, the ICS $\Sigma(Y)$ takes the form
\begin{align}\label{eq:convex-1}
\Sigma(Y) =\left\{ \rho_{Y} = \hspace{-10pt} \sum_{a \in \mathcal{C}, ~ \alpha, \beta \in \mathfrak{C}_y} \hspace{-15pt} p^{\beta}_{a\alpha}  \rho^{a\alpha \beta}_{Y}
\;\middle|\;
\rho^{a\alpha \beta}_{Y} \in \Sigma^{\beta}_{a\alpha}(Y)
\right\},
\end{align}
where $\{p_{a \alpha}^{\beta}\}$ is a probability distribution.
\begin{proof}
First, we need a property of any extreme point of $\Sigma(X)$, known as the factorization property:
\begin{equation}
    (S_X + S_{X'} - S_{X'\setminus X})_\rho =0, \quad {\rm for} \quad \rho\in {\rm ext}\Sigma(X),
\end{equation}
where $X'$ is the thickening of $X$. (The proof follows from Appendix C of Ref.~\onlinecite{shi_entanglement_2021}.)

This statement has two important consequences. First, for any extreme point of the ICS associated with a given region, the reduced density matrix factorizes over the distinct connected components of the thickened boundary. Second, let $Z$ be any connected component of the thickened boundary. In the present setting, $Z$ is an annulus thickening an entanglement boundary of $Y$. The onion-like decomposition $Z=ABC$, illustrated for the local and noncontractible annuli in Figs.~\ref{fig:fusion}(b) and~\ref{fig:fusion}(c), respectively, satisfies $I(A:C)=0$ in both cases (see Proposition D.6 of Ref.~\onlinecite{shi_fusion_2020}). It follows that the reduced density matrix on $Z$ must itself be an extreme point. Therefore, every extreme point of $\Sigma(Y)$ belongs to some $\Sigma_{a\alpha}^{\beta}$.

We next use the fact that $\Sigma(Y)$ is a compact convex set (see Proposition~3.4 of Ref.~\onlinecite{shi_fusion_2020}). Hence, any state in $\Sigma(Y)$ can be written as a convex combination of extreme points, which establishes \eqref{eq:convex-1}.
\end{proof}
\end{theorem}

We now show that the following convex subsets are isomorphic to state spaces of finite-dimensional Hilbert spaces, i.e., to the sets of all density matrices on the corresponding Hilbert spaces. This fact is essential for a precise definition of the fusion rules.

\begin{proposition}\label{pro:statespace}
\begin{align}
&\Sigma_{a\alpha}^{\beta}(Y)\cong \mathcal{S}(\mathbb{V}_{a\alpha}^\beta)
\end{align}
where $A \cong B$ indicates that $A$ is isomorphic to $B$, and $\mathcal{S}(\mathbb{V}_{a\alpha}^\beta)$ is the state space of the Hilbert space $\mathbb{V}_{a\alpha}^\beta$.
\begin{proof}
The proof follows from the Hilbert space theorem, which is Theorem~D.1 of Ref.~\onlinecite{shi_entanglement_2021}.
\end{proof}
    
\end{proposition}

\begin{definition}[Fusion rules]\label{def:fusion}
We define the fusion rules of labels $a \in {\cal C}$ and $\alpha, \beta \in \mathfrak{C}_y$ by the formal product
\begin{align*}
a \times \alpha =\sum_{\beta \in \mathfrak{C}_y} N_{a\alpha}^{\beta} \beta
\end{align*}
where $N^{\beta}_{a \alpha} \equiv \dim \mathbb{V}^{\beta}_{a \alpha}$~\footnote{We use the same symbol $N$ for distinct fusion coefficients; the different index structures distinguish them unambiguously, so we keep this notation.}.
\end{definition}

\noindent Now we can define quantum dimensions and the corresponding property.
\begin{definition}[Quantum dimension $d_\alpha$]\label{def:qdalpha}
The quantum dimension $d_\alpha$ associated with $\rho_X^{\alpha} \in {\rm ext}\Sigma(X)$ and $X \in \mathtt{A}_y$ is defined as
\begin{equation*}
d_\alpha \equiv \exp\left( \frac{1}{4} \Delta(A,B,C)_{\rho^{\alpha}_X} \right),
\end{equation*}
where $X = A B C$ as illustrated in Fig.~\ref{fig:fusion}(d), with $A C = \partial X$ denoting the thickened boundary and $B = X \setminus \partial X$ the interior. Each of $A$ and $C$ is a union of two disks.    
\end{definition}
\noindent One readily sees that $d_\alpha$ is the analogue of $d_a$, as is evident from comparing Definition~\ref{def:qdalpha} with Definition~\ref{def:qd}. Therefore, we obtain the following proposition.

\begin{proposition}\label{pro:alphabeta}
\begin{align}
2\log (d_\alpha/d_\beta)=S(\rho_X^\alpha)-S(\rho_X^\beta),
\end{align}
where $\alpha, \beta \in \mathfrak{C}_y$ and $X\in \mathtt{A}_y$.

\begin{proof}
From Definition~\ref{def:qdalpha}, we get 
\begin{align}
\Delta(A,B,C)_{\rho^{\alpha}}-\Delta(A,B,C)_{\rho^{\beta}}= 4\log(d_\alpha/d_\beta)
\end{align}
for the regions $A,B,C$ as in Fig. \ref{fig:fusion}(d).
Then 
\begin{align}
&\left[ S(\rho^{\alpha}_{AB})-S(\rho^{\beta}_{AB})\right]+\left[S(\rho^{\alpha}_{BC})-S(\rho^{\beta}_{BC})\right]\nn
&=4\log(d_\alpha/d_\beta).\label{eq: aux}
\end{align}
Note that both regions $AB$ and $BC$ are related to each other by elementary steps of deformation (in the sense of Ref.~\onlinecite{shi_fusion_2020}). Then the isomorphism theorem in Ref.~\onlinecite{shi_fusion_2020} gives that the entropy difference between two elements in the ICS is preserved under the deformation of the subregion $AB$ into $BC$ (or, equivalently, $BC$ into $AB$). That is,
\begin{equation}
S(\rho^{\alpha}_{AB})-S(\rho^{\beta}_{AB})=S(\rho^{\alpha}_{BC})-S(\rho^{\beta}_{BC})=2 \log(d_\alpha/d_\beta),
\end{equation}
where \eqref{eq: aux} was used to obtain the last equality. Note that the regions $AB$ and $BC$ are also deformable to subregion $X \in \mathtt{A}_y$. Then, by the isomorphism theorem, the difference between the entanglement entropies for that subregion is also preserved, which means
\begin{equation}
S(\rho^\alpha_X)-S(\rho_X^{\beta})=2\log(d_\alpha/d_\beta).
\end{equation}
\end{proof}
\end{proposition}

Moreover, we obtain the following properties, which will be useful for the subsequent analysis.

\begin{lemma}\label{lemma:factorization}
Consider the punctured noncontractible annulus $Y$ in Fig.~\ref{fig:fusion}(a). The entropy difference of two extreme points $\rho^{a\alpha \beta}_{Y} \in \Sigma^{\beta}_{a\alpha}(Y)$ and $\rho_Y^{a'\alpha'\beta'} \in \Sigma^{\beta'}_{a'\alpha'}(Y)$ is
\begin{equation}\label{eq:extreme-pt-entropy}
    S(\rho_Y^{a \alpha\beta}) - S(\rho_Y^{a'\alpha'\beta'}) = \log \frac{d_a d_\alpha d_\beta}{d_{a'} d_{\alpha'} d_{\beta'}}.
\end{equation}
\end{lemma}
\begin{proof}
The statement and the proof are analogous to Lemma 4.8 of Ref.~\onlinecite{shi_fusion_2020}, though we have to prove this again for generality, including defects. The key property we shall need is the so-called factorization property of extreme points~\cite{shi_entanglement_2021}, which says
\begin{equation}
    S(\rho_{Y}) + S(\rho_{Y \setminus \partial Y}) = S(\rho_{\partial Y}), \quad \rho \in \text{ext}\Sigma(Y).
\end{equation}
Here $\partial Y$ is the thickened boundary of $Y$.
Applying this to two extreme points $\rho_Y^{a \alpha\beta}$ and $\rho_Y^{a'\alpha'\beta'}$, we have
\begin{align}\label{eq:details}
&S({\rho_Y^{a \alpha \beta}}) + S({\rho_{Y \setminus \partial Y}^{a \alpha \beta}})-S({\rho_Y^{a' \alpha' \beta'}}) - S({\rho_{Y \setminus \partial Y}^{a' \alpha' \beta'}})\nn
&= S({\rho_{\partial Y}^{a \alpha \beta}})-S({\rho_{\partial Y}^{a' \alpha' \beta'}})
\end{align}
One can deform $Y\setminus \partial Y$ to $Y$ smoothly, and thus the isomorphism theorem implies that the left-hand side of \eqref{eq:details} is twice $S({\rho_Y^{a \alpha \beta}})-S({\rho_Y^{a' \alpha' \beta'}})$. The right-hand side can also be simplified as the states on the three connected components factorize and give entropy-difference contributions as $2\log (d_a/d_{a'})$, $2\log (d_\alpha/d_{\alpha'})$ and $2\log (d_\beta/d_{\beta'})$ respectively. This leads to \eqref{eq:extreme-pt-entropy}.
\end{proof}

\begin{figure*}
\centering
\begin{overpic}[width=0.93\linewidth]{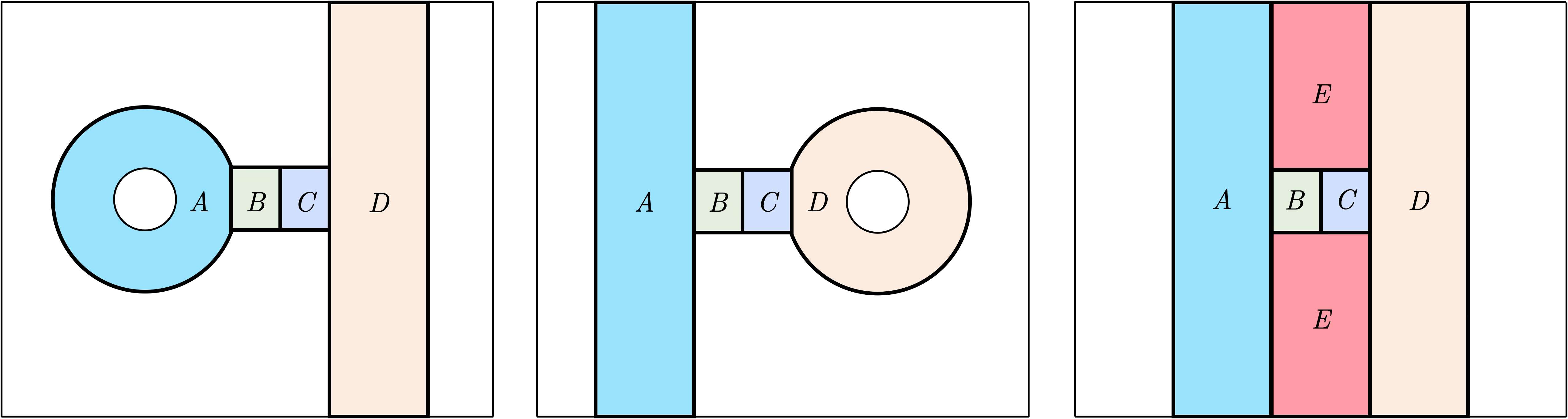}
\put(-2,27){(a)}
\put(32,27){(b)}
\put(66,27){(c)}
\end{overpic}
\caption{The subregion $ABCD$ used in the proof of Theorem~\ref{th:qdalpha}: (a) (i), (b) (ii), (c) (iii).}
\label{fig:qdalpha}
\end{figure*}

\begin{theorem}\label{th:qdalpha}
The quantum dimensions $d_a$ associated with $a \in \mathcal{C}$, together with $d_\alpha$ and $d_\beta$ associated with $\alpha,\beta \in \mathfrak{C}_y$, satisfy the following set of equations:
\begin{align}
{\rm (i)}~~~d_a d_\alpha &=\sum_{\beta \in \mathfrak{C}_y} N_{a\alpha}^{\beta} d_\beta,\nn
{\rm (ii)}~~~d_a d_\beta &=\sum_{\alpha \in \mathfrak{C}_y} N_{a \alpha}^{\beta} d_\alpha, \nn
{\rm (iii)}~~d_\alpha d_\beta &= \sum_{a \in {\cal C}} N_{a \alpha}^{\beta} d_a,
\end{align}
where $N_{a\alpha}^{\beta}$ is defined in Definition~\ref{def:fusion}, while $d_a$ and $d_\alpha$ are defined in Definitions~\ref{def:qd} and \ref{def:qdalpha}, respectively.
\begin{proof}
We establish these equalities one by one. The key nontrivial step is to compute the entropy difference in two different ways: first by decomposing the region into pieces, and second by using the fusion space of the full region. The consistency between the two computations leads to a ``bootstrap equation.'' In particular, we shall need the merging technique explained in Ref.~\onlinecite{shi_fusion_2020}. Such a method of deriving the bootstrap equation has become standard in the entanglement bootstrap approach, and has many applications~\cite{shi_entanglement_2021,huang_knots_2023,shi_seeing_2019}. The uniqueness of the merged state is guaranteed by axioms $A_0$ and $A_1$ in \eqref{eq:axioms}. We denote by $\rho_{ABCD}^{a\triangleright\!\triangleleft b}$ the unique state obtained by merging $\rho_{ABC}^a \in {\rm ext}\Sigma(ABC)$ and $\rho_{BCD}^b\in {\rm ext}\Sigma(BCD)$.

(i) Consider merging a local annulus $ABC$ with a noncontractible annulus
$BCD\in\mathtt{A}_y$, which yields the subregion $ABCD$ shown in
Fig.~\ref{fig:qdalpha}(a). We compare the entanglement entropies of the merged states $\rho_{ABCD}^{a\triangleright\!\triangleleft\alpha}$ and $\rho_{ABCD}^{1\triangleright\!\triangleleft\alpha}$, obtained by merging $\rho_{ABC}^a$ and $\rho_{BCD}^\alpha$, and $\rho_{ABC}^1$ and $\rho_{BCD}^\alpha$, respectively. Namely,
\begin{align}
\delta S
=
S(\rho_{ABCD}^{a\triangleright\!\triangleleft\alpha})
-
S(\rho_{ABCD}^{1\triangleright\!\triangleleft\alpha}),
\end{align}
where $a,1\in\mathcal C$ and $\alpha\in\mathfrak C_y$.

Using the entropy-difference preservation property of merging~\cite{shi_fusion_2020},
we obtain
\begin{align}
\delta S=S(\rho_{ABC}^{a})-S(\rho_{ABC}^{1})=2\log d_a .
\end{align}
Alternatively, the merged state is the maximum-entropy element in the convex hull
of the allowed fusion sectors. Then, by solving an entropy maximization problem to find the probabilities with which each extreme point appears, using Lemma~\ref{lemma:factorization}, we find
\begin{align}
\rho_{ABCD}^{a\triangleright\!\triangleleft\alpha}=\sum_{\beta \in \mathfrak C_y} \sum_{j=1}^{N_{a\alpha}^\beta} \frac{d_\beta}{\sum_{\beta'\in\mathfrak C_y}
N_{a\alpha}^{\beta'}d_{\beta'} }\rho_{ABCD}^{a \alpha \beta, j}.
\end{align}
Additionally, using Lemma~\ref{lemma:factorization}, we obtain
\begin{align}
\delta S=\log d_a+\log\left(
\sum_{\beta\in\mathfrak C_y}
N_{a\alpha}^{\beta}d_\beta
\right)-\log d_\alpha .
\end{align}
Comparing the two expressions for $\delta S$ gives
\begin{align}
d_a d_\alpha
=
\sum_{\beta\in\mathfrak C_y}
N_{a\alpha}^{\beta}d_\beta .
\end{align}

(ii) We now consider merging a noncontractible annulus $ABC\in\mathtt A_y$ with a local annulus $BCD$, yielding the subregion $ABCD$ shown in Fig.~\ref{fig:qdalpha}(b). The proof is essentially identical to that of case (i), and is therefore omitted.

(iii) Consider merging two noncontractible annuli $ABC, BCD \in \mathtt A_y$, which gives the subregion $ABCD$ shown in
Fig.~\ref{fig:qdalpha}(c). We compare the merged state $\rho^{\beta\triangleright\!\triangleleft\alpha}_{ABCD}$, obtained by merging $\rho_{ABC}^\beta$ and $\rho_{BCD}^\alpha$, with the state $\rho^\alpha_{ABCD}$ obtained by partially tracing out the disk $E$ from $ABCDE$ in the state $\rho_{ABCDE}^\alpha$:
\begin{align}
\delta S
=
S(\rho_{ABCD}^{\beta \triangleright\!\triangleleft \alpha})-S(\rho_{ABCD}^{\alpha}),
\end{align}
where $\alpha, \beta\in\mathfrak C_y$.

Using the fact that $I(A:D|BC)=0$ for $\rho_{ABCD}^{\beta \triangleright\!\triangleleft \alpha}$ and $I(A:D|BC)=2\log d_\alpha$ for $\rho_{ABCD}^{\alpha}$,
\begin{align}
&S(\rho_{ABCD}^{\beta \triangleright\!\triangleleft\alpha})-S(\rho_{ABCD}^{\alpha})\nn
&=S(\rho_{ABC}^\beta)+S(\rho_{BCD}^\alpha)-S(\rho_{BC}^1)\nn
&-\Bigl(S(\rho_{ABC}^\alpha)+S(\rho_{BCD}^\alpha)-S(\rho_{BC}^1)-2\log d_\alpha \Bigl)\nn
&=2\log d_\beta -2\log d_\alpha +2\log d_\alpha\nn
&=2 \log d_\beta.
\end{align}

Another way is to consider the following:
\begin{align}
&S(\rho_{ABCD}^{\beta \triangleright\!\triangleleft \alpha})-S(\rho_{ABCD}^{\alpha})\nn
&= [S(\rho_{ABCD}^{\beta \triangleright\!\triangleleft\alpha})-S(\rho_{ABCD}^{\alpha\triangleright\!\triangleleft\alpha})]+[S(\rho_{ABCD}^{\alpha \triangleright\!\triangleleft\alpha})-S(\rho_{ABCD}^{\alpha})],
\end{align}
where $\rho_{ABCD}^{\alpha\triangleright\!\triangleleft\alpha}$ is the state obtained by merging $\rho_{ABC}^\alpha$ and $\rho_{BCD}^\alpha$. Then, by solving an entropy maximization problem making use of the information convex set of the whole region $ABCD$, we find
\begin{align}
\rho_{ABCD}^{\beta \triangleright\!\triangleleft\alpha}&=\sum_{a \in \mathcal{C}} \sum_{j=1}^{N_{a\alpha}^\beta} \frac{d_a}{\sum_{a' \in {\cal C}} N_{a' \alpha}^{\beta} d_{a'}}\rho_{ABCD}^{a \alpha \beta, j}.
\end{align}
Moreover, we can obtain
\begin{align}
&S(\rho_{ABCD}^{\beta \triangleright\!\triangleleft\alpha})-S(\rho_{ABCD}^{\alpha})\nn
&= [S(\rho_{ABCD}^{\beta \triangleright\!\triangleleft\alpha})-S(\rho_{ABCD}^{\alpha \triangleright\!\triangleleft\alpha})]+[S(\rho_{ABCD}^{\alpha \triangleright\!\triangleleft \alpha})-S(\rho_{ABCD}^{\alpha})]\nn
&=[(\log d_\beta+\log \sum_{a \in \mathcal{C}} N_{a \alpha}^\beta d_a)-(\log d_\alpha+\log \sum_{a \in \mathcal{C}} N_{a \alpha}^\alpha d_a)]\nn
&+\log \sum_{a \in \mathcal{C}} N_{a \alpha}^\alpha d_a\nn
&=\log(d_\beta / d_\alpha)+\log \sum_{a \in \mathcal{C}} N_{a \alpha}^\beta d_a.
\end{align}
Comparing the two expressions for $\delta S$ gives
\begin{align}
d_\alpha d_\beta &= \sum_{a \in {\cal C}} N_{a \alpha}^{\beta} d_a.
\end{align}

\end{proof}
\end{theorem}

\noindent We also note that all associativity relations among the labels in $\mathcal{C}$ and $\mathfrak{C}_{y}$ are satisfied by the associativity theorem of Ref.~\onlinecite{huang_knots_2023}. The only relation that requires explicit verification is $(a \times b) \times \alpha = a \times (b \times \alpha),$
which is guaranteed by
\begin{align}
    \sum_{c \in \mathcal{C}}N_{a b}^{c} N_{c \alpha}^{\beta}= \sum_{\gamma \in \mathfrak{C}_y} N_{a \gamma}^{\beta} N_{b \alpha }^{\gamma}.
\end{align}

Now, considering the automorphism of local and noncontractible annuli, we find the following relations.

\begin{theorem}\label{th:NNNtau}
\begin{align}
&N_{a\alpha}^{\beta}=N_{\phi_x(a)\tau_x(\alpha)}^{\tau_x(\beta)}=N_{\phi_y^{n}(a)\alpha}^\beta,
\end{align}
where $a\in \mathcal{C}$, $\alpha, ~\beta \in \mathfrak{C}_y$, and $n \in \mathbb{Z}$.
\begin{proof}
This follows straightforwardly by topologically deforming $Y$ in Fig.~\ref{fig:fusion}(a) along the relevant paths back to the original region and identifying which of the annuli $B_1$, $B_2$, and $B_3$ acquire the automorphism associated with $\phi_x$, $\tau_x$, or $\phi_y$.
\end{proof}
\end{theorem}
\noindent Theorem~\ref{th:NNNtau} also allows us to address the following question: when both a transparent domain wall for $\mathcal{C}$ along $l_y$ and one for $\mathfrak{C}_y$ along the same direction are present, can they be assigned different positions? At first sight, this may appear possible, since the position of a transparent domain wall is a matter of convention. In general, however, such a choice is incompatible with the consistency of the fusion rules. If the transparent domain walls were placed at different positions, the label of the extreme point of the ICS on the local annulus could change under the automorphism, while that on the noncontractible annulus remained fixed. This would in turn imply the additional constraint
\begin{align}
N_{a\alpha}^{\beta}=N_{\phi_x(a)\alpha}^{\beta},
\end{align}
which is generically incompatible with Theorem~\ref{th:NNNtau}. We therefore conclude that all transparent domain walls along $l_y$ must be assigned the same position. Accordingly, we henceforth refer to transparent domain walls for $\mathcal{C}$ and $\mathfrak{C}_{x(y)}$ simply as transparent domain walls.

Additionally, we show that the total quantum dimension of $\mathfrak{C}_{y}$ is equal to that of $\mathcal{C}$, a fact that is useful for determining the quantum dimension of each $\alpha \in \mathfrak{C}_{y}$.
\begin{theorem}\label{th:tqd}
\begin{align}
\sum_{\alpha \in \mathfrak{C}_{y}} d_\alpha^2=\mathcal{D}^2.
\end{align}
\begin{proof}
For a fixed $\beta \in \mathfrak{C}_y$, identities (ii) and (iii) of Theorem~\ref{th:qdalpha} imply that
\begin{align}
\sum_{\alpha \in \mathfrak{C}_{y}} d_\alpha^2&=\sum_{\alpha \in \mathfrak{C}_{y}, ~a \in {\cal C}} d_\alpha d_\beta^{-1}N_{a \alpha }^{\beta} d_a ~~({\rm using ~(iii)})\nn
&=\sum_{\alpha \in \mathfrak{C}_{y}, ~a \in {\cal C}} N_{a \alpha}^{\beta}d_\alpha d_\beta^{-1} d_a\nn
&=\sum_{a \in {\cal C}} d_a^2=\mathcal{D}^2.~~({\rm using ~(ii)})
\end{align}

\end{proof}
\end{theorem}

\section{Topological Entanglement Entropy, symmetries, and Entanglement Asymmetry}
\label{sec:4}

We now specialize to the case $d_a = 1$ for all $a \in \mathcal{C}$ and derive additional properties specific to fixed points of Abelian topological phases. We begin by establishing several general results that will be used throughout this section and by clarifying their implications for the TEE. We then define the MESs and explain their relation to the extreme points of the ICSs on noncontractible annuli. We also clarify the properties of symmetries. Finally, building on these results, we interpret the entanglement asymmetry from the perspective of the ICSs.

\subsection{Topological entanglement entropy}

We begin by proving several propositions specific to Abelian topological phases that will be useful for understanding the TEE.

\begin{proposition}\label{pro:N=1}
Let $a$ be an Abelian anyon. For the fusion rules
\begin{equation}
a\times\alpha = \sum_{\beta\in\mathfrak{C}_y}
N_{a\alpha}^{\beta}\beta,
\end{equation}
where $\alpha\in\mathfrak{C}_y$, we have
\begin{equation}
\sum_{\beta\in\mathfrak{C}_y} N_{a\alpha}^{\beta}=1.
\end{equation}
Equivalently, fusion with anyon $a$ permutes the labels in $\mathfrak{C}_y$.
\begin{proof}
Because $a$ is Abelian, it is invertible under fusion, with
inverse $\bar a$:
\begin{equation}
\bar a\times a
=
a\times\bar a
=
1.
\end{equation}
Consequently, fusion with $a$ is reversible, with inverse given by fusion with $\bar a$.

Suppose, for contradiction, that $a\times\alpha$ is not a single label in $\mathfrak{C}_y$. Then its fusion decomposition contains either at least two distinct labels with nonzero multiplicities, or a single label with multiplicity greater than one. Fusing this decomposition with $\bar a$ and using associativity gives
\begin{align}
\alpha&= (\bar a\times a)\times\alpha =
\bar a\times(a\times\alpha)\nn
&= \sum_{\beta\in\mathfrak{C}_y}
N_{a\alpha}^{\beta} \left(\bar a\times\beta\right).
\end{align}
Every $\bar a\times\beta$ appearing on the
right-hand side is nonzero, since fusion with $a$ recovers $\beta$. Moreover, all fusion multiplicities are nonnegative integers, so distinct contributions cannot cancel. The right-hand side would therefore yield a nontrivial decomposition of the single label $\alpha$, leading to a contradiction.

It follows that $a\times\alpha$ is itself a single label. Therefore, there exists a unique $\beta\in\mathfrak{C}_y$ such that
\begin{equation}
a\times\alpha=\beta.
\end{equation}
Equivalently,
\begin{equation}
N_{a\alpha}^{\gamma}
=\delta_{\gamma,\beta},
\end{equation}
where $\gamma\in\mathfrak{C}_y$. Hence,
\begin{equation}
\sum_{\gamma\in\mathfrak{C}_y}
N_{a\alpha}^{\gamma} =1.
\end{equation}
Since fusion with $\bar a$ is the inverse operation, fusion with $a$ defines a permutation of
$\mathfrak{C}_y$.
\end{proof}
\end{proposition}

\begin{proposition}\label{pro:aalphabeta}
For $\alpha, \beta \in \mathfrak{C}_y$, we can find an $a \in {\cal C}$ that satisfies $a \times \alpha=\beta$.
\begin{proof}
Consider merging two noncontractible annuli labeled by $\beta$ and $\alpha$. Since the merged state exists as an element of the information convex set, it must be a convex combination of $\{\rho_X^{a\alpha\beta}\}$ for choices of $a \in \mathcal{C}$. Any choice of $a$ in this convex combination must have $N_{a\alpha}^\beta$ nonzero. Therefore, we can find an $a$ that satisfies $a \times \alpha =\beta$. Precisely, the probability that a given $a$ appears is $d_a N_{a\alpha}^\beta/(d_\alpha d_\beta)$.

\end{proof}
\end{proposition}

\begin{corollary} \label{cor: abelian d}
$d_{\alpha}=d_{\alpha'}$ where $\alpha, \alpha' \in \mathfrak{C}_y$.
\end{corollary}

We now introduce the equivalence classes of anyons, defined through an equivalence relation among them, and establish several related propositions.

\begin{definition}\label{def:5.1}
We define the set of equivalence classes of anyons associated with $\mathfrak{C}_{x(y)}$, denoted by $E_{x(y)}$, as
\begin{align*}
E_{x(y)}=\left\{\, a/\overset{x(y)}{\sim} \;\middle|\; a \in \mathcal{C} \right\},
\end{align*}
where the equivalence relation $\overset{x(y)}{\sim}$ is defined as follows: $a \overset{x(y)}{\sim} b$ if
\[
a \times \alpha = b \times \alpha
\]
for all $\alpha \in \mathfrak{C}_{x(y)}$.
\end{definition}
\noindent One can easily check with associativity that $a \overset{x(y)}{\sim} b$ implies $c \times a \overset{x(y)}{\sim} c \times b$ for any $c \in \mathcal{C}$. Moreover, we can prove the following proposition.

\begin{proposition}\label{pro:sim}
If
\begin{align}\label{eq:aalphabbeta}
    a \times \alpha =b \times \alpha
\end{align}
for a fixed $\alpha \in \mathfrak{C}_{x(y)}$, then $a \overset{x(y)}{\sim} b$.
\begin{proof}
It remains to show that $a \times \beta = b \times \beta$ for any $\beta \in \mathfrak{C}_{x(y)}$. By Proposition~\ref{pro:aalphabeta}, any $\beta \in \mathfrak{C}_{x(y)}$ can be written as $\beta=c\times\alpha$ for some $c\in\mathcal{C}$. \eqref{eq:aalphabbeta} gives
\begin{align}
a\times c\times\alpha=b\times c\times\alpha,
\end{align}
and therefore $a\times\beta=b\times\beta$ for any $\beta \in \mathfrak{C}_{x(y)}$. Thus, $a \overset{x(y)}{\sim} b$.

\end{proof}
\end{proposition}

\noindent We can further derive the following propositions.

\begin{proposition}\label{pro:equiv1}
The number of anyons in $\mathcal{C}$ that are equivalent to $1$ under the relation $\overset{x(y)}{\sim}$ is $|\mathcal{C}|/{|\mathcal{C}^{\phi_{x(y)}}|}$.
\begin{proof}
From Theorem~\ref{th:qdalpha} and Corollary~\ref{cor: abelian d}, together with the fact that $d_a=1$ for all $a \in \mathcal{C}$, we obtain
\begin{equation}
\sum_{a \in \mathcal{C}} N_{a \alpha}^\beta = d_\alpha d_\beta = d_\alpha^2.\label{eq:C1}
\end{equation}
This means that the number of anyons mapping $\alpha$ to $\beta$ is independent of $\beta$ and equals $d_\alpha^2$. In particular, the number of anyons that map $\alpha$ to itself is also $d_\alpha^2$. Using Theorems~\ref{th:cc} and \ref{th:tqd}, together with Corollary~\ref{cor: abelian d}, we obtain
\begin{equation}
\sum_{\alpha \in \mathfrak{C}_{x(y)}} d_\alpha^2
=
|\mathfrak{C}_{x(y)}|d_\alpha^2
=
|\mathcal{C}^{\phi_{x(y)}}|d_\alpha^2
=
\mathcal{D}^2
=
|\mathcal{C}|. 
\end{equation}
Therefore, we find
$d_\alpha^2=|\mathcal{C}|/{|\mathcal{C}^{\phi_{x(y)}}|}.
$

\end{proof}
\end{proposition}

\begin{proposition}
$|E_{x(y)}|=|\mathcal{C}^{\phi_{x(y)}}|$.
\begin{proof}
Using Proposition~\ref{pro:equiv1}, we conclude that each equivalence class of anyons associated with $\mathfrak{C}_{x(y)}$ contains $|\mathcal{C}|/|\mathcal{C}^{\phi_{x(y)}}|$ distinct anyons. Therefore, the set $E_{x(y)}$ has cardinality
\begin{align}
|E_{x(y)}|=|\mathcal{C}^{\phi_{x(y)}}|.
\end{align}
\end{proof}
\end{proposition}

\begin{lemma}\label{lem:equiv2}
For $b \in \mathcal{C}$, one has $1 \overset{x(y)}{\sim} b$ if and only if
\begin{align}
b=\phi_{x(y)}^{n}(a)\times \bar{a}
\end{align}
for some $a \in \mathcal{C}$ and $n \in \mathbb{Z}$.
\begin{proof}
We first prove the ``if'' direction. Suppose that
\[
b=\phi_{x(y)}^{n}(a)\times \bar{a}.
\]
Theorem~\ref{th:NNNtau} implies
\[
N_{a\alpha}^{\beta}
=
N_{\phi_{x(y)}^{n}(a)\alpha}^{\beta},
\]
and hence
\[
a \overset{x(y)}{\sim} \phi_{x(y)}^{n}(a).
\]
Multiplying both sides by $\bar{a}$ gives
\begin{align}
1 \overset{x(y)}{\sim} \phi_{x(y)}^{n}(a)\times \bar{a}=b.
\end{align}

We next prove the converse by counting the number of distinct elements of the form $\phi_{x(y)}^{n}(a)\times \bar{a}$ for all $a\in\mathcal{C}$ and $n\in\mathbb{Z}$. By Proposition~\ref{pro:phi}, we may restrict to $n\geq 0$ without loss of generality. Define
\begin{align}
S
&=
\left\{
\phi_{x(y)}^{n}(a)\times \bar{a}
\;\middle|\;
a\in\mathcal{C},\ n\in\mathbb{Z}
\right\},
\nonumber\\
S_1
&=
\left\{
\phi_{x(y)}(a)\times \bar{a}
\;\middle|\;
a\in\mathcal{C}
\right\}.
\end{align}
By definition, $S_1\subseteq S$. Conversely, for $n=0$, one has $1\in S$. Moreover, $1\in S_1$ as well, since $\phi_{x(y)}(1)\times \bar{1}=1$. For any $n>0$, define
\begin{align}
c
=
a \times \phi_{x(y)}(a) \times \phi_{x(y)}^2(a)
\times \cdots \times \phi_{x(y)}^{n-1}(a).
\end{align}
Then
\begin{align}
\phi_{x(y)}(c)
=
\phi_{x(y)}(a) \times \phi_{x(y)}^2(a)
\times \cdots \times \phi_{x(y)}^{n}(a),
\end{align}
and therefore
\begin{align}
\phi_{x(y)}(c)\times \bar{c}
=
\phi_{x(y)}^n(a)\times \bar{a}.
\end{align}
Thus every element of $S$ is contained in $S_1$, and hence $S=S_1$. Now define the homomorphism
\begin{align}
\eta:\mathcal{C}\to\mathcal{C},
\qquad
\eta(a)=\phi_{x(y)}(a)\times \bar{a}.
\end{align}
Then
\begin{equation}
S_1=\operatorname{Im}(\eta).
\end{equation}
The kernel of $\eta$ is
\begin{align}
\ker(\eta)
&=
\left\{
a\in\mathcal{C}
\;\middle|\;
\phi_{x(y)}(a)\times \bar{a}=1
\right\}
\nonumber\\
&=
\mathcal{C}^{\phi_{x(y)}}.
\end{align}
By the first isomorphism theorem, we obtain
\begin{equation}
|\operatorname{Im}(\eta)|
=
\frac{|\mathcal{C}|}{|\ker(\eta)|}
=
\frac{|\mathcal{C}|}{|\mathcal{C}^{\phi_{x(y)}}|}.
\end{equation}
Finally, by Proposition~\ref{pro:equiv1}, the equivalence class of $1$ under $\overset{x(y)}{\sim}$ has the same cardinality. Therefore, if $1 \overset{x(y)}{\sim} b$, then $b\in S$, i.e.,
\[
b=\phi_{x(y)}^{n}(a)\times \bar{a}
\]
for some $a\in\mathcal{C}$ and $n\in\mathbb{Z}$. This completes the proof.
\end{proof}
\end{lemma}

This lemma will be useful when we discuss symmetries in Sec.~\ref{sec:4c}.

\begin{porism}\label{porism:equiv2}
For $b \in \mathcal{C}$, one has $1 \overset{x(y)}{\sim} b$ if and only if
\begin{align}
b=\phi_{x(y)}(a)\times \bar{a}
\end{align}
for some $a \in \mathcal{C}$.
\end{porism}

We now relate the tunneling process to the fusion rules and prove a property of the entanglement entropy of the extreme points of the ICS.

\begin{proposition}\label{pro:fusion}
Consider a local annulus $X$ associated with an extreme point labeled by $a \in \mathcal{C}$, and a noncontractible annulus $X' \in \mathtt{A}_y$ associated with an extreme point labeled by $\alpha \in \mathfrak{C}_y$. If the fusion rule is of the form $a \times \alpha = \beta$, then tunneling $X$ along $l_x$ across $X'$ transforms $\alpha$ into $\beta$; see Fig.~\ref{fig:tunneling}.

\begin{figure}
\centering
\includegraphics[width=0.95\linewidth]{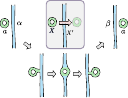}
\caption{Tunneling a contractible annulus through a noncontractible annulus by applying a sequence of quantum channels, as in Section~3.1 of Ref.~\onlinecite{shi_immersed_2024}.}
\label{fig:tunneling}
\end{figure}

\begin{proof}
The tunneling process is as follows. First, the local annulus labeled by $a$ is merged with the noncontractible annulus labeled by $\alpha$, resulting in the configuration shown in Fig.~\ref{fig:fusion}(a). The entire region can then be continuously deformed, as illustrated in the bottom panels of Fig.~\ref{fig:tunneling}. We can then trace out the local annulus. By the definition of the fusion rules and the relation $a \times \alpha = \beta$, the remaining noncontractible annulus must be labeled by $\beta$.
\end{proof}
\end{proposition}

\begin{lemma} \label{lem:sameQD}
The extreme points of $\Sigma(X)$ for $X \in \mathtt{A}_y$, labeled by $\alpha\in\mathfrak{C}_y$, have the same entanglement entropy $S(\rho^{\alpha}_X)$ on the subregion $X$.
\end{lemma}
\begin{proof}
From Corollary \ref{cor: abelian d} and Proposition~\ref{pro:alphabeta}, we find 
\begin{equation}
S(\rho^\alpha_X)=S(\rho_X^{\beta}),
\end{equation}
where $\alpha, \beta \in \mathfrak{C}_y$.
\end{proof}

\begin{figure}[!t]
\centering
\begin{overpic}[width=0.75\linewidth]{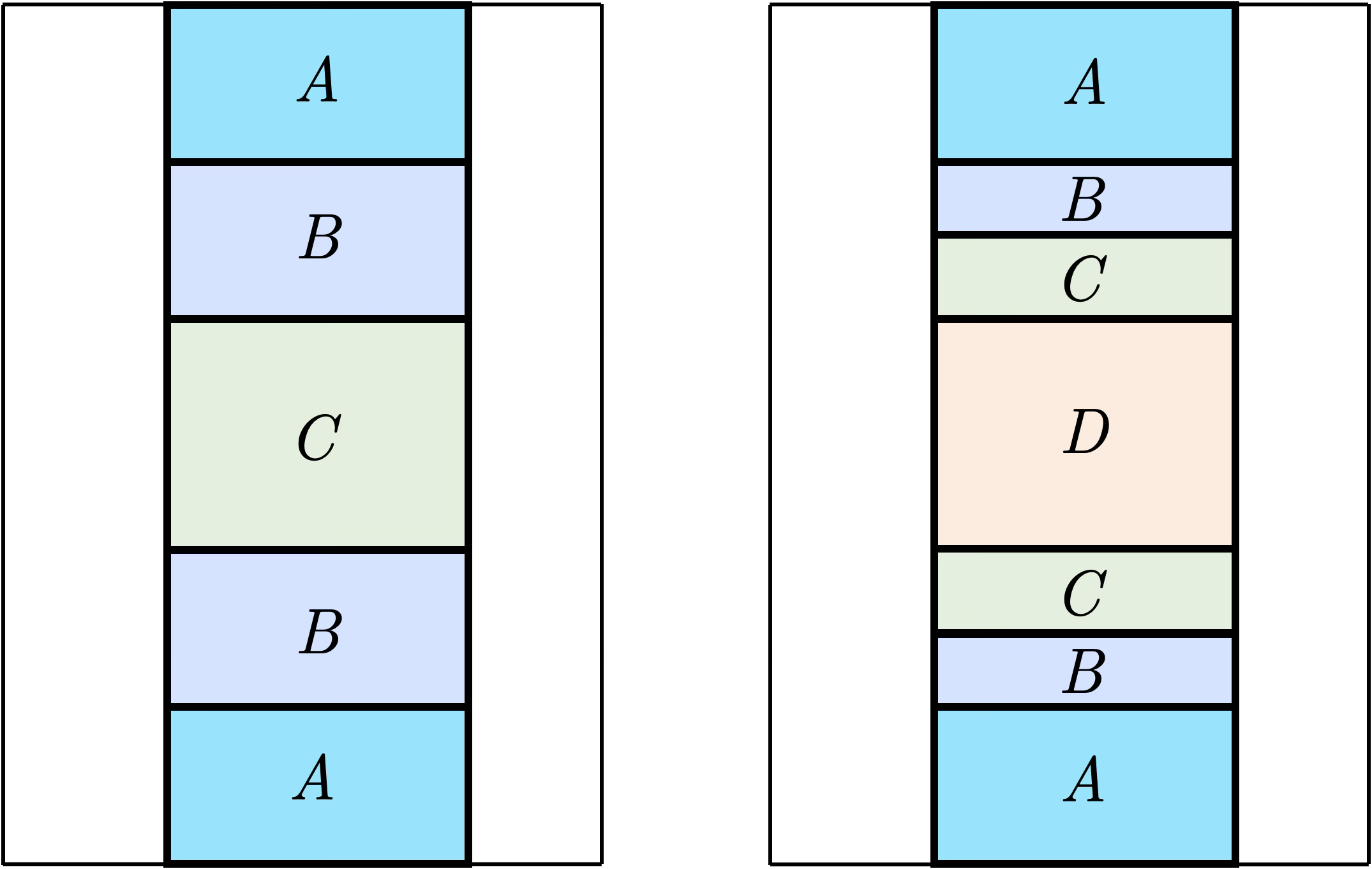}
\put(-6.5,63){(a)}
\put(49,63){(b)}
\end{overpic}
\caption{(a) Partition of $X=ABC \in \mathtt{A}_y$ used in Definition~\ref{def:TEE} and Theorem~\ref{th:TEE}. (b) Partition of $X=ABCD \in \mathtt{A}_y$ used in Theorem~\ref{th:TEE}.}
\label{fig:TEE}
\end{figure}

We now introduce the TEE for a noncontractible annulus.

\begin{definition}[Topological entanglement entropy]\label{def:TEE}
For $X \in \mathtt{A}_{y}$, the topological entanglement entropy of $\rho_X$ is defined as
\begin{align*}
\gamma_{{\rm LW},y}&\equiv \frac{1}{2} I(A:C|B)_{\rho_X}\nn
&=\frac{1}{2}\Bigl( S(\rho_{AB})+S(\rho_{BC})-S(\rho_{B})-S(\rho_{ABC})\Bigl),\label{eq:TEE}
\end{align*}
where the corresponding partition is illustrated in Fig.~\ref{fig:TEE}(a).
\end{definition}
\noindent Note that Definition~\ref{def:TEE} is analogous to the definition of TEE for a local annulus introduced in Ref.~\onlinecite{levin_tee_2006}. We can then establish the following theorem, which constitutes one of the main results of this work.

\begin{theorem}\label{th:TEE}
Consider the partition of the subregion $X=ABC$ in Fig.~\ref{fig:TEE}(a). We can show that
\begin{equation}
    \gamma_{\mathrm{LW},y}^{\max} \equiv \frac{1}{2}I(A:C|B)_{\rho_X^\alpha} = \log \frac{\mathcal{D}}{d_{\alpha}},
\label{eq:gamma-MES} \end{equation}
where $\alpha \in \mathfrak{C}_y$, and $\rho_X^\alpha$ denotes the extreme point of $\Sigma(X)$ labeled by $\alpha$.
\end{theorem}
\begin{proof}
We consider the merging of $\rho_{ABC}^a \in {\rm ext}\Sigma(ABC)$ and $\rho_{BCD}^b\in {\rm ext}\Sigma(BCD)$ for the regions shown in Fig.~\ref{fig:TEE}(b). The unique merged state $\rho_{ABCD}^{a\triangleright\!\triangleleft b}$ is defined on the noncontractible annulus $X$ and will be denoted by $\tilde{\rho}_X$ for simplicity. By the property of merging, $\tilde{\rho}_X$ lies at the center of $\Sigma(X)$. Dividing $X$ as in Fig.~\ref{fig:TEE} gives
\begin{equation}
I(A:C|B)_{\tilde{\rho}_X} = 0,\label{eq:IACB}
\end{equation}
due to the properties of the merged state. Because of the simplex structure of $\Sigma(X)$, the state $\tilde{\rho}_X$ can be expressed as a convex combination of the extreme points:
\begin{align}
\tilde{\rho}_X = \sum_{\alpha \in \mathfrak{C}_y} p_\alpha \rho_X^{\alpha},
\end{align}
where $\{p_\alpha \}$ is a probability distribution. Since the extreme points are orthogonal and have the same entanglement entropy (by Lemma \ref{lem:sameQD}), 
\begin{align}
S({\tilde{\rho}_X}) =- \sum_{\alpha \in \mathfrak{C}_y}p_\alpha \log p_\alpha +S(\rho_X^{\alpha}),\label{eq:EEarb}
\end{align}
and by the standard structure of the central point of the information convex set, the central state has weights $p_\alpha=\frac{d_\alpha^2}{\mathcal{D}^2}$. Therefore, $\tilde{\rho}_X$ is expressed as
\begin{align}
\tilde{\rho}_X = \sum_{\alpha \in \mathfrak{C}_y}  \frac{d_{\alpha}^2}{D^2} \rho_X^{\alpha}.
\end{align}
Since $d_{\alpha}=d_{\alpha'}$ by Corollary~\ref{cor: abelian d}, one finds that $\tilde\rho_X$ is a uniform mixture of the extreme points of the ICS, and hence lies at the center of $\Sigma(X)$. As a consequence,
\begin{equation}
S(\tilde{\rho}_X) - S(\rho_X^{\alpha})= 2\log \frac{D}{d_{\alpha}}.
\end{equation}
By the isomorphism theorem, $\tilde{\rho}_X$ and $\rho_X^\alpha$ have identical entropies on $AB$, $BC$, and $B$. Therefore, the difference between $I(A:C|B)_{\rho_X^\alpha}$ and $I(A:C|B)_{\tilde{\rho}_X}$ comes entirely from the $S(ABC)$ term. It then follows from \eqref{eq:IACB} that
\begin{equation}
I(A:C|B)_{\rho_X^\alpha}=2\log \frac{\mathcal{D}}{d_{\alpha}} .
\end{equation}
Thus, \eqref{eq:gamma-MES} follows.
\end{proof}

\noindent Because the density matrices corresponding to distinct extreme points have mutually orthogonal supports, while $S(\rho_{AB})$, $S(\rho_{BC})$, and $S(\rho_{B})$ are fixed throughout the information convex set, $I(A:C|B)_{\rho_X}$ in Fig.~\ref{fig:TEE}(a) attains its maximum at an extreme point. This observation motivates the definition of $\frac{1}{2}I(A:C|B)_{\rho_X^\alpha}$ as $\gamma_{\mathrm{LW},y}^{\max}$. We emphasize that $\gamma_{{\rm LW},y}^{\rm max}$ captures only the contributions from transparent domain walls intersecting the subregion $X \in \mathtt{A}_y$, because the quantum dimensions associated with the extreme points of the ICS on $X$ encode only the effects of such transparent domain walls.

A similar result was obtained in previous work~\cite{kim_unveiling_2025}, where the lattice model was interpreted as realizing a translation symmetry-enriched topological phase, and the TEE was understood to reflect the quantum dimension of the corresponding translation symmetry defect. In contrast, our interpretation is in terms of transparent domain walls and does not invoke any translation symmetry.

\subsection{Minimum entropy states}
\label{sec:4b}

We now establish the connection between the information extracted from the ICS on the torus $\mathtt{R}$ and that obtained from the ICS on a noncontractible annulus $X$. To this end, we first note that the ICS on $\mathtt{R}$, denoted by $\Sigma(\mathtt{R})$, contains infinitely many extreme points when $N(\mathtt{R})>1$, where $N(\mathtt{R})$ denotes the maximal number of mutually orthogonal extreme points of $\Sigma(\mathtt{R})$. For example, when a toric code ground state is chosen as the reference state, ${\rm ext}\Sigma(\mathtt{R})$ consists of the pure state density matrices associated with all normalized superpositions of the toric code ground states and therefore contains infinitely many elements. Moreover, $N(\mathtt{R})$ is equal to the GSD of the parent Hamiltonian associated with the reference state~\cite{yang_topological_2025}.

Among the extreme points of the ICS on $\mathtt{R}$, we define the MESs associated with $X \in \mathtt{A}_{x(y)}$ as follows:
\begin{definition}[Minimum entropy states]\label{def:MES}
MESs associated with $X \in \mathtt{A}_{x(y)}$ are extreme points of $\Sigma(\mathtt{R})$ whose reduced density matrices on the subregion $X$, obtained by tracing out its complement, are extreme points of the ICS $\Sigma(X)$.
\end{definition}

\noindent The MESs then satisfy the following proposition, which explains both why we refer to them as \textit{minimum entropy states} and how they are equivalent to the original definition introduced in Ref.~\onlinecite{zhang_quasiparticle_2012}.

\begin{proposition} The MES associated with $X \in \mathtt{A}_{x(y)}$ gives the minimum entanglement entropy for the subregion $X\in \mathtt{A}_{x(y)}$.
\begin{proof}
By definition, taking the partial trace over the complement of a subregion $X \in \mathtt{A}_{x(y)}$ yields a density matrix that is an extreme point of $\Sigma(X)$. Then, using Lemma~\ref{lem:sameQD} together with the fact that the entanglement entropy of an arbitrary point in $\Sigma(X)$ is given by \eqref{eq:EEarb}, we conclude that the corresponding extreme points realize the minimum entanglement entropy.
\end{proof}
\end{proposition}

\begin{definition}
We denote by $\mathfrak{C}_{x(y)}^{\tau_{y(x)}}$ the subset of $\mathfrak{C}_{x(y)}$ whose corresponding density matrices are obtained by taking the partial trace of MESs over the complement of a subregion $X\in \mathtt{A}_{x(y)}$~\footnote{We have only two choices: $\mathfrak{C}_{x}^{\tau_{y}}$ and $\mathfrak{C}_{y}^{\tau_{x}}$.}.
\end{definition}

We now show with the following propositions that the definition above implies that  $\mathfrak{C}_{x(y)}^{\tau_{y(x)}}$ is the subset of extreme points in $\mathfrak{C}_{x(y)}$ which are left invariant under the action of $\tau_{y(x)}$.

\begin{proposition}
If $\rho$ is an extreme point of  $\Sigma(\mathtt{R})$ and it is an MES associated with $X_0 \in \mathtt{A}_{x(y)}$, then $\rho$ is also an MES associated with $X_1 \in \mathtt{A}_{x(y)}$.
\begin{proof}
Let $\rho_{X_0}$ be an extreme point of the ICS on $X_0$. Consider the region $ABC$ shown in Fig.~\ref{fig:qdalpha}(b), where $X_0 = AB$ and $BC$ is a disk. Since Axiom ${\bf A1}$ applies to disks, we have $I(A:C | B)=0$ for the region $ABC$. Moreover, $\rho_{BC}$ is completely fixed because $BC$ is a disk. Hence, the density matrix on $ABC$ is uniquely determined and coincides with the reduced density matrix obtained from $\rho$ by tracing out the complement of $ABC$. Similarly, when $X_0 = A'B'C'$, we conclude that $\rho_{A'B'}$ coincides with the reduced density matrix obtained from $\rho$ by tracing out the complement of $A'B'$. These two procedures correspond to enlarging and reducing the region $X_0$.

Repeating this argument sequentially, we find that the density matrix on the region $X_1$ is likewise uniquely determined and is given by the reduced density matrix of $\rho$ on $X_1$. By the isomorphism theorem, this density matrix on $X_1$ must be an extreme point of the ICS. Therefore, $\rho$ is an MES associated with $X_1 \in \mathtt{A}_{x(y)}$.
\end{proof}
\end{proposition}

\begin{porism}\label{porism:1}
For any extreme point in $\mathfrak{C}_{x(y)}^{\tau_{y(x)}}$ on a subregion $X \in \mathtt{A}_{x(y)}$, transporting it back to the original region along the path returns the same extreme point. In other words, $\tau_{y(x)}$ acts trivially on the extreme points of $\mathfrak{C}_{x(y)}^{\tau_{y(x)}}$.
\end{porism}

\begin{proposition}\label{pro:Ctau}
For  $\alpha \in \mathfrak{C}_{x(y)}$, we have $ \alpha \in \mathfrak{C}_{x(y)}^{\tau_{y(x)}}$ if and only if $\tau_{y(x)}(\alpha) = \alpha$. Equivalently,
\begin{align}
\mathfrak{C}_{x(y)}^{\tau_{y(x)}}=\{\alpha \in \mathfrak{C}_{x(y)}| \tau_{y(x)}(\alpha)=\alpha \}.
\end{align}

\begin{proof}
By Porism~\ref{porism:1}, any $\alpha \in \mathfrak{C}_{x(y)}^{\tau_{y(x)}}$ is invariant under $\tau_{y(x)}$. It thus remains to show that any $\alpha \notin \mathfrak{C}_{x(y)}^{\tau_{y(x)}}$ is not invariant under $\tau_{y(x)}$, i.e., $\tau_{y(x)}(\alpha)\neq \alpha$. We establish this by contradiction. Suppose that there exists an $\alpha \notin \mathfrak{C}_{x}^{\tau_y}$ that is nevertheless invariant under $\tau_y$. We then partition the torus into annuli $A$, $B$, $C$, and $D$, as illustrated in Fig.~\ref{fig:e_merge}. The state $\rho^\alpha_{BCD}$ is obtained from $\rho^\alpha_{ABC}$ by a deformation. Since $\alpha$ is not permuted, the two states coincide on the overlap region $BC$. Moreover, both states satisfy the Markov condition, namely, $I(A:C|B)=0$ and $I(B:D|C)=0$, respectively. It then follows from the merging theorem that one can construct a merged state $\sigma_{ABCD}\in \Sigma(\mathtt{R})$.

We then show that such a state $\sigma_{ABCD}$ must be a pure state. One way to show this is to use the associativity theorem of Ref.~\onlinecite{huang_knots_2023} (Theorem~2.22 there). It implies that $\sigma_{ABCD}$ is an extreme point as the fusion multiplicity for both $ABC$ and $BCD$ equals one. The purity of $\sigma$ follows from the fact that any extreme point on a torus (or any closed manifold) must be a pure state. 

Such a pure state $|\sigma\rangle$ must be an MES associated with $X \in \mathtt{A}_{x(y)}$, i.e., $\alpha \in \mathfrak{C}_{x(y)}^{\tau_{y(x)}}$. This contradicts the assumption.
\end{proof}
\end{proposition}

\begin{figure}[!t]
\centering
\includegraphics[width=0.55\linewidth]{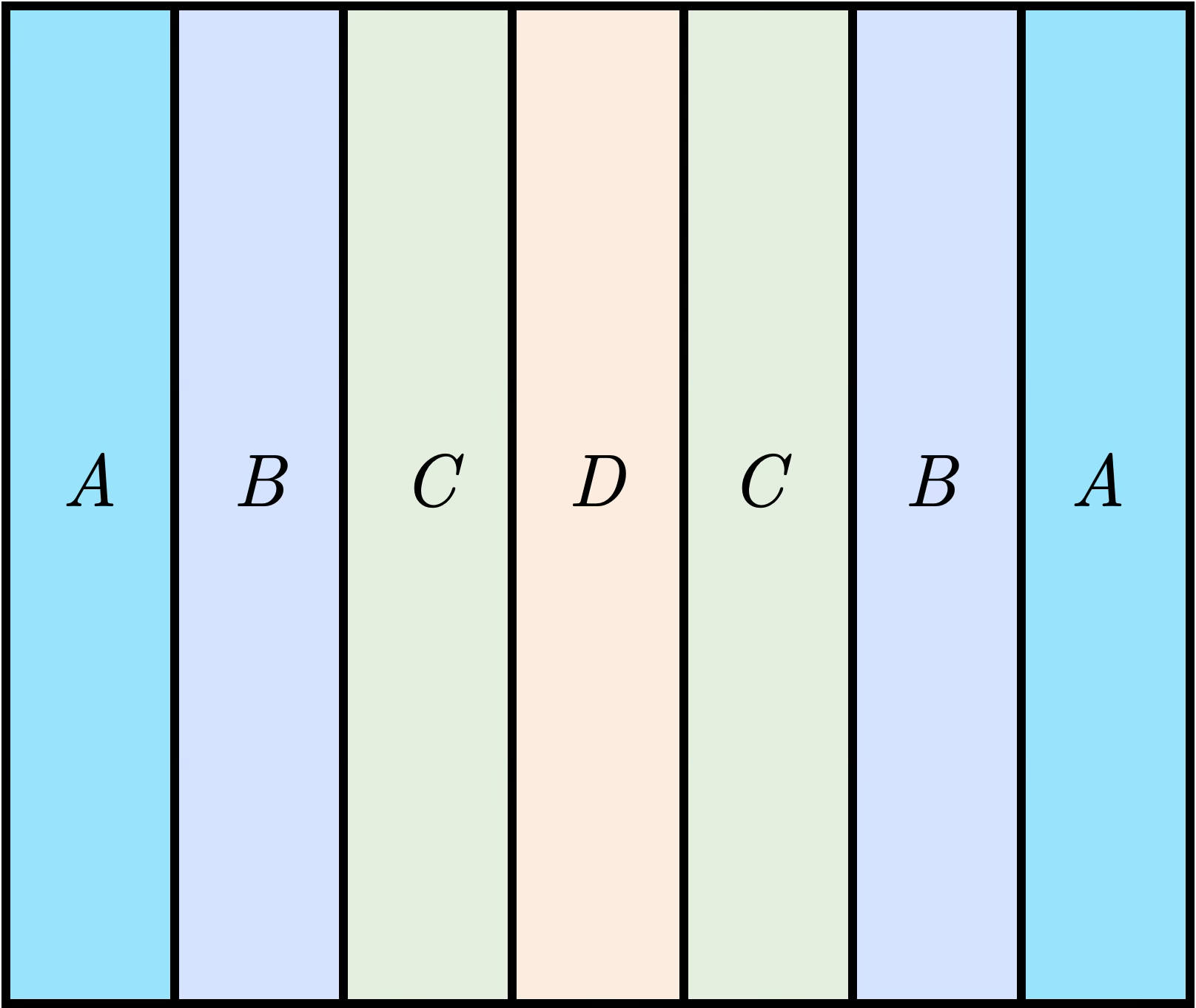}
\caption{Partition of $\mathtt{R}$ used in Propositions~\ref{pro:Ctau} and \ref{pro:oto}}
\label{fig:e_merge}
\end{figure}

\begin{proposition}[One-to-one correspondence]\label{pro:oto}
Each $\alpha \in \mathfrak{C}_{x(y)}^{\tau_{y(x)}}$ is in one-to-one correspondence with an MES associated with $X \in \mathtt{A}_{x(y)}$. 
\begin{proof}
We have already shown that for each $\alpha \in \mathfrak{C}_{x(y)}^{\tau_{y(x)}}$, there is at least one MES. Furthermore, each MES is associated with a label in $\mathfrak{C}_{x(y)}^{\tau_{y(x)}}$.

It remains to show that there cannot be multiple MESs that correspond to the same label $\alpha \in \mathfrak{C}_{x(y)}^{\tau_{y(x)}}$. We prove this using the uniqueness of a quantum Markov state given its marginals. Specifically, we consider the partition $ABCD$ in Fig.~\ref{fig:e_merge}, and suppose we have two MESs $|\psi\rangle$ and $|\varphi\rangle$ associated with the same $\alpha$. Both states reduce to $\rho^\alpha_{ABC}$ on $ABC$ and to $\rho^\alpha_{BCD}$ on $BCD$.
 Both states have $I(A:D|BC)=0$. Therefore, according to the property of a quantum Markov state (under the partition $A,BC,D$) we must have $|\psi\rangle$ equal to $|\varphi\rangle$ up to a global phase. Therefore, there is a one-to-one correspondence between the MESs associated with $X \in \mathtt{A}_{x(y)}$ and $\alpha \in \mathfrak{C}_{x(y)}^{\tau_{y(x)}}$.

\end{proof}
\end{proposition}

\noindent Following Proposition~\ref{pro:oto}, we introduce the following notation for the MESs and use it throughout the remainder of the paper.
\begin{definition}\label{def:M}
For $X \in \mathtt{A}_{x(y)}$, the density matrix of an MES is denoted by $|\alpha\rangle\langle \alpha|$, where $\alpha \in \mathfrak{C}_{x(y)}^{\tau_{y(x)}}$. When emphasizing that the MES is a pure state, we simply denote it by $|\alpha\rangle$.
\end{definition}

Moreover, we establish the following propositions, which will be useful throughout the remainder of the paper.

\begin{proposition}[Mutual orthogonality of MESs]\label{pro:mo}
MESs associated with $X \in \mathtt{A}_{x(y)}$ are mutually orthogonal.
\begin{proof}
By Theorem~\ref{th:simplex}, the reduced density matrices on $X$ of distinct MESs associated with $X \in \mathtt{A}_{x(y)}$ have mutually orthogonal supports and hence zero fidelity. Since fidelity is nondecreasing under partial trace, the fidelity between the corresponding MESs must also vanish. Therefore, the MESs associated with $X \in \mathtt{A}_{x(y)}$ are mutually orthogonal.
\end{proof}
\end{proposition}

\begin{proposition}\label{pro:GSD}
The maximal number of mutually orthogonal extreme points of the ICS on $\mathtt{R}$ is equal to the cardinality of $\mathfrak{C}_x^{\tau_{y}}$ and also to that of $\mathfrak{C}_y^{\tau_{x}}$.
\begin{proof}
First, the region $\mathtt{R}$ is a closed manifold. Therefore, $\Sigma(\mathtt{R})$ is isomorphic to the state space of an $N(\mathtt{R})$-dimensional Hilbert space.

Next, we show  that $N (\mathtt{R})= |\mathfrak{C}_x^{\tau_{y}}|$. For this, we apply the associativity theorem (Theorem~2.22 of Ref.~\onlinecite{huang_knots_2023}). We partition the torus as $\mathtt{R}= X X'$ where $X,X' \in \mathtt{A}_x$. Applying the simplex theorem to each annulus, we find a set of integers $N^{\alpha \beta}(X)$ and $N^{\alpha \beta}(X')$, such that
\begin{equation}
    N^{\alpha \beta}(X)= \delta_{\alpha, \beta}, \quad N^{\alpha \beta}(X')= \delta_{\alpha, \tau_y(\beta)},
\end{equation}
assuming the defect is in $X'$. By the associativity theorem, gluing the two annuli to obtain the torus gives a way to compute $N(\mathtt{R})$ from $\{ N^{\alpha \beta}(X)\}$ and $\{ N^{\alpha \beta}(X')\}$ as
\begin{equation}
\begin{aligned}
    N(\mathtt{R}) &= \sum_{\alpha, \beta \in \mathfrak{C}_x} N^{\alpha \beta}(X) N^{\alpha \beta}(X') \nn
    &= \sum_{\alpha ,\beta \in \mathfrak{C}_x}  \, \delta_{\alpha, \beta} \delta_{\alpha, \tau_y(\beta)} \nn
    &= \sum_{\alpha \in \mathfrak{C}_x} \delta_{\alpha, \tau_y(\alpha)}= |\mathfrak{C}_x^{\tau_{y}}|.
\end{aligned}   
\end{equation}
Finally, we can repeat the same argument in the other direction, and obtain $N (\mathtt{R})= |\mathfrak{C}_y^{\tau_{x}}|$. From this we conclude that $N(\mathtt{R})=|\mathfrak{C}_y^{\tau_{x}}|= |\mathfrak{C}_x^{\tau_{y}}|$.
\end{proof}
\end{proposition}

As noted at the beginning of this subsection, $N(\mathtt{R})$ is equal to the GSD of the parent Hamiltonian for the reference state. Therefore, the GSD also coincides with the number of extreme points of $\Sigma(X)$ that arise as reduced density matrices of MESs.

Using the properties of MESs established above, we can prove the following proposition.

\begin{proposition}\label{pro:id}
\begin{align}
\sum_{\alpha \in \mathfrak{C}_{x(y)}^{\tau_{y(x)}}} |\alpha\rangle \langle\alpha|=I,
\end{align}
where $I$ denotes the identity operator on the Hilbert space whose state space is $\Sigma(\mathtt{R})$.

\begin{proof}
By Propositions~\ref{pro:mo} and~\ref{pro:GSD}, the MESs are mutually orthogonal, and their number equals the dimension of the space spanned by the mutually orthogonal extreme points of $\Sigma(\mathtt{R})$. Therefore, the MESs form a complete orthonormal basis of this space, and the resolution of the identity follows.
\end{proof}
\end{proposition}

\noindent This also leads to the general expression of $\rho \in \Sigma(\mathtt{R})$.

\begin{proposition}\label{pro:rho}
Any $\rho \in \Sigma(\mathtt{R})$ can be expressed as
\begin{align}
\rho=\sum_{i} q_i |\psi_i\rangle \langle\psi_i|,
\end{align}
where $\{q_i\}$ is a probability distribution and
\begin{align}
|\psi_i\rangle
=
\sum_{\alpha \in \mathfrak{C}_{x(y)}^{\tau_{y(x)}}}
c_\alpha^i |\alpha \rangle,
\end{align}
with $\sum_{\alpha \in \mathfrak{C}_{x(y)}^{\tau_{y(x)}}} |c_\alpha^i|^2=1$.

\begin{proof}
By Proposition~\ref{pro:id}, the MESs ${|\alpha\rangle}$, with $\alpha \in \mathfrak{C}_{x(y)}^{\tau_{y(x)}}$, form a complete orthonormal basis for the relevant Hilbert subspace. Hence, every normalized pure state $|\psi_i\rangle$ in this subspace can be expanded as
\begin{align}
|\psi_i\rangle
=
\sum_{\alpha \in \mathfrak{C}_{x(y)}^{\tau_{y(x)}}}
c_\alpha^i |\alpha \rangle.
\end{align}
The normalization of $|\psi_i\rangle$ and the orthonormality of the MESs imply
\begin{align}
\sum_{\alpha \in \mathfrak{C}_{x(y)}^{\tau_{y(x)}}}
|c_\alpha^i|^2=1.
\end{align}
Finally, by the spectral decomposition, any $\rho \in \Sigma(\mathtt{R})$ can be written as
\begin{align}
\rho=\sum_i q_i |\psi_i\rangle\langle\psi_i|,
\end{align}
where $q_i\geq 0$ and $\sum_i q_i=1$. This proves the proposition.
\end{proof}
\end{proposition}

\subsection{Symmetries}
\label{sec:4c}

At fixed points of Abelian topological phases without transparent domain walls, the generalized symmetries underlying the spontaneous symmetry breaking description of topological order can be realized by creating an anyon--anti-anyon pair, transporting the anyons along prescribed paths, and fusing all localized excitations back to the vacuum. These processes define $1$-form symmetry operators~\cite{mcgreevy_generalized_2023,lake_ssb_2018,xu_entanglement_2025,liu_information-theoretic_2025}. The noncommutative algebra of these symmetry operators reveals the structure of the ground state degeneracy and how the symmetries connect distinct ground states, consistent with the interpretation of topological order in terms of spontaneous $1$-form symmetry breaking. We choose a ground state basis of simultaneous eigenstates of the symmetry operators along $l_x$, while those along $l_y$ connect different basis states.

In the presence of transparent domain walls, however, different phenomena arise. To explain them, we first formalize states with localized anyon excitations and the various operations used in this construction. We then define symmetries and characterize their induced actions on the ICS and the MES labels.

\begin{definition}[States with localized anyon excitations]\label{def:slae}
Let $\sigma_{\mathtt{R}}$ be a reference state defined on $\mathtt{R}$. A state $\rho_{\mathtt{R}}$ is said to contain localized anyon excitations $a_1,\ldots,a_n$ within mutually disjoint disks $D_1,\ldots,D_n$ if the following local properties are satisfied:
\begin{enumerate}
    \item Away from the excitation disks, $\rho_{\mathtt R}$ is locally indistinguishable from $\sigma_{\mathtt R}$. More precisely,
    \begin{equation*}
        \Tr_{\mathtt R\setminus X}\rho_{\mathtt R}
        =
        \Tr_{\mathtt R\setminus X}\sigma_{\mathtt R}
    \end{equation*}
    for every local disk $X$ whose buffered neighborhood is disjoint from $\bigcup_iD_i$. Here, for a fixed microscopic width $w>0$, we define the buffered neighborhood of a region $X$ by
    \begin{align*}
    X^{+w} \equiv \{p\in\mathtt{R} :
        \operatorname{dist}(p,X)\le w\}.
    \end{align*}
    \item For each $i\in \{1, \cdots, n\}$, the reduced density matrix on an annulus $A_i$ surrounding only $D_i$ is the extreme point of $\Sigma(A_i)$ labeled by $a_i$.

    \item The state $\rho_{\mathtt R}$ satisfies Axiom {\bf A0} on every admissible local disk whose buffered neighborhood intersects $\bigcup_iD_i$, including disks containing one or more anyon locations.
 
\end{enumerate}
\end{definition}

We remark that the third condition excludes additional entanglement between the positions of spatially separated excitations. In particular, it rules out coherent superpositions of joint position configurations in which the positions of different anyons are correlated, even when their charges are fixed, provided that the corresponding position degrees of freedom can be separated by an admissible {\bf A0} partition. Furthermore, {\bf A0} on the excitation and the Abelian nature of the excitation imply that {\bf A1} is satisfied on the state.

\begin{figure}[ht]
\centering
\includegraphics[width=0.95\linewidth]{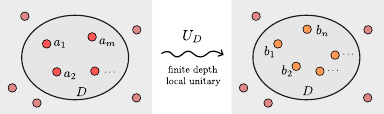}
\caption{Using a finite-depth local unitary $U_D$ to change the anyon content within a disk $D$ from ${a_1, a_2, \cdots, a_m}$ to ${b_1, b_2, \cdots, b_n}$ without affecting the excitations outside $D$.}
\label{fig:placeholder}
\end{figure}

\begin{proposition}[Anyon-transport operators, \cite{remote_braiding_2026}] \label{pro:anyon-transport}
For Abelian topological phases, consider a state defined on $\mathtt{R}$ with localized anyon excitations $a_1,a_2,\cdots, a_m$ contained within a disk $D$, denoted by $\rho_{{\bf a};D}$, which may also contain additional excitations outside $D$. If there exists a finite set of anyon labels ${ b_1, \cdots, b_n}$ such that
\begin{equation}
    a_1 \times a_2 \times \cdots \times a_m = b_1 \times b_2 \times \cdots \times b_n,
    \end{equation}
then the following statements hold:
\begin{enumerate}
    \item For any chosen positions of the localized anyons ${ b_1, \cdots, b_n}$ within the disk $D$, there exists a state $\rho_{{\bf b};D}'$ whose reduced density matrix on the complement of $D$ is identical to that of $\rho_{{\bf a};D}$.
    \item There exists a finite-depth local unitary $U_D$ supported on the disk $D$ such that
    \begin{align}
    &\rho_{{\bf b};D}' = U_D \rho_{{\bf a};D} U_D^\dagger.
    \end{align}
    Moreover, $U_D$ is independent of the state outside $D$.
\end{enumerate}
\end{proposition}

\noindent This proposition is proved in a forthcoming work~\cite{remote_braiding_2026}. It is sufficient to characterize the creation, transport, and fusion processes and provides a concrete description of states with localized anyon excitations. The case $m=0$ corresponds to an initial state with no localized anyon excitations in $D$. When $m=0$, $n=2$, and $b_1 \times b_2=1$, the proposition guarantees the existence of a finite-depth local unitary that creates an anyon--anti-anyon pair. Moreover, the supporting disk can be chosen as a thickened string terminating at the positions of $b_1, b_2$. Similarly, when $m=2$, $n=0$, and $a_1 \times a_2=1$, the proposition guarantees the existence of a finite-depth pair-annihilation operator whose support can also be chosen as a thickened string. The cases $(m,n)=(1,1)$ and $(m,n)=(1,2)$ correspond to the transport and splitting operators, respectively, introduced in Ref.~\onlinecite{Kwagoe2020}. Using these elementary operations, we can construct various states with localized anyon excitations starting from $\rho \in \Sigma(\mathtt{R})$.

Then, we can define the symmetry as follows.

\begin{definition}[Symmetry]
\label{def:symmetry}
A symmetry is a unitary $U$ associated with a fixed finite,
charge-labeled Abelian anyon-transport process
$\mathcal P_U$. The process $\mathcal P_U$ consists of a finite
sequence of definite-charge pair-creation, transport, splitting,
and fusion operations along prescribed paths, and terminates
without leaving any localized excitations. Both $\mathcal P_U$
and the corresponding unitary $U$ are fixed independently of the
input state.

We require that $U$ induce an automorphism of the ICS:
\begin{equation*}
U\Sigma(\mathtt R)U^\dagger=\Sigma(\mathtt R).
\end{equation*}
Equivalently, for every $\rho\in\Sigma(\mathtt R)$,
\begin{equation*}
U\rho U^\dagger\in\Sigma(\mathtt R),
\qquad
U^\dagger\rho U\in\Sigma(\mathtt R).
\end{equation*}
Thus, a single choice of $(\mathcal P_U,U)$ must act consistently
on every state in $\Sigma(\mathtt R)$. Coherent superpositions or
coherently controlled combinations of distinct charge-labeled
anyon-transport processes are excluded.
\end{definition}

\noindent Based on this definition, we can show that the symmetries form a group and admit a 1-skeleton support structure.

\begin{proposition}[Group structure of symmetries]
\label{pro:symmetry-group}
Let $\mathcal G(\mathtt R)$ denote the set of symmetries. Then $\mathcal G(\mathtt R)$ forms a group under operator multiplication.
\end{proposition}

\begin{proof}
Proposition~\ref{pro:anyon-transport} implies that every elementary pair-creation, transport, splitting, or fusion operation can be implemented by a finite-depth local unitary that is independent of the state outside its support. If $W$ implements such an elementary operation, then $W^\dagger$ has the same depth and support and implements the reverse operation. Thus, pair creation and annihilation are mutual inverses, splitting and fusion are mutual inverses, and transport is inverted by transport along the reversed path.

First, allowing a finite process to have length zero, the empty process is admissible and implements the identity $I$. Hence, $I\in\mathcal G(\mathtt R)$.

Next, let $U,V\in\mathcal G(\mathtt R)$ be implemented by processes $\mathcal P_U$ and $\mathcal P_V$, respectively. Performing $\mathcal P_U$ followed by $\mathcal P_V$ gives a finite charge-labeled process
\begin{equation}
\mathcal P_{VU}
=
\mathcal P_V\circ\mathcal P_U
\end{equation}
implementing $VU$. Moreover,
\begin{align}
&(VU)\Sigma(R)(VU)^\dagger=
V\left(U\Sigma(R)U^\dagger\right)V^\dagger\nn
&=V\Sigma(R)V^\dagger= \Sigma(R).
\end{align}
Therefore, $VU\in\mathcal G(\mathtt R)$.

Finally, suppose that
\begin{equation}
\mathcal P_U=(W_1,\ldots,W_N),
\qquad
U=W_N\cdots W_1.
\end{equation}
Reversing the process gives
\begin{equation}
\bar{\mathcal P}_U
=
(W_N^\dagger,\ldots,W_1^\dagger),
\end{equation}
which is admissible by Proposition~\ref{pro:anyon-transport} and implements
\begin{equation}
W_1^\dagger\cdots W_N^\dagger
=
U^\dagger.
\end{equation}
Since $\Sigma(\mathtt R)$ is the state space of the finite-dimensional Hilbert space identified above and $U$ preserves this space, $U^\dagger$ also preserves it. Hence, $U^\dagger\in\mathcal G(\mathtt R)$.

Associativity is inherited from operator multiplication. Therefore, $\mathcal G(\mathtt R)$ is a group.
\end{proof}

\begin{proposition}[One-skeleton support]
\label{pro:fdlu}
Let $U$ be a symmetry with a fixed realization $\mathcal P_U$. Then there exist a finite embedded graph $\Gamma_{\mathcal P_U}\subset\mathtt R$, a finite cell decomposition $K_{\mathcal P_U}$ of $\mathtt R$, and a microscopic width $w=O(1)$ such that
\begin{equation}
    \operatorname{supp}(U)
    \subseteq
    \mathcal N_w(\Gamma_{\mathcal P_U})
    \subseteq
    \mathcal N_w\!\left(K_{\mathcal P_U}^{(1)}\right).
    \label{eq:one-skeleton-support}
\end{equation}
Here, $\mathcal N_w(S)=\left\{p\in\mathtt R \,\middle|\, \operatorname{dist}(p,S)\leq w \right\}$,
where the distance is measured using the microscopic metric on
$\mathtt R$, and $K_{\mathcal P_U}^{(1)}$ denotes the
$1$-skeleton of $K_{\mathcal P_U}$.

The graph, cell decomposition, and width are determined by the
chosen realization $\mathcal P_U$ and are independent of the
input state $\rho\in\Sigma(\mathtt R)$.
\begin{proof}
Write the fixed process as a finite composition
\begin{equation}
    U=V_N\cdots V_1,
\end{equation}
where each $V_j$ implements a definite-charge pair-creation,
transport, splitting, or fusion operation.

Let $\gamma_j\subset\mathtt R$ be the finite piecewise-smooth path
network associated with $V_j$, having only finitely many endpoints
and intersections. By the bounded-width realization of the elementary operation, there exists a microscopic constant
$w_j=O(1)$ such that
\begin{equation}
    \operatorname{supp}(V_j)
    \subseteq
    \mathcal N_{w_j}(\gamma_j).
\end{equation}
The constants $w_j$ are independent of the input state and the
system size. Define
\begin{equation}
    \Gamma_{\mathrm{raw}}
    :=
    \bigcup_{j=1}^{N}\gamma_j,
    \qquad
    w:=\max_{1\leq j\leq N}w_j.
\end{equation}
Since the support of a product of operators is contained in the
union of the supports of its factors,
\begin{align}
    \operatorname{supp}(U)
    &\subseteq
    \bigcup_{j=1}^{N}\operatorname{supp}(V_j) \subseteq
    \mathcal N_w(\Gamma_{\mathrm{raw}}).
    \label{eq:raw-graph-support}
\end{align}

Subdivide $\Gamma_{\mathrm{raw}}$ at all endpoints, intersections, overlapping-segment endpoints, and pair-creation, splitting, fusion, and charge-changing events. Introduce an auxiliary vertex on every circular component containing no such vertex. Because $\Gamma_{\mathrm{raw}}$ is a finite union of piecewise-smooth paths with only finitely many endpoints and intersections, subdividing
it at these points yields a finite embedded graph $\Gamma_{\mathcal P_U}$ with the same underlying subset as $\Gamma_{\mathrm{raw}}$. Equation \eqref{eq:raw-graph-support}
therefore gives
\begin{equation}
    \operatorname{supp}(U)
    \subseteq
    \mathcal N_w(\Gamma_{\mathcal P_U}).
\end{equation}

Finally, the finite piecewise-smooth graph $\Gamma_{\mathcal P_U}$ embedded in the compact surface $\mathtt R$ can be included in the $1$-skeleton of a finite cell decomposition of $\mathtt R$. Hence, there exists a finite cell decomposition $K_{\mathcal P_U}$ satisfying
\begin{equation}
    \Gamma_{\mathcal P_U}
    \subseteq
    K_{\mathcal P_U}^{(1)}.
\end{equation}
It follows that
\begin{equation}
    \mathcal N_w(\Gamma_{\mathcal P_U})
    \subseteq
    \mathcal N_w\!\left(K_{\mathcal P_U}^{(1)}\right),
\end{equation}
which proves \eqref{eq:one-skeleton-support}.
\end{proof}
\end{proposition}

Among these symmetries, many induce the same action on the ICS $\Sigma(\mathtt{R})$. It is therefore natural to introduce an effective symmetry group by identifying symmetries that act equivalently on $\Sigma(\mathtt{R})$.

\begin{definition}[Equivalence classes of symmetries]
\label{def:effective-total-symmetry}
For $U,V\in\mathcal G(\mathtt R)$, define the equivalence relation
\begin{equation*}
U\sim V
\quad\Longleftrightarrow\quad
U\rho U^\dagger=V\rho V^\dagger
\quad
\text{for every }\rho\in\Sigma(\mathtt R).
\label{eq:symmetry-equivalence}
\end{equation*}
The group of equivalence classes of symmetries is then defined as
\begin{equation*}
F_{\mathrm{total}}
\equiv
\mathcal G(\mathtt R)/\sim.
\label{eq:F-total}
\end{equation*}
\end{definition}
\noindent Since every symmetry maps $\Sigma(\mathtt R)$ onto itself, the equivalence relation is compatible with composition, and the quotient $ F_{\mathrm{total}}$ inherits a group structure. The group operation is given by $[U][V]=[UV]$, with identity $[I]$ and inverse $[U]^{-1}=[U^\dagger]$.

\begin{theorem}\label{th:equiv}
Two symmetries $U^{(1)}$ and $U^{(2)}$ are equivalent if and only if $U^{(1)}|\alpha\rangle=e^{i\theta_{(1),(2)}}U^{(2)}|\alpha\rangle$ for every $\alpha \in \mathfrak{C}_{x(y)}^{\tau_{y(x)}}$, where $e^{i\theta_{(1),(2)}}$ is independent of $\alpha$.

\begin{proof}
The ``only if'' direction is proved as follows. By the equivalence relation defined in Definition~\ref{def:effective-total-symmetry}, if $U^{(1)}$ and $U^{(2)}$ are equivalent, then
\begin{align}
(U^{(2)})^\dagger U^{(1)} \rho (U^{(1)})^\dagger U^{(2)}=\rho
\end{align}
for every $\rho \in \Sigma(\mathtt{R})$. First, choosing $\rho=|\alpha\rangle\langle\alpha|$ implies
\begin{align}
(U^{(2)})^\dagger U^{(1)} |\alpha\rangle = \langle \alpha |(U^{(2)})^\dagger U^{(1)}|\alpha \rangle |\alpha\rangle.
\end{align}
Thus, $(U^{(2)})^\dagger U^{(1)}|\alpha\rangle$ lies in the one-dimensional subspace spanned by $|\alpha\rangle$. Since $(U^{(2)})^\dagger U^{(1)}$ is unitary, $\langle \alpha |(U^{(2)})^\dagger U^{(1)}|\alpha \rangle$ has unit modulus and can therefore be written as $e^{i\theta_{(1),(2),\alpha}}$. It follows that $U^{(1)}|\alpha\rangle=U^{(2)}e^{i\theta_{(1),(2),\alpha}}|\alpha\rangle$ for every $\alpha \in \mathfrak{C}_{x(y)}^{\tau_{y(x)}}$. To show that the phase is independent of the MES label, consider any distinct pair $\alpha,\beta \in \mathfrak{C}_{x(y)}^{\tau_{y(x)}}$, where $\alpha \neq \beta$, and choose
\begin{align}
\rho=\frac{1}{2}(|\alpha\rangle \langle \alpha|+|\alpha\rangle \langle \beta|+|\beta\rangle \langle \alpha|+|\beta\rangle \langle \beta|).
\end{align}
The condition $(U^{(2)})^\dagger U^{(1)} \rho (U^{(1)})^\dagger U^{(2)}=\rho$ then requires the relative phase between $|\alpha\rangle$ and $|\beta\rangle$ to remain unchanged. Hence $\theta_{(1),(2),\alpha}=\theta_{(1),(2),\beta}$ modulo $2\pi$. Since $\alpha$ and $\beta$ are arbitrary, the phase is independent of the MES label, and therefore
$U^{(1)}|\alpha\rangle=U^{(2)}e^{i\theta_{(1),(2)}}|\alpha\rangle$
for every $\alpha \in \mathfrak{C}_{x(y)}^{\tau_{y(x)}}$.

The ``if'' direction follows directly. By Proposition~\ref{pro:rho}, $\rho$ can be expressed as
\begin{align}
\rho=\sum_{i} q_i |\psi_i\rangle \langle\psi_i|,
\end{align}
where $|\psi_i\rangle=\sum_{\alpha \in \mathfrak{C}_{x(y)}^{\tau_{y(x)}}} c_\alpha^i |\alpha \rangle$, and $\{q_i\}$ is a probability distribution. If $U^{(1)}|\alpha\rangle=U^{(2)}e^{i\theta_{(1),(2)}}|\alpha\rangle$ for every $\alpha \in \mathfrak{C}_{x(y)}^{\tau_{y(x)}}$, with a phase independent of $\alpha$, then $(U^{(2)})^\dagger U^{(1)}|\psi_i\rangle=e^{i\theta_{(1),(2)}}|\psi_i\rangle$
for every $i$. Consequently, $(U^{(2)})^\dagger U^{(1)} \rho (U^{(1)})^\dagger U^{(2)}=\rho$, which proves that $U^{(1)}$ and $U^{(2)}$ are equivalent.
\end{proof}
\end{theorem}

\noindent This theorem is particularly useful because it allows us to characterize the equivalence of symmetries entirely in terms of a basis of pure MESs.

\begin{definition}[Total charge transported through a noncontractible annulus]
Let $U$ be a symmetry with a fixed realization $\mathcal P_U$, let $\Gamma_{\mathcal P_U}$ be the spatial transport graph of $\mathcal P_U$, and let $X\in\mathtt A_{y(x)}$. Fix a noncontractible loop $\kappa_X\subset X\setminus\partial X$ topologically equivalent to $l_{y(x)}$, together with a choice of positive transverse direction locally agreeing with the orientation of $l_{x(y)}$. We require that $\kappa_X$ intersect the transport strands transversely and avoid all vertices and charge-transmutation points. The loop $\kappa_X$ is fixed as part of the data associated with $X$ and is suppressed from the notation below.

Let $p_1,\ldots,p_r$ denote the crossing events at which the individual charge-labeled transport strands of $\mathcal P_U$ cross $\kappa_X$, with each traversal counted separately. The local oriented strand passing through $p_j$ is called a through-strand. Let $a_j$ be its charge at the crossing, evaluated after any transmutation caused by transparent domain walls encountered before that crossing, and define
\begin{equation*}
\widetilde a_j:=
\begin{cases}
a_j,
& \text{if the through-strand crosses $\kappa_X$}\\ &\text{in the positive transverse direction},\\
\bar{a}_j,
& \text{if it crosses $\kappa_X$ in the negative}\\ &\text{transverse direction}.
\end{cases}
\end{equation*}

The total charge of $\mathcal P_U$ transported through $X$ is
defined by
\begin{equation*}
a_{x(y)}(\mathcal P_U;X)
:=
\widetilde a_1
\times\cdots\times
\widetilde a_r .
\label{eq:transported-charge}
\end{equation*}
When there are no crossings, the product is defined to be the vacuum charge $1$. Because the anyon theory is Abelian, the ordering of the crossing events does not affect the result.
\end{definition}

The total charge depends on the chosen realization and loop $\kappa_X$; moving the cut across a wall or a vertex can change its value. In contrast, the induced fusion action on the MES labels, defined below, is the invariant quantity relevant to the ICS.
\begin{lemma}[Action of transport on MESs]
\label{lem:transport-action}
Let $U$ be a symmetry with a fixed realization $\mathcal P_U$. Choose $X\in\mathtt A_{y(x)}$ to be a sufficiently narrow annular neighborhood of $\kappa_X$, disjoint from all process vertices and from the excitation endpoints occurring between elementary operations. Thus, every crossing of $\kappa_X$ restricts to a strand joining the two boundary components of $X$. For every $\alpha\in\mathfrak C_{y(x)}^{\tau_{x(y)}}$, the label
\begin{equation}
\beta= a_{x(y)}(\mathcal P_U;X)\times\alpha
\end{equation}
belongs to
$\mathfrak C_{y(x)}^{\tau_{x(y)}}$ and satisfies
\begin{equation}
U|\alpha\rangle\langle\alpha|U^\dagger
=
|\beta\rangle\langle\beta|.
\end{equation}

\begin{proof}
Refine the description of $\mathcal P_U$, if necessary, so
that every transport strand crossing $X$ defines a complete
tunneling step from one boundary component of $X$ to the other.
Immediately before and after each such step, $X$ contains no
localized excitation.

Order the crossing events according to the temporal order of
the refined process and set
\begin{equation}
\alpha_0=\alpha,
\qquad
\alpha_j
=
\widetilde a_j\times\alpha_{j-1},
\qquad
j=1,\ldots,r.
\end{equation}
For a positive crossing, Proposition~\ref{pro:fusion} implies that tunneling
the charge $a_j$ through $X$ changes the extreme point by
fusion with $a_j$. For a negative crossing, the inverse
tunneling process changes the sector by fusion with
$\bar{a}_j$. Thus, in either case, the sector changes by
fusion with the oriented charge $\widetilde a_j$. Consequently,
the successive extreme points are
\begin{equation}
\rho_X^{\alpha_0}
\longmapsto
\rho_X^{\alpha_1}
\longmapsto
\cdots
\longmapsto
\rho_X^{\alpha_r}.
\end{equation}

Elementary operations whose strands do not cross $\kappa_X$ preserve its extreme point label. Closed Abelian braiding processes may contribute a sector-dependent phase, but this phase disappears at the density-operator level. Therefore,
\begin{align}
\alpha_r
&=
\left(
\widetilde a_1
\times\cdots\times
\widetilde a_r
\right)\times\alpha
=
a_{x(y)}(\mathcal P_U;X)\times\alpha
=
\beta.
\end{align}

Since $\mathcal P_U$ terminates without leaving any localized excitation and $U$ is a symmetry, $U|\alpha\rangle\langle\alpha|U^\dagger$ is a pure element of $\Sigma(\mathtt R)$ whose reduced density matrix on $X$ is the extreme point $\rho_X^\beta$. It is therefore an MES associated with $X$, so $\beta \in \mathfrak C_{y(x)}^{\tau_{x(y)}}$. Proposition~\ref{pro:oto} then gives the unique MES projector carrying this label:
\begin{equation}
U|\alpha\rangle\langle\alpha|U^\dagger
=
|\beta\rangle\langle\beta|.
\end{equation}
\end{proof}
\end{lemma}

The preceding lemma shows that the action of a symmetry on the MES labels is determined by fusion with the total charge transported through the annulus. We now define the effective group of these annular actions.

\begin{definition}[Effective symmetry group on a noncontractible annulus]
\label{def:effective-annular-group}
Let $F_{\mathrm{total}}$ be the group of equivalence classes of
symmetries. For $X\in\mathtt A_{x(y)}$, define
\begin{equation*}
    \pi_{x(y)}:
    F_{\mathrm{total}}
    \longrightarrow
    {\rm Perm}(\mathfrak C_{x(y)}^{\tau_{y(x)}})
\end{equation*}
by
\begin{equation*}
    \pi_{x(y)}([U])(\alpha)=\beta
    \quad\Longleftrightarrow\quad
    U|\alpha\rangle\langle\alpha|U^\dagger
    =
    |\beta\rangle\langle\beta|.
    \label{eq:effective-annular-action}
\end{equation*}
Here, $U$ is any representative of $[U]$, and ${\rm Perm}(\mathfrak C_{x(y)}^{\tau_{y(x)}})$ denotes the permutation group of $\mathfrak C_{x(y)}^{\tau_{y(x)}}$. This definition is independent of the choice of representative because equivalent symmetries induce the same conjugation action on every state in $\Sigma(R)$.

The map $\pi_{x(y)}$ is a group homomorphism:
\begin{equation*}
    \pi_{x(y)}([U][V])
    =
    \pi_{x(y)}([U])
    \circ
    \pi_{x(y)}([V]).
\end{equation*}
We define the effective symmetry group associated with
$X\in\mathtt A_{x(y)}$ by
\begin{equation*}
    F_{x(y)}
    :=
    \operatorname{Im}\pi_{x(y)}
    \subseteq
    {\rm Perm}(\mathfrak C_{x(y)}^{\tau_{y(x)}}).
    \label{eq:effective-annular-group}
\end{equation*}
\end{definition}

\noindent By the first isomorphism theorem, $F_{x(y)}\simeq F_{\mathrm{total}}/\ker\pi_{x(y)}$. Thus, two elements of $F_{\mathrm{total}}$ determine the same element of $F_{x(y)}$ if and only if they induce the same permutation of the MES labels.

\begin{theorem} \label{th:uni}
The elements of $F_{x(y)}$ are precisely the permutations of
$\mathfrak C_{x(y)}^{\tau_{y(x)}}$ induced by fusion with
Abelian anyons. If $a\times\alpha_0=\beta$ for some
$\alpha_0,\beta \in\mathfrak C_{x(y)}^{\tau_{y(x)}}$, then fusion by $a$ preserves the entire set
$\mathfrak C_{x(y)}^{\tau_{y(x)}}$ and is realized by an element of $F_{x(y)}$.

\begin{proof}
Let $g\in F_{x(y)}$. Choose
$[U]\in F_{\mathrm{total}}$ satisfying $\pi_{x(y)}([U])=g$, together with a representative $U$, its fixed realization $\mathcal P_U$, and $X \in \mathtt{A}_{x(y)}$. By Lemma~\ref{lem:transport-action},
\begin{equation}
g(\alpha)= a_{y(x)}(\mathcal P_U;X)\times\alpha
\end{equation}
for every $\alpha\in\mathfrak C_{x(y)}^{\tau_{y(x)}}$.
Hence, every element of $F_{x(y)}$ is induced by fusion with
an Abelian anyon.

Conversely, suppose that $a\times\alpha_0=\beta$, where $\alpha_0,\beta \in \mathfrak C_{x(y)}^{\tau_{y(x)}}$. We prove the statement for $X\in\mathtt A_y$ and write $a=a_x$. The case $X\in\mathtt A_x$ follows by exchanging $x$ and $y$.

Create an $a_x$--$\bar{a}_x$ pair and transport $a_x$
once around $l_x$. When it returns to its initial position,
the remaining localized charge is
\begin{equation}
    b
    \equiv 
    \phi_x(a_x)\times\bar{a}_x.
    \label{eq:residual-charge}
\end{equation}
If $b=1$, the closed transport already leaves no localized
excitation.

Suppose that $b\neq1$. Since $\alpha_0$ and $\beta$ are invariant under $\tau_x$, Theorem~\ref{th:NNNtau} implies
\begin{equation}
\phi_x(a_x)\times\alpha_0 =a_x\times\alpha_0=\beta.
\end{equation}
By Proposition~\ref{pro:sim}, $a_x \overset{y}{\sim}  \phi_x(a_x)$, and therefore $1\overset{y}{\sim}b$. Porism~\ref{porism:equiv2} then gives an anyon $c\in\mathcal C$ such that
\begin{equation}
    b=\phi_y(c)\times\bar c.
\end{equation}
Split $b$ into $\phi_y(c)$ and $\bar c$, and transport
$\bar c$ once around $l_y$ along a path disjoint from
$X$. After the transport,
\begin{equation}
    \bar c
    \longmapsto
    \phi_y(\bar c)
    =
    \overline{\phi_y(c)},
\end{equation}
which annihilates with $\phi_y(c)$. The complete process
therefore leaves no localized excitation, while its total
charge transported through $X$ is $a_x$.

By Proposition~\ref{pro:anyon-transport}, the elementary
operators can be chosen independently of the input state. Their
composition defines a single unitary $U$ satisfying
\begin{equation}
    U\Sigma(\mathtt R)U^\dagger
    \subseteq
    \Sigma(\mathtt R).
\end{equation}
Applying the reversed process gives
\begin{equation}
    U^\dagger\Sigma(\mathtt R)U
    \subseteq
    \Sigma(\mathtt R),
\end{equation}
and hence
\begin{equation}
    U\Sigma(\mathtt R)U^\dagger
    =
    \Sigma(\mathtt R).
\end{equation}
Thus, $U$ is a symmetry.

Finally, Lemma~\ref{lem:transport-action} gives
\begin{equation}
    \pi_y([U])(\alpha)
    =
    a_x\times\alpha
\end{equation}
for every $\alpha\in\mathfrak C_y^{\tau_x}$. Therefore, fusion by $a_x$ preserves all of $\mathfrak C_y^{\tau_x}$ and defines an element of $F_y$.
\end{proof}

\end{theorem}

\begin{corollary}\label{cor:sym}
For every $g \in F_{x(y)}$, there exists an anyon-transport
symmetry $U \in \mathcal{G}(\mathtt{R})$ such that
\[
    \pi_{x(y)}([U]) = g,
    \qquad
    \operatorname{supp}(U)
    \subseteq N_w(l_x \cup l_y),
\]
where $N_w(S)$ denotes the neighborhood of $S$ of
microscopic width $w$, independent of the system size.
\end{corollary}

According to Theorem~\ref{th:uni} and Corollary~\ref{cor:sym}, the symmetries can be associated with anyon transport along both fundamental cycles. The following proposition gives a criterion for when this two-cycle structure cannot be decomposed into a product of independent $1$-form symmetry operators supported on the individual cycles. We will show that this obstruction is associated with the presence of transparent domain walls. In the toric code without transparent domain walls, the $1$-form symmetry operators can always be factorized into separate $1$-form symmetry operators winding along the $x$- and $y$-cycles. In Wen's plaquette model, however, transparent domain walls arise when the system has odd linear dimensions, and certain $1$-form symmetry operators do not admit such a factorization. We analyze this phenomenon in detail in Sec.~\ref{sec:5a}. For now, the following proposition makes precise the criterion underlying this nonfactorizable structure.

\begin{proposition}[Obstruction to fundamental-cycle factorization]
\label{pro:symd}
Let $g\in F_{y(x)}$ be the permutation induced by fusion with an anyon $a_{x(y)}$, so that
\begin{equation}
    g(\alpha)=a_{x(y)}\times\alpha
\end{equation}
for every $\alpha\in\mathfrak C_{y(x)}^{\tau_{x(y)}}$. Suppose that there is no $d\in\mathcal C^{\phi_{x(y)}}$ satisfying $d\overset{y(x)}{\sim}a_{x(y)}$. Then no symmetry class $[U]\in F_{\mathrm{total}}$ satisfying $\pi_{y(x)}([U])=g$ admits a representative of either form $U_xU_y$ or $U_yU_x$, where $U_x$ and $U_y$ are $1$-form symmetry operators supported on $l_x$ and $l_y$, respectively.
\end{proposition}

\begin{proof}
The statement obtained by exchanging $x$ and $y$ holds analogously. Suppose that such a factorization exists. Since $U_x$ is a
$1$-form symmetry supported on $l_x$, it is represented by
the closed transport of some $d\in\mathcal C^{\phi_x}$.
Its action on the MES labels associated with
$X\in\mathtt A_y$ is therefore
\begin{equation}
    \pi_y([U_x])(\alpha)=d\times\alpha.
\end{equation}
In contrast, $U_y$ transports no charge through $X$ and hence
does not permute these MES labels:
\begin{equation}
    \pi_y([U_y])
    =
    \operatorname{id}.
\end{equation}
Because $\pi_y$ is a homomorphism, either ordering of $U_x$
and $U_y$ induces fusion by $d$. Therefore,
\begin{equation}
    d\times\alpha
    =
    g(\alpha)
    =
    a_x\times\alpha
\end{equation}
for every $\alpha\in\mathfrak C_y^{\tau_x}$. Proposition~\ref{pro:sim} then implies $d\mathrel{\stackrel{y}{\sim}}a_x$, leading to a contradiction.
\end{proof}

The hypothesis of Proposition~\ref{pro:symd} automatically
implies $\phi_x(a_x)\times\bar{a}_x\neq 1$. Indeed, if $\phi_x(a_x)=a_x$, then $a_x\in\mathcal C^{\phi_x}$ itself would satisfy $a_x\mathrel{\stackrel{y}{\sim}}a_x$.

We make two remarks. First, Proposition~\ref{pro:symd} concerns factorization relative to the chosen fundamental cycles $l_x$ and $l_y$. It does not exclude a representative supported on a loop in another homology class or a more general product of $1$-form symmetry operators. Explicit examples are discussed in Secs.~\ref{sec:5a} and~\ref{sec:5c}.

Second, understanding the ground state degeneracy through these symmetries is less straightforward, because they cannot, in general, be decomposed into a product of two $1$-form symmetries supported individually on $l_x$ and $l_y$. However, we can clarify this structure by examining their action on the extreme points of the ICSs on noncontractible annuli. Choose a noncontractible annulus in $\mathtt{A}_x$ and consider the extreme points that arise as reductions of extreme points of the ICS on the torus. An analogous construction applies to a noncontractible annulus in $\mathtt{A}_y$. These extreme points correspond one-to-one to the MESs forming a basis for the ground state space. Symmetries supported along $l_x$ leave these extreme points unchanged and therefore act diagonally on the corresponding MESs, contributing only phases. Symmetries supported along $l_y$ or jointly on $l_x$ and $l_y$ can permute these extreme points and thereby connect distinct MESs. Thus, a chosen MES is invariant up to a phase under symmetries that fix its associated extreme point, whereas symmetries that map it to a different extreme point transform it into a distinct MES. These connections between ICSs on noncontractible annuli, MESs, and symmetry actions provide a constructive description of how symmetries are realized in the ground state space and a basis for discussing spontaneous symmetry breaking in the presence of transparent domain walls. In Sec.~\ref{sec:5a}, we illustrate this connection using Wen's plaquette model.

We now prove the following proposition, which will be useful for defining entanglement asymmetry.
\begin{proposition} \label{pro:Fx}
The natural action of $F_{x(y)}$ on $\mathfrak C_{x(y)}^{\tau_{y(x)}}$ is free and transitive. Consequently,
\begin{equation}
    |F_{x(y)}|
    =
    \left|
        \mathfrak C_{x(y)}^{\tau_{y(x)}}
    \right|.
    \label{eq:effective-group-cardinality}
\end{equation}
\end{proposition}

\begin{proof}
Fix $\alpha_0 \in \mathfrak C_{x(y)}^{\tau_{y(x)}}$
and consider the evaluation map
\begin{equation}
    \operatorname{ev}_{\alpha_0}:
    F_{x(y)}
    \longrightarrow
    \mathfrak C_{x(y)}^{\tau_{y(x)}},
    \qquad
    g\longmapsto g(\alpha_0).
\end{equation}

For every $\beta \in \mathfrak C_{x(y)}^{\tau_{y(x)}}$, Proposition~\ref{pro:aalphabeta} provides an anyon $a$ satisfying
\begin{equation}
    a\times\alpha_0=\beta.
\end{equation}
By Theorem~\ref{th:uni}, fusion by $a$ is realized by an
element $g\in F_{x(y)}$. Therefore,
$\operatorname{ev}_{\alpha_0}$ is surjective.

Now suppose that $g,h\in F_{x(y)}$ satisfy
\begin{equation}
    g(\alpha_0)=h(\alpha_0).
\end{equation}
By Theorem~\ref{th:uni}, there exist anyons $a,d\in\mathcal C$
such that
\begin{equation}
    g(\alpha)=a\times\alpha,
    \qquad
    h(\alpha)=d\times\alpha
\end{equation}
for every
$\alpha\in\mathfrak C_{x(y)}^{\tau_{y(x)}}$. In particular,
\begin{equation}
    a\times\alpha_0
    =
    d\times\alpha_0.
\end{equation}
Proposition~\ref{pro:sim} implies that fusion by $a$ and fusion by $d$
have the same action on every element of
$\mathfrak C_{x(y)}^{\tau_{y(x)}}$. Hence, $g=h$, and
$\operatorname{ev}_{\alpha_0}$ is injective.

Thus, $\operatorname{ev}_{\alpha_0}$ is a bijection. This
proves that the action is free and transitive and establishes \eqref{eq:effective-group-cardinality}.
\end{proof}

\subsection{Entanglement asymmetry}
\label{sec:4d}

The \textit{entanglement asymmetry} was originally introduced as a measure of how much a quantum state locally breaks a global symmetry \cite{ares_asymmetry_2023}. This measure has recently been generalized to broader classes of symmetries, including $p$-form and non-invertible symmetries~\cite{benini_asymmetry_2025,lamas_higher-form_2025,benini_higher-form_2026}. In particular, the entanglement asymmetry can capture the spontaneous breaking of 1-form symmetries in topological phases~\cite{lamas_higher-form_2025,benini_asymmetry_2025}, and noncontractible annuli provide a natural setting for its definition, since the construction is invariant under topology-preserving deformations.

In the present setting, Proposition~\ref{pro:symd} shows that the relevant symmetries need not factor into independent $1$-form symmetries on $l_x$ and $l_y$. Nevertheless, apart from relative phases between MESs, their actions are fully captured by permutations of the MESs associated with a noncontractible annulus. We therefore formulate the entanglement asymmetry in terms of the canonical effective permutation group $F_{x(y)}$. This formulation avoids choosing a noncanonical set of microscopic unitary representatives and ensures that the resulting quantity depends only on the symmetry action visible to the ICS on noncontractible annuli. Before defining the entanglement asymmetry, we first establish several useful propositions and introduce the necessary definitions.

\begin{proposition}\label{pro:rhox}
An arbitrary $\rho_X$ with $X \in \mathtt{A}_{x(y)}$, obtained as the reduced density matrix on $X$ of $\rho \in \Sigma(\mathtt{R})$, can be uniquely expressed as
\begin{equation}
\label{eq:rho_x_MES_ICS}
\rho_X = \sum_{\alpha \in \mathfrak{C}_{x(y)}^{\tau_{y(x)}}} p_\alpha \rho_X^\alpha,
\end{equation}
where $\{p_\alpha\}$ is a probability distribution.
\begin{proof}
By Proposition~\ref{pro:rho}, $\rho$ can be expressed as
\begin{align}
\rho=\sum_{i} q_i |\psi_i\rangle \langle\psi_i|,
\end{align}
where $|\psi_i\rangle=\sum_{\alpha \in \mathfrak{C}_{x(y)}^{\tau_{y(x)}}} c_\alpha^i |\alpha \rangle$, and $\{q_i\}$ is a probability distribution. Here, each $|\alpha\rangle$ is normalized, and the coefficients satisfy
\begin{align}
\sum_{\alpha \in \mathfrak{C}_{x(y)}^{\tau_{y(x)}}} 
|c_\alpha^i|^2=1 .
\end{align}

Then,
\begin{align}
\rho_X&={\rm Tr}_{\mathtt{R} \setminus X} (\rho)\nn
&=\sum_i \sum_{\alpha,\beta \in \mathfrak{C}_{x(y)}^{\tau_{y(x)}}} q_i c_\alpha^i (c_\beta^i)^* {\rm Tr}_{\mathtt{R} \setminus X} (|\alpha \rangle \langle \beta|)\nn
&=\sum_i\sum_{\alpha \in \mathfrak{C}_{x(y)}^{\tau_{y(x)}}} q_i |c_\alpha^i|^2 {\rm Tr}_{\mathtt{R} \setminus X} (|\alpha \rangle \langle \alpha|)\nn
&=\sum_i\sum_{\alpha \in \mathfrak{C}_{x(y)}^{\tau_{y(x)}}} q_i |c_\alpha^i|^2 \rho_X^\alpha\nn
&=\sum_{\alpha \in \mathfrak{C}_{x(y)}^{\tau_{y(x)}}} p_\alpha \rho_X^\alpha,
\end{align}
where $p_\alpha =\sum_i q_i |c_\alpha^i|^2 \geq 0$. Moreover,
\begin{equation}
    \sum_{\alpha}p_\alpha
    =
    \sum_iq_i\sum_{\alpha}|c_\alpha^i|^2
    =
    1.
\end{equation}
The decomposition is unique because the states $\rho_X^\alpha$ have mutually orthogonal supports.

To eliminate the off-diagonal terms, we use
\begin{equation}
    {\rm Tr}_{\mathtt{R} \setminus X} (|\alpha \rangle \langle \beta|)={\rm Tr}_{\mathtt{R} \setminus X} (|\alpha \rangle \langle \alpha|)\delta_{\alpha, \beta}.
\end{equation}
This equation follows from the property of MESs that there exists a projector $\hat{P}^\alpha$ supported on the annulus $\mathtt{R} \setminus X$ such that $\hat{P}^\alpha |\beta\rangle = \delta_{\alpha,\beta}|\beta\rangle$. This property, in turn, follows from the fact that the reduced density matrices of distinct MESs are extreme points of the simplex $\Sigma(\mathtt{R}\setminus X)$ with mutually orthogonal supports (Definition~\ref{def:MES}). Then, ${\rm Tr}_{\mathtt{R} \setminus X} (|\alpha \rangle \langle \beta|) = {\rm Tr}_{\mathtt{R} \setminus X} (|\alpha \rangle \langle \beta|\hat{P}^\beta)={\rm Tr}_{\mathtt{R} \setminus X} ( \hat{P}^\beta|\alpha \rangle \langle \beta|)=\delta_{\alpha,\beta}{\rm Tr}_{\mathtt{R} \setminus X}(|\alpha \rangle \langle \alpha|)$ as claimed.
\end{proof}

\end{proposition}
\noindent Thus, $\rho_X$ belongs to the convex hull of the extreme
points labeled by the elements of $\mathfrak{C}_{x(y)}^{\tau_{y(x)}}$.

\begin{definition}\label{def:grhox}
Let $X\in\mathtt{A}_{x(y)}$ and let
\begin{equation*}
    \rho_X
    =
    \Tr_{\mathtt{R}\setminus X}\rho,
    \qquad
    \rho\in\Sigma(\mathtt{R}).
\end{equation*}
By Proposition~\ref{pro:rhox}, $\rho_X$ has the unique decomposition
\begin{equation*}
    \rho_X
    =
    \sum_{\alpha\in
    \mathfrak{C}_{x(y)}^{\tau_{y(x)}}}
    p_\alpha\rho_X^\alpha,
\end{equation*}
where $\{p_\alpha\}$ is a probability distribution. For
$g\in F_{x(y)}$, we define
\begin{equation*}
g\cdot\rho_X=\sum_{\alpha \in \mathfrak{C}_{x(y)}^{\tau_{y(x)}}} p_\alpha\rho_X^{g(\alpha)}.
\end{equation*}
\end{definition}

\noindent Then, we can show the following proposition.

\begin{proposition}
Let $\rho\in\Sigma(\mathtt{R})$ and let $X\in\mathtt{A}_{x(y)}$. For every $g\in F_{x(y)}$ and every symmetry $U_g\in\mathcal{G}(\mathtt{R})$ satisfying
\begin{equation}
    \pi_{x(y)}([U_g])=g,
\end{equation}
we have
\begin{equation}
    g\cdot\rho_X
    =
    \Tr_{\mathtt{R}\setminus X}
    \left(U_g\rho U_g^\dagger\right).
\end{equation}
In particular, the reduced state on the right-hand side depends on only $g$ and $\rho_X$, and is independent of the choice of representative $U_g$.
    
\begin{proof}
The state $\rho$ can be expanded in the MES basis as
\begin{equation}
\rho=\sum_{\alpha,\beta \in \mathfrak{C}_{x(y)}^{\tau_{y(x)}}} c_{\alpha\beta}|\alpha\rangle\langle\beta|,
\end{equation}
with $c_{\alpha\alpha}=p_\alpha$.
Since $\pi_{x(y)}([U_g])=g$, we have
$U_g|\alpha\rangle=e^{i\theta_\alpha}|g(\alpha)\rangle$.
Using $\Tr_{\mathtt{R}\setminus X} \left(|\alpha\rangle\langle\beta|\right)=\delta_{\alpha,\beta}\rho_X^\alpha$, as used in Proposition~\ref{pro:rhox}, together with the fact that $g$ is a permutation, we obtain
\begin{align}
    \Tr_{\mathtt{R}\setminus X}
    \left(U_g\rho U_g^\dagger\right)
    &=
    \sum_{\alpha}
    c_{\alpha\alpha}\rho_X^{g(\alpha)}
    \nn
    &=
    \sum_{\alpha}
    p_\alpha\rho_X^{g(\alpha)}
    =
    g\cdot\rho_X.
\end{align}
The final expression depends only on $g$ and $\rho_X$ and is therefore independent of the choice of representative $U_g$.
\end{proof}
\end{proposition}

We now define the entanglement asymmetry as follows.

\begin{definition}[Entanglement asymmetry]
\label{def:entanglement-asymmetry}
For $\rho\in\Sigma(\mathtt R)$ and
$X\in\mathtt{A}_{x(y)}$, define
\begin{equation*}
    \rho_X^{\mathrm{sym}}
    \equiv 
    \frac{1}{|F_{x(y)}|}
    \sum_{g\in F_{x(y)}}g\cdot\rho_X.
\label{eq:symmetrized-reduced-state}
\end{equation*}
The entanglement asymmetry is~\footnote{Note that the definition of $\rho_X^{\mathrm{sym}}$ above differs slightly from that used in Refs.~\onlinecite{ares_asymmetry_2023,lamas_higher-form_2025}, where $\rho_X^{\mathrm{sym}}$ is defined as
\begin{align*}
\rho_X^{\mathrm{sym}}= \frac{1}{|F_{x(y)}|} \sum_{g \in F_{x(y)}} U_{g|X} \rho_X U_{g|X}^\dagger,
\end{align*}
with the symmetry operator $U_g$ assumed to factorize as $U_g = U_{g|X} U_{g|\mathtt{R}\setminus X}$, where $U_{g|X}$ and $U_{g|\mathtt{R}\setminus X}$ denote the restrictions of $U_g$ to $X$ and $\mathtt{R}\setminus X$, respectively. In our setting, however, such a factorization need not always exist, since $U_g$ may be a non-onsite unitary operator. Definition~\ref{def:entanglement-asymmetry} resolves this issue by remaining well defined even for non-onsite operators $U$. Moreover, whenever $U_g$ is factorizable, our definition reduces to the standard one used in the literature.}
\begin{equation*}
    \Delta S_{X_{x(y)}}[\rho]
    \equiv 
    S(\rho_X^{\mathrm{sym}})-S(\rho_X) .
\label{eq:entanglement-asymmetry}
\end{equation*}
\end{definition}

\noindent We can now describe the geometric interpretation of the entanglement asymmetry within the information convex picture. To make this precise, we first establish the following propositions.

\begin{proposition}\label{pro:rhoxsym}
For $X \in \mathtt{A}_{x(y)}$, the symmetrized density matrix $\rho_X^{\rm sym}$ can be written as
\begin{align}
\rho_X^{\rm sym}=\frac{1}{|F_{x(y)}|}\sum_{\alpha \in \mathfrak{C}_{x(y)}^{\tau_{y(x)}}} \rho_X^\alpha,
\label{eq:rho_sym_MES}
\end{align}
for an arbitrary $\rho \in \Sigma(\mathtt{R})$.
\begin{proof}
By Proposition~\ref{pro:Fx}, the action of $F_{x(y)}$ on $\mathfrak C_{x(y)}^{\tau_{y(x)}}$ is free and transitive. Therefore, for each fixed $\alpha$, the map $g\mapsto g(\alpha)$ is a bijection from $F_{x(y)}$ to $\mathfrak C_{x(y)}^{\tau_{y(x)}}$. By Definition~\ref{def:grhox},
\begin{align}
\rho_X^{\mathrm{sym}} &= \frac{1}{|F_{x(y)}|} \sum_{g\in F_{x(y)}} \sum_{\alpha\in\mathfrak C_{x(y)}^{\tau_{y(x)}}} p_\alpha\rho_X^{g(\alpha)} \nn
&=\frac{1}{|F_{x(y)}|} \sum_{\beta\in\mathfrak C_{x(y)}^{\tau_{y(x)}}}p_\beta \sum_{\alpha\in\mathfrak C_{x(y)}^{\tau_{y(x)}}}\rho_X^\alpha\nn
&= \frac{1}{|F_{x(y)}|} \sum_{\alpha\in\mathfrak C_{x(y)}^{\tau_{y(x)}}}\rho_X^\alpha.
\end{align}
\end{proof}
\end{proposition}
\noindent Thus, $\rho_X^{\mathrm{sym}}$ is the uniform barycenter of the convex hull of the extreme points labeled by $\mathfrak C_{x(y)}^{\tau_{y(x)}}$. Hence,
$\Delta S_{X_{x(y)}}[\rho]$ is the difference between the von Neumann entropy of this uniform barycenter and that of $\rho_X$.

\begin{theorem} [Entanglement asymmetry formula] \label{th:ea}
Let $\rho_X = \sum_{\alpha \in \mathfrak{C}_{x(y)}^{\tau_{y(x)}}} p_\alpha \rho_X^\alpha$. The entanglement asymmetry of $\rho$ on $X \in \mathtt{A}_{x(y)}$ is given by
\begin{equation}
\label{eq:asymmetry_general}
\Delta S_{X_{x(y)}}[\rho]=\log \bigl|F_{x(y)}\bigr|+ \sum_{\alpha \in \mathfrak{C}_{x(y)}^{\tau_{y(x)}}} p_\alpha \log p_\alpha,
\end{equation}
and is bounded above by
\begin{equation}
\Delta S_{X_{x(y)}}^{\max}
=
\log \bigl|F_{x(y)}\bigr|,
\end{equation}
with the maximum attained when $\rho_X=\rho_X^\alpha$ for some $\alpha \in \mathfrak{C}_{x(y)}^{\tau_{y(x)}}$.
\begin{proof}
Using Propositions~\ref{pro:rhox} and~\ref{pro:rhoxsym}, together with Lemma~\ref{lem:sameQD}, we obtain that
\begin{align}
S(\rho_X)&=-\sum_{\alpha \in \mathfrak{C}_{x(y)}^{\tau_{y(x)}}} p_\alpha\log p_\alpha+S(\rho_X^{\alpha'}),\nn
S(\rho_X^{{\rm sym}})&=\log |F_{x(y)}|+S(\rho_X^{\alpha'}),
\end{align}
where $\alpha' \in \mathfrak{C}_{x(y)}^{\tau_{y(x)}}$. Therefore, the  entanglement asymmetry is given by
\begin{equation}
\Delta S_{X_{x(y)}}[\rho]=\log \bigl|F_{x(y)}\bigr|+ \sum_{\alpha \in \mathfrak{C}_{x(y)}^{\tau_{y(x)}}} p_\alpha \log p_\alpha,
\end{equation} which, in the absence of transparent domain walls, matches the comparable result in Ref.~\onlinecite{lamas_higher-form_2025}.
Since $\sum_{\alpha \in \mathfrak{C}_{x(y)}^{\tau_{y(x)}}} p_\alpha\log p_\alpha \leq 0$, the entanglement asymmetry is bounded above by $\log \bigl|F_{x(y)}\bigr|$, and this maximum is attained for the extreme points labeled by $\mathfrak{C}_{x(y)}^{\tau_{y(x)}}$.
\end{proof}
\end{theorem}

\begin{proposition} \label{pro:maximal-entanglement-asymmetry}
Let $X\in\mathtt{A}_x$ and $X'\in \mathtt{A}_y$. Then
\begin{equation}
\Delta S_{X_x}^{\rm max}=\Delta S_{X'_y}^{\rm max}= \log N(\mathtt R).
\label{eq:maximal-entanglement-asymmetry}
\end{equation}
Using the established identification $N(\mathtt R)=\mathrm{GSD}$,
both maxima equal $\log\mathrm{GSD}$.

\begin{proof}
By Theorem~\ref{th:ea} and Proposition~\ref{pro:Fx}, we have
\begin{align}
\Delta S_{X_x}^{\rm max}
&=
\log\left|\mathfrak C_x^{\tau_y}\right|,
&
\Delta S_{X'_y}^{\rm max}
&=
\log\left|\mathfrak C_y^{\tau_x}\right|.
\end{align}
Proposition~\ref{pro:GSD} further implies
$\left|\mathfrak C_x^{\tau_y}\right|=\left|\mathfrak C_y^{\tau_x}\right|=N(\mathtt R)$.
Therefore, \eqref{eq:maximal-entanglement-asymmetry} follows.
\end{proof}
\end{proposition}

We can conclude that $\Delta S_{X_{x(y)}}^{\rm max}$ reflects the presence of transparent domain walls along both $l_x$ and $l_y$. Indeed, according to Proposition~\ref{pro:Ctau}, $|\mathfrak{C}_{x(y)}^{\tau_{y(x)}}|$ gives the number of extreme points of the ICS on $X \in \mathtt{A}_{x(y)}$ that remain invariant under the automorphism $\tau_{y(x)}$. The total number of extreme points of the ICS on $X \in \mathtt{A}_{x(y)}$ reflects the presence of a transparent domain wall crossing the subregion $X$, while the automorphism $\tau_{y(x)}$ encodes the presence of a transparent domain wall parallel to the boundary of $X$. Therefore, $\Delta S_{X_{x(y)}}^{\rm max}$ is sensitive to the presence of transparent domain walls along both $l_x$ and $l_y$. In contrast, the TEE can generically obey $\gamma_{\text{LW},x}^{\rm max}\neq \gamma_{\text{LW},y}^{\rm max}$, as it captures the constraints on the entanglement entropy of a noncontractible subregion in each direction independently.

In the absence of transparent domain walls, $\gamma_{\text{LW},x}^{\rm max}=\gamma_{\text{LW},y}^{\rm max}=\log \mathcal{D}$ by Theorem~\ref{th:TEE} and $\Delta S_{X_{x(y)}}^{\max}=2\log \mathcal{D}$ by Theorem~\ref{th:ea}. Thus, after accounting for this factor of two, a genuine distinction between the maximal TEE and the maximal entanglement asymmetry becomes manifest only in the presence of transparent domain walls.
 
In the next section, we show examples where $\gamma_{\text{LW},x}^{\rm max}\neq \gamma_{\text{LW},y}^{\rm max}$, while the maximal entanglement asymmetry is the same in both directions.

\section{Lattice model Examples}
\label{sec:5}

In this section, we apply the general results established above to determine the maximal values of the TEE and the entanglement asymmetry, as well as the symmetry structure of the reference state, taken to be the ground state of the parent lattice model. The required inputs are the anyon types, the GSD, and the automorphisms associated with transparent domain walls. This information can be obtained from the reference state and its parent Hamiltonian; the derivation of the automorphisms is given in Appendix~\ref{asec:new}.

We consider three explicit parent lattice models: Wen's plaquette model~\cite{wen_quantum_2003, you_projective_2012}, the anisotropic dipolar toric code~\cite{ebisu_anisotropic_2023}, and the rank-2 toric code~\cite{oh22a,pace-wen,oh22b,oh23,kim_unveiling_2025,kim_gauging_2026}. Heuristically, their ground states satisfy $A_0$ and $A_1$ because these axioms are local, and each model is locally equivalent to one or more copies of the toric code on sufficiently small regions. As a concrete example, we consider the anisotropic dipolar toric code and rigorously show that its ground state satisfies ${\bf A0}$ and ${\bf A1}$ in Appendix~\ref{asec:adtc}. We expect this framework to apply more broadly to topological lattice models exhibiting UV/IR mixing~\cite{oh22a, pace-wen, oh22b, oh23, watanabe23, ebisu_symmetric_2023, kim_unveiling_2025, delfino23, ebisu24, ebisu25, kim_gauging_2026, delfino_topological_2026, bravyi_high-threshold_2024}.

\subsection{Wen's plaquette model}
\label{sec:5a}

Consider an $L_x \times L_y$ square lattice with periodic boundary conditions. Wen's plaquette model is locally equivalent to a single copy of the toric code. We therefore label the anyons and the corresponding fusion rules as follows:
\begin{align}
&{\cal C}=\{1, e, m, f  \}, \nn
&e \times e=m\times m=1, & e\times m=f.
\end{align}
The GSD is $4$ for even $L_x$ and even $L_y$, while it is $2$ in all other cases. When both $L_x$ and $L_y$ are even, only the trivial automorphism occurs. In contrast, when $L_x$ ($L_y$) is odd, transporting an anyon in the unit cell along $l_x$ ($l_y$) induces the nontrivial automorphism
\begin{align}
&\phi_{x(y)}(1)=1, &&\phi_{x(y)}(e)=m, \nn
&\phi_{x(y)}(m)=e, &&\phi_{x(y)}(f)=f.
\end{align}
Therefore, we consider three cases: (i) no transparent domain wall, (ii) a transparent domain wall along $l_x$, and (iii) transparent domain walls along both $l_x$ and $l_y$.

In the absence of transparent domain walls, there are 4 MESs associated with $X \in \mathtt{A}_{x(y)}$, and the ICS has four extreme points for $X \in \mathtt{A}_{x(y)}$. Then, the maximum values of the TEE and the  entanglement asymmetry are given by
\begin{align}
\gamma_{{\rm LW},x(y)}^{\rm max}&=\log 2, & \Delta S_{X_{x(y)}}^{\rm max}&= 2\log 2.
\end{align}
Note that these results are consistent with those of Ref.~\onlinecite{lamas_higher-form_2025}.

To analyze the symmetries connecting MESs associated with $X \in \mathtt{A}_{x(y)}$, we need to understand the structure of the fusion rules. Let us label extreme points of the ICS on $X \in \mathtt{A}_{x(y)}$ as
\begin{align}
\mathfrak{C}_{x(y)}=\{\tilde{1}_{x(y)}, \tilde{e}_{x(y)}, \tilde{m}_{x(y)}, \tilde{f}_{x(y)}\}.
\end{align}
Choosing one extreme point of the ICS on $X \in \mathtt{A}_{x(y)}$ as $\tilde{1}_{x(y)}$, we define $\tilde{a}_{x(y)}$ as the extreme point obtained by fusing $a\in \mathcal{C}$ with $\tilde{1}_{x(y)}$. The corresponding fusion rules are then determined as follows, based on Proposition~\ref{pro:fusion}.

\begin{align}
&e \times \tilde{1}_{x(y)}=\tilde{e}_{x(y)} & &m \times \tilde{1}_{x(y)}=\tilde{m}_{x(y)} & &f \times \tilde{1}_{x(y)}=\tilde{f}_{x(y)}\nn
&e \times \tilde{e}_{x(y)}=\tilde{1}_{x(y)} & &m \times \tilde{e}_{x(y)}=\tilde{f}_{x(y)} & &f \times \tilde{e}_{x(y)}=\tilde{m}_{x(y)}\nn
&e \times \tilde{m}_{x(y)}=\tilde{f}_{x(y)} & &m \times \tilde{m}_{x(y)}=\tilde{1}_{x(y)} & &f \times \tilde{m}_{x(y)}=\tilde{e}_{x(y)}\nn
&e \times \tilde{f}_{x(y)}=\tilde{m}_{x(y)} & &m \times \tilde{f}_{x(y)}=\tilde{e}_{x(y)} & &f \times \tilde{f}_{x(y)}=\tilde{1}_{x(y)}.\label{eq:pfusion}
\end{align}
It is straightforward that every $\alpha \in \mathfrak{C}_{x(y)}$ belongs to $\mathfrak{C}_{x(y)}^{\tau_{y(x)}}$, since the cardinality of $\mathfrak{C}_{x(y)}$ equals the number of MESs associated with $X \in \mathtt{A}_{x(y)}$. Moreover, \eqref{eq:pfusion} implies that different anyon types produce distinct fusion outcomes for the elements of $\mathfrak{C}_{x(y)}$. Since the fusion rules reflect the outcomes of anyon tunneling, we conclude that transporting different anyon types along either fixed cycle $l_x$ or $l_y$ yields distinct $1$-form symmetries. This symmetry structure is consistent with the conventional interpretation of topological order in terms of spontaneous $1$-form symmetry breaking.

In the presence of a transparent domain wall along $l_x$, there are two MESs associated with $X \in \mathtt{A}_{x(y)}$, while the ICS has four extreme points for $X \in \mathtt{A}_{x}$ and two extreme points for $X \in \mathtt{A}_{y}$. The maximum values of the TEE and the  entanglement asymmetry are given by
\begin{align}
\gamma_{{\rm LW},x}^{\rm max}&=\log 2, & \gamma_{{\rm LW},y}^{\rm max}&=\frac{1}{2} \log 2,
& \Delta S_{X_{x(y)}}^{\rm max}&= \log 2.
\end{align}
In this case, the maximal TEE depends on the orientation of the noncontractible annulus, reflecting the fact that the annulus can detect only transparent domain walls that cross it. In contrast, the maximal entanglement asymmetry takes the same value for the two orientations.

Let us label extreme points of the ICS on $X \in \mathtt{A}_{x(y)}$ as
\begin{align}
&\mathfrak{C}_{x}=\{\tilde{1}_x, \tilde{e}_x, \tilde{m}_x, \tilde{f}_x \}, & &\mathfrak{C}_{y}=\{\tilde{\sigma}_+, \tilde{\sigma}_- \},
\end{align}
in analogy with the notation used in Ref.~\onlinecite{bombin_topological_2010}. The physical interpretation of labels $\tilde{\sigma}_\pm$ will be discussed at the end of this example.

The fusion rules of $\mathfrak{C}_x$ and $\cal C$ are the same as those given in \eqref{eq:pfusion}. The nontrivial automorphism acting on $\mathfrak{C}_x$ is summarized as follows:
\begin{align}
&\tau_y(\tilde{1}_x)=\tilde{1}_x & &\tau_y(\tilde{e}_x)=\tilde{m}_x\nn
&\tau_y(\tilde{m}_x)=\tilde{e}_x & &\tau_y(\tilde{f}_x)=\tilde{f}_x.
\end{align}
This implies that $\mathfrak{C}_x^{\tau_y}=\{\tilde{1}_x,\tilde{f}_x\}$. Moreover, transporting $1$ or $f$ along $l_y$ defines a symmetry, which is also a 1-form symmetry.

Using Theorem~\ref{th:NNNtau}, we summarize the fusion rules between $\mathfrak{C}_y$ and $\mathcal{C}$ as follows:
\begin{align}
&1 \times \tilde{\sigma}_+=f \times \tilde{\sigma}_+=\tilde{\sigma}_+ & &1 \times \tilde{\sigma}_-=f \times \tilde{\sigma}_-=\tilde{\sigma}_-\nn
&e \times \tilde{\sigma}_+=m \times \tilde{\sigma}_+=\tilde{\sigma}_- & &e \times \tilde{\sigma}_-=m \times \tilde{\sigma}_-=\tilde{\sigma}_+.\label{eq:pfusion2}
\end{align}
Here, we conclude that $\mathfrak{C}_{y}^{\tau_x}=\{\tilde{\sigma}_+,\tilde{\sigma}_-\}$. Moreover, the relations $1 \overset{y}{\sim} f$ and $e \overset{y}{\sim} m$, derived from \eqref{eq:pfusion2}, imply that the two distinct symmetry classes can be represented by $1$-form symmetry operators associated with transporting $1$ and $e$ along $l_x$, respectively. This symmetry structure is again consistent with the interpretation of topological order in terms of spontaneous $1$-form symmetry breaking.

For the final case, we consider transparent domain walls along both $l_x$ and $l_y$. There are two MESs associated with $X \in \mathtt{A}_{x(y)}$, and the ICS has two extreme points for $X \in \mathtt{A}_{x(y)}$. The maximum values of the TEE and the entanglement asymmetry are given by
\begin{align}
\gamma_{{\rm LW},x(y)}^{\rm max}&=\frac{1}{2} \log 2,\nn
\Delta S_{X_{x(y)}}^{\rm max}&= \log 2.
\end{align}

Let us label extreme points of the ICS on $X \in \mathtt{A}_{x(y)}$ as
\begin{align}
&\mathfrak{C}_{x(y)}=\{\tilde{\sigma}_{+,x(y)}, \tilde{\sigma}_{-,x(y)} \}.
\end{align}
Then the fusion rules of $\mathfrak{C}_{x(y)}$ and $\cal C$ are the same as those given in \eqref{eq:pfusion2}. Hence, $\mathfrak{C}_{x(y)}=\mathfrak{C}_{x(y)}^{\tau_{y(x)}}$. The symmetries are realized either by transporting $1$ or $f$ along $l_x$ or $l_y$, or by transporting one $e$ along $l_x$ and another $e$ along $l_y$. The latter process, involving the transport of two anyons along different noncontractible loops, is necessary within the fixed fundamental-cycle decomposition $(l_x,l_y)$, because transport of $e$ or $m$ along either $l_x$ or $l_y$ alone does not define a symmetry in the presence of the nontrivial automorphism. Moreover, since $e \overset{x(y)}{\sim} m$ is the only nontrivial equivalence involving $e$ and $\phi_x(m)=\phi_y(m)=e$, this construction provides an explicit example illustrating Proposition~\ref{pro:symd}. A schematic illustration can be found in Fig.~\ref{fig:sym}(b). We note, however, that a process inducing the same fusion permutation can be realized by transporting $e$ along $l_{(1,1)}$, which corresponds to a 1-form symmetry.

We now examine the ground state degeneracy and the spontaneous symmetry breaking interpretation using these symmetries. When we choose a noncontractible annulus in $\mathtt{A}_x$, the two extreme points of its ICS remain unchanged under transport of $1$ or $f$ along $l_x$. By the fusion rules, they also remain unchanged under transport of $1$ or $f$ along $l_y$. Among these symmetry actions, only the combined process of transporting one $e$ along $l_x$ and another $e$ along $l_y$ permutes the extreme points. Thus, the combined symmetry exchanges the two corresponding MESs, whereas the individual transport symmetries preserve each MES up to a phase. These two MESs form a basis for the ground state space, giving a ground state degeneracy of $2$. This structure provides a basis for discussing spontaneous symmetry breaking in this topological phase. An analogous description applies when we choose a noncontractible annulus in $\mathtt{A}_y$. One remark here is that, in a sense, the topological phase can still be interpreted in terms of spontaneous $1$-form symmetry breaking because the combined process of transporting one $e$ along $l_x$ and another $e$ along $l_y$ is equivalent to transporting $e$ along $l_{(1,1)}$.

To conclude our discussion of Wen's plaquette model, we clarify the interpretation of the extreme points labeled $\tilde \sigma_\pm$ or $\tilde \sigma_{\pm,x(y)}$. Although both $\{\tilde{1}_x, \tilde{e}_x, \tilde{m}_x, \tilde{f}_x \}$ and $\{\tilde \sigma_{+},\tilde\sigma_{-}\}$ (or $\{\tilde \sigma_{+,x(y)},\tilde\sigma_{-,x(y)}\}$) label extreme points of the ICSs on noncontractible annuli, only the labels in the former set are chosen by analogy with the anyon types of Wen's plaquette model. In contrast, the labels $\tilde{\sigma}_{\pm}$ are chosen by analogy with the extrinsic twist point defects terminating a domain wall that implements the exchange $e\leftrightarrow m$, as discussed in Ref.~\onlinecite{bombin_topological_2010}. Specifically, the fusion rules in \eqref{eq:pfusion2} and the associated quantum dimensions coincide with those of the twist defects $\sigma_{\pm}$ introduced in Ref.~\onlinecite{bombin_topological_2010}, motivating this notation.

This interpretation is further supported by the construction in Appendix~\ref{asec:pd}, where we show that the ICS on a local annulus enclosing a twist point defect is isomorphic to that on a noncontractible annulus crossing a domain wall connecting the twist point defects. Note that transporting one defect around a noncontractible cycle of the torus and subsequently annihilating the defect pair converts this open wall into a closed transparent domain wall.


\subsection{Anisotropic dipolar toric code}

Consider an $L_x \times L_y$ lattice having periodic boundary conditions. The anisotropic dipolar toric code is locally equivalent to two copies of the $\mathbb{Z}_N$ toric code. We therefore label the anyons and the corresponding fusion rules as follows:
\begin{align}
&\mathcal{C}=\left\{
e_1^{n_1} m_1^{n_2} e_2^{n_3} m_2^{n_4}
\;\middle|\;
n_1,n_2,n_3,n_4 \in \mathbb{Z}_N
\right\},\nn
&e_{1(2)}^N=m_{1(2)}^N=1,
\end{align}
where $e_{1(2)}$ and $m_{1(2)}$ denote the elementary anyons in the first (second) copy of the toric code. The GSD is $N^2\gcd(L_x,N)^2$, where $\gcd$ is the greatest common divisor. 

In general, transporting an anyon along $l_x$ induces the nontrivial automorphism
\begin{align}
\phi_{x}(e_1)&=e_1\times m_2^{(-L_x~{\rm mod}~N)},\nn
\phi_{x}(e_2)&=e_2 \times m_1^{-(-L_x~{\rm mod}~N)},\nn
\phi_{x}(m_1)&=m_1,\nn
\phi_{x}(m_2)&=m_2,
\end{align}
so there is a nontrivial transparent domain wall along $l_y$ when $L_x$ is not an integer multiple of $N$. The number of extreme points of the ICS on $X \in \mathtt{A}_x$ is $N^2\gcd(L_x,N)^2$, while that on $X \in \mathtt{A}_y$ is $N^4$; see Appendix~\ref{asec:new} for details.

We summarize the maximum values of the TEE and the entanglement asymmetry as follows:
\begin{align}
\gamma_{{\rm LW},x}^{\rm max}&=\log [N\gcd(L_x,N)]\nn
\gamma_{{\rm LW},y}^{\rm max}&= 2\log N\nn
\Delta S_{X_{x(y)}}^{\rm max}&=2 \log [N\gcd(L_x,N)].
\end{align}
Explicit calculations of these quantities can be found in Appendices~\ref{asec:C} and~\ref{asec:D}.

To make the nontrivial fusion rules explicit, and to analyze the symmetries connecting MESs associated with $X \in \mathtt{A}_{x(y)}$, we specialize to $N=2$ with odd $L_x$. Accordingly, for $X \in \mathtt{A}_{x}$, there are four MESs and, correspondingly, four extreme points of the ICS. Let us label extreme points of the ICS on $X \in \mathtt{A}_{x}$ as
\begin{align}
\mathfrak{C}_{x}=\{\tilde{\alpha}_{(0,0)}, \tilde{\alpha}_{(0,1)}, \tilde{\alpha}_{(1,0)}, \tilde{\alpha}_{(1,1)}\}.
\end{align}
Choosing $\tilde{\alpha}_{(0,0)}$ as an element of $\mathfrak{C}_{x}^{\tau_y}$, and using Lemma~\ref{lem:equiv2}, we can conclude that
\begin{align}
1\overset{x}{\sim}m_1\overset{x}{\sim}m_2\overset{x}{\sim}m_1\times m_2
\end{align}
and we define other elements of $\mathfrak{C}_{x}$ as
\begin{align}
&e_1\times \tilde{\alpha}_{(0,0)}= \tilde{\alpha}_{(1,0)}\nn
&e_2\times \tilde{\alpha}_{(0,0)}= \tilde{\alpha}_{(0,1)}.
\end{align}
It is straightforward that all elements in $\mathfrak{C}_{x}$ exhibit trivial automorphisms, since the cardinality of $\mathfrak{C}_{x}$ equals the number of MESs associated with $X \in \mathtt{A}_{x}$. Moreover, the anyon types $1$, $e_1$, $e_2$, and $e_1e_2$ produce distinct fusion outcomes for the elements of $\mathfrak{C}_{x}$. Transporting each of these anyon types along $l_y$ therefore defines a distinct $1$-form symmetry.

For $X \in \mathtt{A}_{y}$, there are four MESs and $16$ extreme points of the ICS. Let us label extreme points of the ICS on $X \in \mathtt{A}_{y}$ as
\begin{align}
\mathfrak{C}_{y}=\{\tilde{\alpha}_{(a,b,c,d)}| a,b,c,d \in \mathbb{Z}_2\}.
\end{align}
Choosing $\tilde{\alpha}_{(0,0,0,0)}$ as an element of $\mathfrak{C}_{y}^{\tau_x}$, we define other elements of $\mathfrak{C}_{y}$ as
\begin{align}
&e_1^a\times e_2^b \times m_1^c \times m_2^d \times \tilde{\alpha}_{(0,0,0,0)}= \tilde{\alpha}_{(a,b,c,d)}.
\end{align}
Here, $\mathfrak{C}_y^{\tau_x}=\{\tilde{\alpha}_{(0,0,0,0)}, \tilde{\alpha}_{(0,0,0,1)}, \tilde{\alpha}_{(0,0,1,0)}, \tilde{\alpha}_{(0,0,1,1)} \}$. Moreover, the anyons $1$, $m_1$, $m_2$, and $m_1m_2$ all connect elements of $\mathfrak{C}_{y}^{\tau_x}$. Therefore, transporting any of these anyons along $l_x$ is a symmetry, and all such symmetries are 1-form symmetries.

\subsection{Rank-2 toric code}
\label{sec:5c}

Consider an $L_x \times L_y$ lattice with periodic boundary conditions. The rank-2 toric code is locally equivalent to three copies of the $\mathbb{Z}_N$ toric code. We therefore label the anyons and the corresponding fusion rules as follows:
\begin{align}
&\mathcal{C}=\left\{
e_1^{n_1} m_1^{n_2} e_2^{n_3} m_2^{n_4} e_3^{n_5} m_3^{n_6}
\;\middle|\;
n_1,n_2,n_3,n_4,n_5,n_6 \in \mathbb{Z}_N
\right\},\nn
&e_{1}^N=m_{1}^N=e_{2}^N=m_{2}^N=e_{3}^N=m_{3}^N=1,
\end{align}
where the subscripts $1$, $2$, and $3$ denote the first, second, and third copies of the toric code, respectively. The GSD is $N^3\gcd(L_x,N)\gcd(L_y,N)\gcd(L_x,L_y,N)$, where $\gcd(L_x,L_y,N)$ is the greatest common divisor of $L_x$, $L_y$, and $N$.

In general, transporting an anyon in the unit cell along $l_x$ ($l_y$) induces the nontrivial automorphism
\begin{align}
\phi_{x}(e_1)&=e_1 \times e_2^{(-L_x ~{\rm mod }~N)},\nn
\phi_{y}(e_1)&=e_1 \times e_3^{(-L_y ~{\rm mod }~N)},\nn
\phi_{x}(e_2)&=\phi_{y}(e_2)=e_2,\nn
\phi_{x}(e_3)&=\phi_{y}(e_3)=e_3,\nn
\phi_{x}(m_1)&=\phi_{y}(m_1)=m_1,\nn
\phi_{x}(m_2)&=m_2 \times m_1^{-(-L_x ~{\rm mod }~N)},\nn
\phi_{y}(m_2)&=m_2,\nn
\phi_{x}(m_3)&=m_3,\nn
\phi_{y}(m_3)&=m_3 \times m_1^{-(-L_y ~{\rm mod }~N)},
\end{align}
so a nontrivial transparent domain wall appears along $l_y$ when $L_x$ is not divisible by $N$, while one appears along $l_x$ when $L_y$ is not divisible by $N$. The number of extreme points of the ICS on $X \in \mathtt{A}_x$ is $N^4 \gcd(L_x,N)^2$, whereas for $X \in \mathtt{A}_y$ it is $N^4 \gcd(L_y,N)^2$; see Appendix~\ref{asec:new} for details.

We summarize the maximum values of the TEE and the entanglement asymmetry as follows:
\begin{align}
\gamma_{{\rm LW},x}^{\rm max}
&=
\log [N^2\gcd(L_x,N)], \\
\gamma_{{\rm LW},y}^{\rm max}
&=
\log [N^2\gcd(L_y,N)],\nn
\Delta S_{X_{x(y)}}^{\rm max}&=\log [N^3\gcd(L_x,N)\gcd(L_y,N) \gcd(L_x,L_y,N)].
\end{align}

To make the nontrivial fusion rules explicit, and to analyze the symmetries connecting MESs associated with $X \in \mathtt{A}_{x(y)}$, we specialize to the case $N=2$ with both $L_x$ and $L_y$ odd. For $X \in \mathtt{A}_{x(y)}$, there are $8$ MESs, and the ICS on $X$ has $16$ extreme points. Let us label extreme points of the ICS on $X \in \mathtt{A}_{x}$ as
\begin{align}
\mathfrak{C}_{x}=\{\tilde{\beta}_{(a,b,c,d)}| a,b,c,d \in \mathbb{Z}_2\}.
\end{align}
Choosing $\tilde{\beta}_{(0,0,0,0)}$ as an element of $\mathfrak{C}_{x}^{\tau_y}$ and using Lemma~\ref{lem:equiv2}, we can conclude that
\begin{align}
1\overset{x}{\sim}e_2\overset{x}{\sim}m_1\overset{x}{\sim}e_2\times m_1,
\end{align}
and we define other elements of $\mathfrak{C}_{x}$ as
\begin{align}
&e_1^a \times e_3^b \times m_2^c \times m_3^d \times \tilde{\beta}_{(0,0,0,0)}=\tilde{\beta}_{(a,b,c,d)},
\end{align}
where $a,b,c,d \in \mathbb{Z}_2$. Among them, not all exhibit trivial automorphisms, and $\mathfrak{C}_x^{\tau_y}$ can be expressed as
\begin{align}
\mathfrak{C}_x^{\tau_y}=\{&\tilde{\beta}_{(0,0,0,0)}, \tilde{\beta}_{(0,0,0,1)}, \tilde{\beta}_{(0,0,1,0)}, \tilde{\beta}_{(0,0,1,1)},\nn
&\tilde{\beta}_{(0,1,0,0)}, \tilde{\beta}_{(0,1,0,1)}, \tilde{\beta}_{(0,1,1,0)}, \tilde{\beta}_{(0,1,1,1)}\}.
\end{align}
Moreover, the anyons generated by $e_3$, $m_2$, and $m_3$ connect elements of $\mathfrak{C}_{x}^{\tau_y}$.

Let us label extreme points of the ICS on $X \in \mathtt{A}_{y}$ as
\begin{align}
\mathfrak{C}_{y}=\{\tilde{\gamma}_{(a,b,c,d)}| a,b,c,d \in \mathbb{Z}_2\}.
\end{align}
Choosing $\tilde{\gamma}_{(0,0,0,0)}$ as an element of $\mathfrak{C}_{y}^{\tau_x}$ and using Lemma~\ref{lem:equiv2}, we can conclude that
\begin{align}
1\overset{y}{\sim}e_3\overset{y}{\sim}m_1\overset{y}{\sim}e_3\times m_1,
\end{align}
and we define other elements of $\mathfrak{C}_{y}$ as
\begin{align}
&e_1^a \times e_2^b \times m_2^c \times m_3^d \times \tilde{\gamma}_{(0,0,0,0)}=\tilde{\gamma}_{(a,b,c,d)},
\end{align}
where $a,b,c,d \in \mathbb{Z}_2$. Among them, not all exhibit trivial automorphisms, and $\mathfrak{C}_y^{\tau_x}$ can be expressed as
\begin{align}
\mathfrak{C}_y^{\tau_x}=\{&\tilde{\gamma}_{(0,0,0,0)}, \tilde{\gamma}_{(0,0,0,1)}, \tilde{\gamma}_{(0,0,1,0)}, \tilde{\gamma}_{(0,0,1,1)},\nn
&\tilde{\gamma}_{(0,1,0,0)}, \tilde{\gamma}_{(0,1,0,1)}, \tilde{\gamma}_{(0,1,1,0)}, \tilde{\gamma}_{(0,1,1,1)}\}.
\end{align}
Moreover, the anyons generated by $e_2$, $m_2$, and $m_3$ connect elements of $\mathfrak{C}_{y}^{\tau_x}$.

Among the symmetries, some provide explicit examples of Proposition~\ref{pro:symd}. One such example is constructed as follows. The anyon $m_2$ connects two distinct elements of $\mathfrak{C}_y^{\tau_x}$, namely,
$m_2 \times \tilde{\gamma}_{(0,b,c,d)}=\tilde{\gamma}_{(0,b,c+1,d)}$. Starting from a pair of $m_2$ anyons, we transport one $m_2$ along $l_x$ and return it to its original position, thereby producing
\begin{align}
\phi_x(m_2)\times m_2=m_1.    
\end{align}
The resulting $m_1$ can be annihilated as follows. We first decompose it as the fusion product of $m_3\times m_1$ and $m_3$, and then transport the $m_3$ anyon along $l_y$. This gives
\begin{align}
\phi_y(m_3)\times m_3 \times m_1=1.
\end{align}
Moreover, $m_2 \overset{y}{\sim} m_2\times e_3 \overset{y}{\sim} m_2\times m_1 \overset{y}{\sim} m_2\times e_3\times m_1$ and none of these anyons belongs to $\mathcal{C}^{\phi_x}$. We note, however, that a process inducing the same fusion permutation can be chosen as the product of three 1-form symmetries, corresponding to transporting $m_2 m_3$ along $l_{(1,1)}$, $m_2$ along $l_y$, and $m_3$ along $l_x$.

\section{Discussion}
\label{sec:6}

We have developed an entanglement-bootstrap framework that characterizes transparent domain walls directly from fixed-point ground state wavefunctions of Abelian topological phases on a torus. Using ICSs on local and noncontractible annuli, we derive fusion rules and quantum dimensions and relate these data to MESs, the GSD, and symmetry actions. Central to this framework is the distinction between extreme points permitted by local entanglement constraints on a noncontractible annulus and the subset that can arise as reductions of global ground states. A noncontractible annulus with a fixed microscopic width sufficient for admissibility occupies a vanishing fraction of the system's area in the thermodynamic limit. Its ICS, together with transport around the complementary cycle, thus connects entanglement information on regions of subextensive area to global topological properties.

The framework distinguishes the information captured by TEE and entanglement asymmetry. The maximal TEE, $\gamma_{\rm LW}^{\max}=\log(\mathcal{D}/d_\alpha)$, probes the net effect of walls crossing the chosen annulus. The maximal entanglement asymmetry, $\Delta S_X^{\max}=\log\mathrm{GSD}$, reflects the combined effect of the full wall configuration on the global ground state space and is identical for both annulus orientations. These diagnostics therefore capture distinct aspects of the wall configuration, extending the comparison developed for Abelian topological phases without transparent domain walls in Ref.~\onlinecite{lamas_higher-form_2025}. Related TEE results were obtained in Ref.~\onlinecite{kim_unveiling_2025}, which interpreted the rank-2 toric code as a translation-symmetry-enriched topological phase and explained its lattice-size dependence through translation symmetry defects. Under our fixed-point assumptions, the corresponding TEE relations follow directly from the ICS description of transparent domain walls, without requiring microscopic translation symmetry.

This comparison relies on our extension of entanglement asymmetry to configurations with transparent domain walls, which requires identifying the relevant symmetry action. Such walls can give rise to symmetries supported jointly on the two chosen fundamental cycles that cannot be decomposed into a product of two $1$-form symmetry operators, one supported on each cycle. We construct these symmetries through Abelian anyon transport and use their induced permutations of globally realizable extreme points of ICSs on noncontractible annuli to define the symmetrization underlying entanglement asymmetry. This construction connects transitions between MESs to symmetry breaking visible in their annular reduced states. The resulting diagnostic is insensitive to relative phases between MESs and therefore does not characterize the full symmetry action on the global ground state space.

A useful comparison is with the \emph{memory capacity} introduced in Ref.~\onlinecite{yang_topological_2025}, defined as the maximal entropy difference between two density matrices in the same ICS. In the Abelian setting considered here, all extreme points of $\Sigma(X)$ have equal entropy. For $X\in\mathtt{A}_y$, the memory capacity is therefore attained between any extreme point and the uniform mixture of all extreme points, giving
\begin{align*}
C_{\mathrm{mem}}(X)=\log|\mathfrak{C}_y|.  \end{align*}
Entanglement asymmetry instead uses the effective symmetry action on the extreme points compatible with global ground states, whose number equals the GSD. Consequently,
\begin{align*}
C_{\mathrm{mem}}(X)-\Delta S_X^{\max}
=\log\frac{|\mathfrak{C}_y|}{\mathrm{GSD}}\geq 0.
\end{align*}
This difference quantifies the restriction imposed by compatibility with complementary-cycle transport. Memory capacity characterizes the full annular ICS, whereas maximal entanglement asymmetry reflects its globally realizable subset. Both quantities contain topological information, but neither uniquely determines the wall configuration.

Natural extensions include non-Abelian topological phases and immersed regions~\cite{shi_immersed_2024} with noncontractible boundaries. These settings may exhibit richer ICS structures and symmetry actions, requiring a broader treatment of the relation between fusion data, TEE, and entanglement asymmetry.

Point defects provide another direction. One could retain the axioms ${\bf A0}$ and ${\bf A1}$ away from specified defect cores while modifying the entanglement constraints in their neighborhoods. ICSs on local annuli enclosing point defects and on noncontractible annuli crossing transparent domain walls could then be studied within a common framework. Appendix~\ref{asec:pd} provides a concrete example: in the construction considered there, the ICS on a noncontractible annulus intersecting a wall is isomorphic to that on a local annulus enclosing a point defect~\cite{bombin_topological_2010}. Extending this correspondence may clarify how defect fusion degrees of freedom contribute to the GSD and modify the global compatibility conditions.

Finally, our framework could be extended beyond fixed points by allowing controlled violations of the entropic axioms that diminish under coarse-graining~\cite{yang_topological_2025}. One could then investigate whether topological entanglement entropy and entanglement asymmetry retain their distinct sensitivities to transparent domain walls and remain stable under deformations of the annulus that preserve its topology and wall-crossing structure. A further direction is to develop a mixed-state generalization of the axioms for studying transparent domain walls under decoherence. For specified noise channels, one could examine how these diagnostics change with noise strength and determine which features of the wall configuration remain detectable in the resulting mixed states.

\section{Acknowledgments}

We are grateful to Tai-Hsuan Yang and Yuta Hirasaki for helpful discussions. J.K. was supported by the education and training program of the Quantum Information Research Support Center, funded through the National Research Foundation of Korea (NRF) by the Ministry of Science and ICT (MSIT) of the Korean government (No. RS-2023-NR057243). A.G.L. was supported by the AAUW International Fellowship 2025-26. J.G. and T.L.H. were supported by the US Office of Naval Research MURI grant N00014-20-1-2325. B.S. and J.Y.L. are supported by the IQUIST fellowship, faculty startup grant at the University of Illinois, Urbana-Champaign, and IBM-Illinois Discovery Accelerator Institute.

\onecolumngrid
\appendix
\renewcommand{\theequation}{\Alph{section}.\arabic{equation}}
\renewcommand{\thedefinition}{\Alph{section}.\arabic{definition}}
\renewcommand{\thetheorem}{\Alph{section}.\arabic{theorem}}

\section{Transparent domain walls and point defects}
\label{asec:pd}

We now rigorously demonstrate how ICSs elucidate the relation between point defects and their domain wall for the fixed points of the Abelian topological phases. To incorporate point defects into the reference state framework, we allow the axioms $\bf{A0}$ and $\bf{A1}$ to be violated within a set of local disks containing the point defects.

\begin{figure}[ht]
\centering
\begin{overpic}[width=0.6\linewidth]{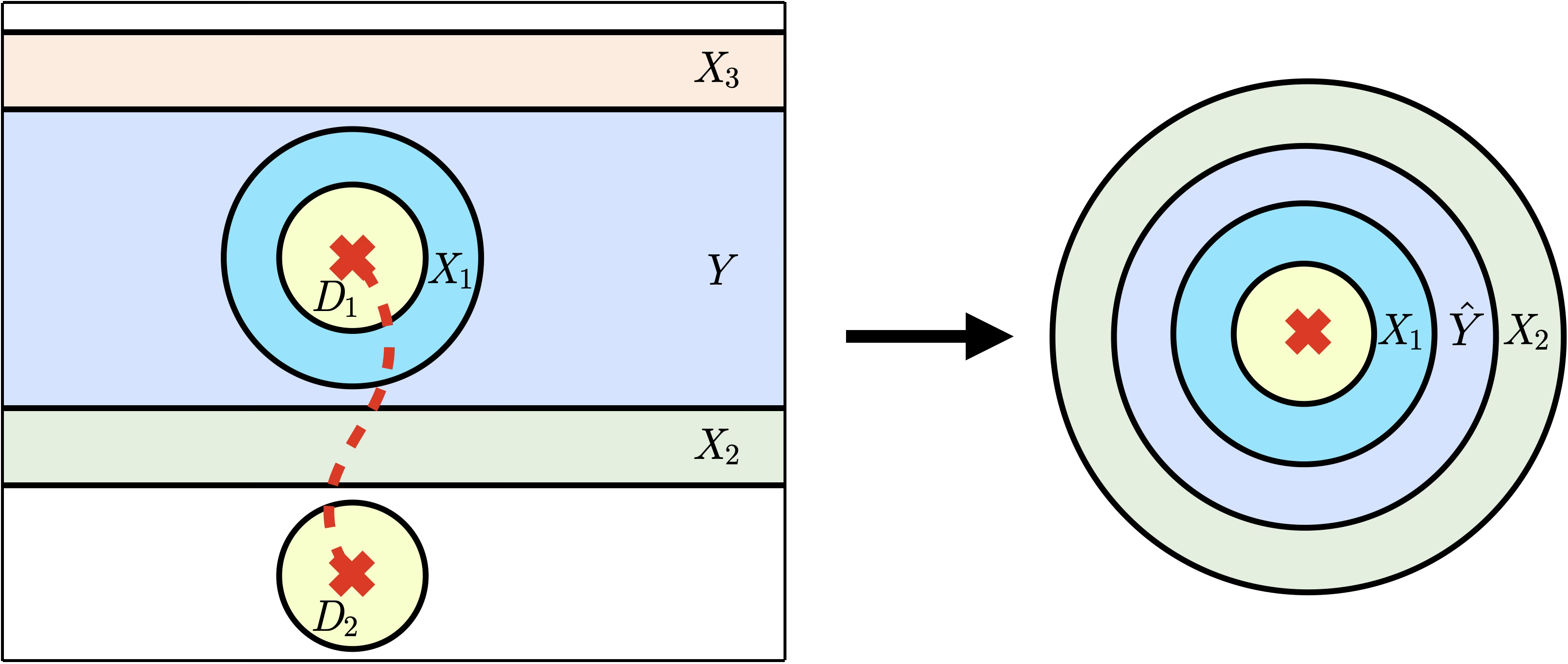}
\put(-4,46){(a)}
\put(68,46){(b)}
\end{overpic}
\caption{(a) Configuration of three annuli $X_1$, $X_2$, and $X_3$, together with the punctured noncontractible annulus $Y$. (b) Application of the completion trick to heal the boundary not crossed by the defect line. This yields a region $X_1 \cup \widehat{Y} \cup X_2$, which is topologically an annulus. This construction explains \eqref{eq:iso-X2-X3}.
}
\label{fig:defect-pt}
\end{figure}

\begin{proposition}
We consider a reference state on $\mathtt{R} \setminus (D_1 \cup D_2)$ at a fixed point of an Abelian topological phase, where $D_1$ and $D_2$ are local disks containing the point defects. Let $X_1$ denote a local annulus given by the thickened boundary of $D_1$, and let $X_2\in \mathtt{A}_x$. We assume that an open domain wall terminates at these defects and intersects $X_2$, as shown in Fig.~\ref{fig:defect-pt}(a). We then have
\begin{equation}\label{eq:iso-X2-X3}
\Sigma(X_1) \cong \Sigma(X_2).
\end{equation}

\begin{proof}
Let $Y$ be a punctured noncontractible annulus with a hole $D_1 \cup X_1$, as shown in Fig.~\ref{fig:defect-pt}(a). The three thickened boundary components of $Y$ consist of one local annulus $X_1$ and two noncontractible annuli $X_2$ and $X_3$. Here, $X_3$ is a noncontractible annulus that is not intersected by the domain wall. Equivalently, in the context of Abelian models, the ICS of $X_3$ contains only Abelian extreme points. This property was established in Ref.~\onlinecite{shi_immersed_2024}; see Proposition~13 and the resolved Abelian case of Conjecture~14.

Pick an arbitrary extreme point of $\Sigma(Y)$. Such a state necessarily carries an Abelian charge on $X_3$. We purify the top boundary of $X_3$ and collapse this boundary to a point. This point is referred to as a ``completion point'' (Ref.~\onlinecite{shi_remote_2025}, Section 4). Since the charge is Abelian, the axioms remain satisfied at the completion point, allowing us to disregard it. The resulting state is supported on the region $X_1 \cup \widehat{Y} \cup X_2$, which is topologically an annulus, as shown in Fig.~\ref{fig:defect-pt}(b). Taking the state on $X_1 \cup \widehat{Y} \cup X_2$ as a reference state, the annulus $X_1$ can be smoothly deformed into $X_2$ by the isomorphism theorem. Therefore, \eqref{eq:iso-X2-X3} holds.

\end{proof}
\end{proposition}

\section{Path independence of annulus transport}
\label{asec:1}

In this section, we rigorously define paths for the transportation of the annuli, and prove a proposition.

\begin{definition}[Path]
A finite sequence of subsystems $\{X_t\}$, with $t=i/N$ and $i=0,1,2,\ldots,N$ for some positive integer $N$, is called a path connecting $X_0$ and $X_1$ if each consecutive pair of subsystems is related by an elementary deformation step.
\end{definition}
\noindent Here, the elementary deformation steps consist of the merging and partial trace operations, whose details can be found in Ref.~\onlinecite{shi_fusion_2020}. We can then establish the following proposition.
\begin{proposition}\label{pro:path}
Let $X_0$ and $X_1$ be elements of $\mathtt{A}_{y}$ contained in $C \in \mathtt{A}_{y}$; see Fig.~\ref{fig:path} for an example. Let $\{X_t^{(1)}\}$ and $\{X_t^{(2)}\}$ be two paths connecting $X_0$ and $X_1$ such that $X_0^{(i)} = X_0$ and $X_1^{(i)} = X_1$ for $i = 1,2$. Moreover, assume that $\bigcup_t X_t^{(1)} \subseteq C$ and $\bigcup_t X_t^{(2)} \subseteq C$. Then the isomorphisms
\[
\Phi_{\{X_t^{(1)}\}} : \Sigma(X_0) \to \Sigma(X_1)
\quad {\rm and} \quad
\Phi_{\{X_t^{(2)}\}} : \Sigma(X_0) \to \Sigma(X_1)
\]
are identical.
\begin{figure}[ht]
\centering
\includegraphics[width=0.32\linewidth]{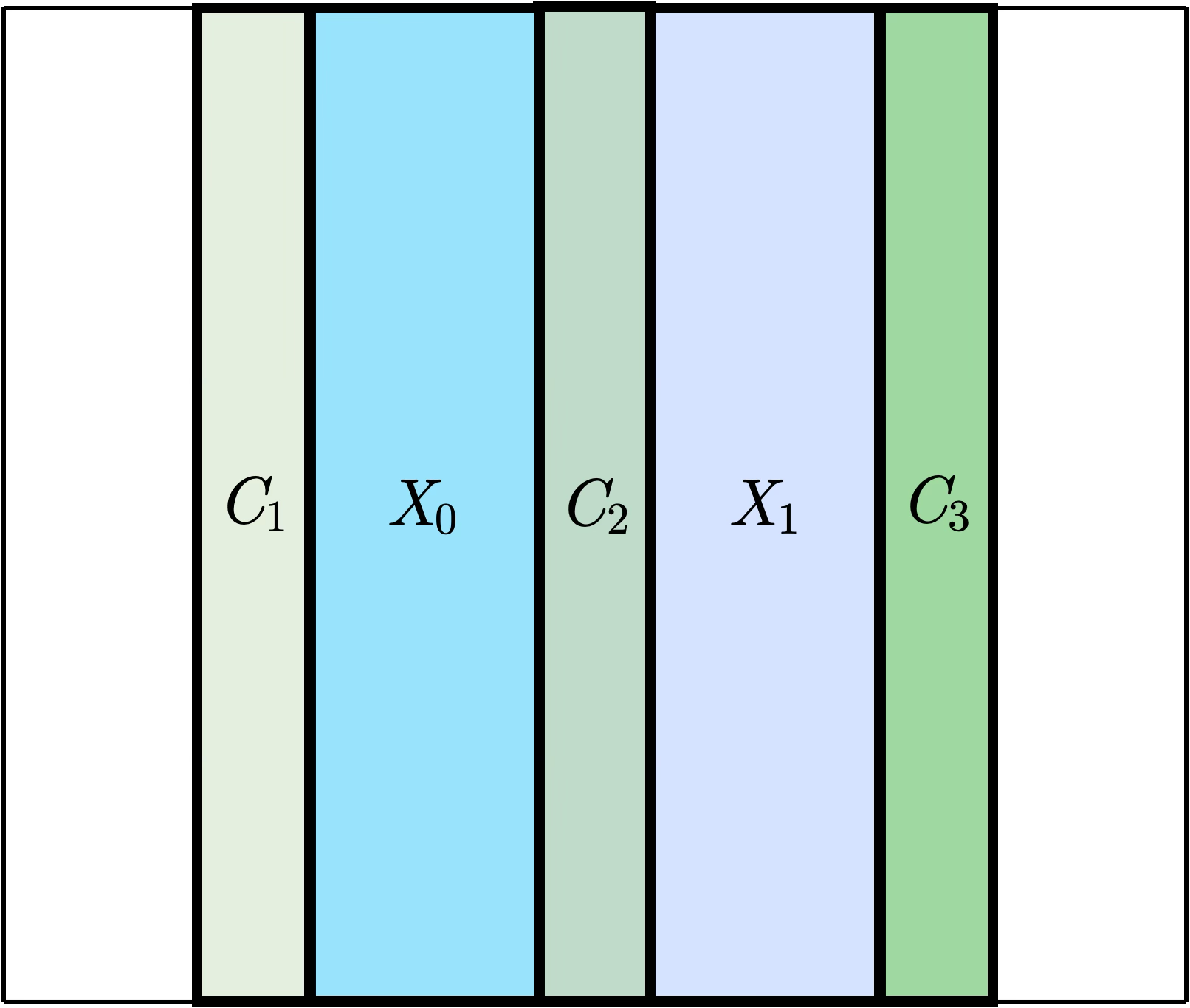}
\caption{One example of regions $X_0$, $X_1$, and $C=C_1 X_0 C_2 X_1 C_3$ satisfying the conditions of Proposition~\ref{pro:path}. More generally, $X_0 \cap X_1$ may be nonempty.}
\label{fig:path}
\end{figure}
\begin{proof}
The key point of the proof is that the two paths $\{X_t^{(1)}\}$ and $\{X_t^{(2)}\}$ become trivial upon extending each $X_t^{(i)}$ to a region $C$ for all $t$. Since every subsystem along the paths is $C$, the reduced density matrix on $C$ remains unchanged throughout the transport for all $t$. It then follows that the two isomorphisms for extended paths must coincide. Finally, $\Phi_{\{X_t^{(1)}\}}$ and $\Phi_{\{X_t^{(2)}\}}$ are obtained from the corresponding trivialized paths by partial trace. This completes the proof.
\end{proof}
\end{proposition}
The essence of the proof is to copy the information to $C$, which is unchanged under the extended paths. Such a proof does not work if we allow immersed intermediate configurations in the path, which is a case worth exploring.

\section{Anisotropic dipolar toric code as a reference state}
\label{asec:adtc}
In this section, we show that the ground state of the anisotropic dipolar toric code satisfies Axioms ${\bf A0}$ and ${\bf A1}$ in \eqref{eq:axioms}~\cite{ebisu_anisotropic_2023}. The model is defined on an $L_x \times L_y$ square lattice under periodic boundary conditions, with $\mathbb{Z}_N$ qudits located on the vertices and vertical edges. It is a stabilizer Hamiltonian $H=-\sum_{\bf r} V_{\bf r}-\sum_{\bf r} P_{{\bf r}-\frac{\hat{y}}{2}} +{\rm h.c.}$, where
\begin{align}
V_{\bf r}&= \sigma^x_{{\bf r}+\hat{x}} (\sigma^x)^{-2}_{\bf r} \sigma^x_{{\bf r}-\hat{x}}\sigma^x_{{\bf r}+\frac{\hat{y}}{2}} (\sigma^x)^{-1}_{{\bf r}-\frac{\hat{y}}{2}} \nn 
P_{{\bf r}-\frac{\hat{y}}{2}}&={\sigma}^z_{{\bf r}+\hat{x}-\frac{\hat{y}}{2}} ({\sigma^z_{{\bf r}-\frac{\hat{y}}{2}}})^{-2} {\sigma^z_{{\bf r}-\hat{x}-\frac{\hat{y}}{2}}} {\sigma^z_{{\bf r}}} ({\sigma^z_{{\bf r}-\hat{y}}})^{-1},
\end{align}
and $\sigma^x$, ${\sigma^z}$ are $\mathbb{Z}_N$ Pauli operators satisfying ${\sigma^z}\sigma^x=\omega \sigma^x {\sigma^z}$ ($\omega=e^{2\pi i/N}$). $V_{\bf r}$ and $P_{{\bf r}-\frac{\hat{y}}{2}}$ are illustrated in Fig.~\ref{fig:VP} and ${\bf r}=r_x \hat{x}+r_y \hat{y}$. The GSD is given by $N^2 [\gcd(L_x,N)]^2$, and the corresponding logical operators are
\begin{align}
&V_{\sigma^x}=\prod_{r_y} {\sigma}^x_{r_y\hat{y}} & &
\bar{V}_{\sigma^x}=\prod_{r_y} (\sigma^x_{r_y\hat{y}} )^\dagger{\sigma}^x_{r_y\hat{y}+\hat{x}}\nn
&H_{\sigma^x}=\prod_{r_x} {\sigma}^x_{{\bf r}-\hat{y}/2} & &\bar{H}_{\sigma^x}=\prod_{r_x}  ({\sigma^x_{{\bf r}-\hat{y}/2}})^{Nr_x/\gcd(L_x,N)} \nn
&V_{\sigma^z}=\prod_{r_y} {\sigma}^z_{r_y\hat{y}-\hat{y}/2} & &\bar{V}_{\sigma^z}=\prod_{r_y} (\sigma^z_{r_y\hat{y}-\hat{y}/2})^\dagger {\sigma}^z_{r_y\hat{y}-\hat{y}/2+\hat{x}} \nn
&H_{\sigma^z}=\prod_{r_x} {\sigma}^z_{{\bf r}} & & \bar{H}_{\sigma^z}=\prod_{r_x}  ({\sigma}^z_{{\bf r}})^{Nr_x/\gcd(L_x,N)}. \label{eq:lo}
\end{align}
Then the nontrivial commutation relations are
\begin{align}
&H_{\sigma^z} V_{\sigma^x}=\omega V_{\sigma^x}H_{\sigma^z}, & &\bar{H}_{\sigma^z} \bar{V}_{\sigma^x}=\omega^{N/\gcd(L_x,N)} \bar{V}_{\sigma^x} \bar{H}_{\sigma^z},\nn
&V_{\sigma^z} H_{\sigma^x}=\omega H_{\sigma^x} V_{\sigma^z}, & &\bar{V}_{\sigma^z} \bar{H}_{\sigma^x}=\omega^{N/\gcd(L_x,N)} \bar{H}_{\sigma^x} \bar{V}_{\sigma^z}.\label{eq:ZX}
\end{align}
We note that $\bar{V}_{\sigma^x}^n$ and $\bar{V}_{\sigma^z}^n$ for arbitrary $n$ are not, in general, independent logical operators, since $\bar{V}_{\sigma^z}^{\gcd(L_x,N)}$ can be expressed as the product of stabilizers. For example, when $N=4$ and $L_x=6$, $\bar{V}_{\sigma^x}$ is equivalent to $\bar{V}_{\sigma^x}^3$ up to stabilizers. These relations are well captured in \eqref{eq:ZX}, where two expressions equivalent up to stabilizers yield the same commutation relation.
\begin{figure}[ht]
\centering
\includegraphics[width=0.5\linewidth]{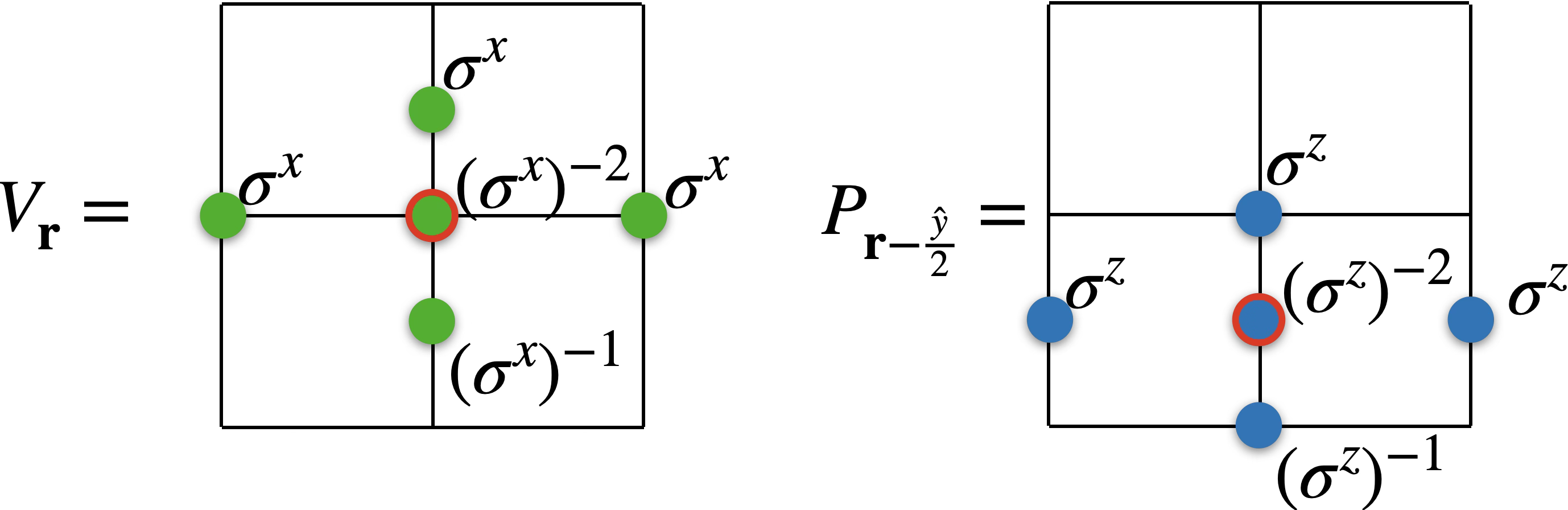}
\caption{$V_{\bf r}$ and $P_{{\bf r}-\frac{\hat{y}}{2}}$ are illustrated, where the red boundary dots for $V_{\bf r}$ and $P_{{\bf r}-\frac{\hat{y}}{2}}$ indicate the positions ${\bf r}$ and ${\bf r}-\frac{\hat{y}}{2}$, respectively.}
\label{fig:VP}
\end{figure}

To carry out the proof, we introduce an $l_x \times l_y$ region with $l_x \geq 2$, containing equal numbers of vertex qudits and vertical-edge qudits, and treat it as a supersite~\cite{shi_fusion_2020}. We then choose $b$ to be a ball containing nine supersites. To prove ${\bf A0}$ and ${\bf A1}$, we must consider all distinct subsystem configurations of $B$ and $C$ for ${\bf A0}$, and of $B$, $C$, and $D$ for ${\bf A1}$. The corresponding subsystem configurations are shown in Fig.~\ref{fig:EBP}(a) and Fig.~\ref{fig:EBP}(b), respectively. Here, each green region represents a single supersite. In this sense, the regions $B$ and $D$ may be viewed as tessellations of the region $C$.

\begin{figure}[ht]
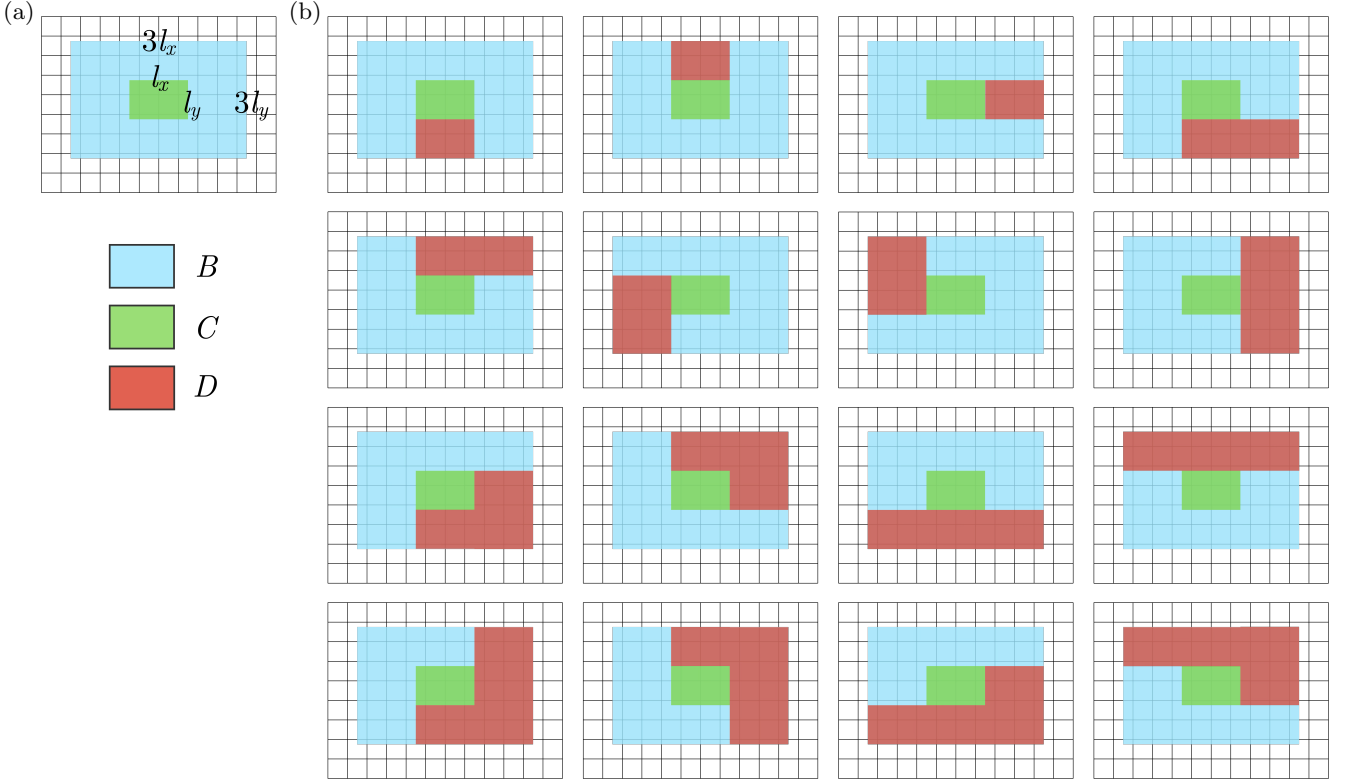

\centering
\begin{overpic}[width=0.95\linewidth]{EBP.png}
\put(-3,59){(a)}
\put(19.2,59){(b)}
\end{overpic}   
\caption{(a) Configuration used to show that the anisotropic dipolar toric code satisfies ${\bf A0}$. (b) Configuration used to show that the anisotropic dipolar toric code satisfies ${\bf A1}$. The figure is shown for $l_x = 3$ and $l_y = 2$, although $l_x$ and $l_y$ are arbitrary.}
\label{fig:EBP}
\end{figure}

\begin{table*}[ht]
\centering
\begin{tabular}{c|cccc}
\hline\hline
\multirow{2}{*}{} & ~~\multirow{2}{*}{$S(\rho_{BC})$}~~ & ~~\multirow{2}{*}{$S(\rho_{CD})$}~~& ~~\multirow{2}{*}{$S(\rho_{B})$}~~ & ~~\multirow{2}{*}{$S(\rho_{D})$}~~\\
& & \\
\hline
\multirow{2}{*}{$\mathtt{1}$}& ~~\multirow{2}{*}{$(6l_x+16l_y-6)\log N$}~~& ~~ \multirow{2}{*}{$(2l_x+8l_y-4)\log N$}~~ & ~~ \multirow{2}{*}{$(6l_x+20l_y-6)\log N$}~~ & ~~ \multirow{2}{*}{$(2l_x+4l_y-4)\log N$}~~ \\
&&\\
\multirow{2}{*}{$\mathtt{2}$}& ~~\multirow{2}{*}{$(6l_x+16l_y-6)\log N$}~~& ~~ \multirow{2}{*}{$(2l_x+8l_y-4)\log N$}~~ & ~~ \multirow{2}{*}{$(6l_x+20l_y-6)\log N$}~~ & ~~ \multirow{2}{*}{$(2l_x+4l_y-4)\log N$}~~ \\
&&\\
\multirow{2}{*}{$\mathtt{3}$}& ~~\multirow{2}{*}{$(8l_x+12l_y-6)\log N$}~~& ~~ \multirow{2}{*}{$(4l_x+4l_y-4)\log N$}~~ & ~~ \multirow{2}{*}{$(10l_x+12l_y-6)\log N$}~~ & ~~ \multirow{2}{*}{$(2l_x+4l_y-4)\log N$}~~ \\
&&\\
\multirow{2}{*}{$\mathtt{4}$}& ~~\multirow{2}{*}{$(6l_x+12l_y-5)\log N$}~~& ~~ \multirow{2}{*}{$(4l_x+8l_y-5)\log N$}~~ & ~~ \multirow{2}{*}{$(6l_x+16l_y-6)\log N$}~~ & ~~ \multirow{2}{*}{$(4l_x+4l_y-4)\log N$}~~ \\
&&\\
\multirow{2}{*}{$\mathtt{5}$}& ~~\multirow{2}{*}{$(6l_x+12l_y-5)\log N$}~~& ~~ \multirow{2}{*}{$(4l_x+8l_y-5)\log N$}~~ & ~~ \multirow{2}{*}{$(6l_x+16l_y-6)\log N$}~~ & ~~ \multirow{2}{*}{$(4l_x+4l_y-4)\log N$}~~ \\
&&\\
\multirow{2}{*}{$\mathtt{6}$}& ~~\multirow{2}{*}{$(6l_x+12l_y-5)\log N$}~~& ~~ \multirow{2}{*}{$(4l_x+8l_y-5)\log N$}~~ & ~~ \multirow{2}{*}{$(8l_x+12l_y-6)\log N$}~~ & ~~ \multirow{2}{*}{$(2l_x+8l_y-4)\log N$}~~ \\
&&\\
\multirow{2}{*}{$\mathtt{7}$}& ~~\multirow{2}{*}{$(6l_x+12l_y-5)\log N$}~~& ~~ \multirow{2}{*}{$(4l_x+8l_y-5)\log N$}~~ & ~~ \multirow{2}{*}{$(8l_x+12l_y-6)\log N$}~~ & ~~ \multirow{2}{*}{$(2l_x+8l_y-4)\log N$}~~ \\
&&\\
\multirow{2}{*}{$\mathtt{8}$}& ~~\multirow{2}{*}{$(4l_x+12l_y-4)\log N$}~~& ~~ \multirow{2}{*}{$(4l_x+12l_y-6)\log N$}~~ & ~~ \multirow{2}{*}{$(6l_x+12l_y-6)\log N$}~~ & ~~ \multirow{2}{*}{$(2l_x+12l_y-4)\log N$}~~ \\
&&\\
\multirow{2}{*}{$\mathtt{9}$}& ~~\multirow{2}{*}{$(6l_x+12l_y-4)\log N$}~~& ~~ \multirow{2}{*}{$(4l_x+8l_y-6)\log N$}~~ & ~~ \multirow{2}{*}{$(6l_x+12l_y-5)\log N$}~~ & ~~ \multirow{2}{*}{$(4l_x+8l_y-5)\log N$}~~ \\
&&\\
\multirow{2}{*}{$\mathtt{10}$}& ~~\multirow{2}{*}{$(6l_x+12l_y-6)\log N$}~~& ~~ \multirow{2}{*}{$(4l_x+8l_y-4)\log N$}~~ & ~~ \multirow{2}{*}{$(6l_x+12l_y-5)\log N$}~~ & ~~ \multirow{2}{*}{$(4l_x+8l_y-5)\log N$}~~ \\
&&\\
\multirow{2}{*}{$\mathtt{11}$}& ~~\multirow{2}{*}{$(6l_x+8l_y-4)\log N$}~~& ~~ \multirow{2}{*}{$(6l_x+8l_y-6)\log N$}~~ & ~~ \multirow{2}{*}{$(6l_x+12l_y-6)\log N$}~~ & ~~ \multirow{2}{*}{$(6l_x+4l_y-4)\log N$}~~ \\
&&\\
\multirow{2}{*}{$\mathtt{12}$}& ~~\multirow{2}{*}{$(6l_x+8l_y-4)\log N$}~~& ~~ \multirow{2}{*}{$(6l_x+8l_y-6)\log N$}~~ & ~~ \multirow{2}{*}{$(6l_x+12l_y-6)\log N$}~~ & ~~ \multirow{2}{*}{$(6l_x+4l_y-4)\log N$}~~ \\
&&\\
\multirow{2}{*}{$\mathtt{13}$}& ~~\multirow{2}{*}{$(4l_x+12l_y-5)\log N$}~~& ~~ \multirow{2}{*}{$(4l_x+12l_y-5)\log N$}~~ & ~~ \multirow{2}{*}{$(4l_x+12l_y-5)\log N$}~~ & ~~ \multirow{2}{*}{$(4l_x+12l_y-5)\log N$}~~ \\
&&\\
\multirow{2}{*}{$\mathtt{14}$}& ~~\multirow{2}{*}{$(4l_x+12l_y-5)\log N$}~~& ~~ \multirow{2}{*}{$(4l_x+12l_y-5)\log N$}~~ & ~~ \multirow{2}{*}{$(4l_x+12l_y-5)\log N$}~~ & ~~ \multirow{2}{*}{$(4l_x+12l_y-5)\log N$}~~ \\
&&\\
\multirow{2}{*}{$\mathtt{15}$}& ~~\multirow{2}{*}{$(6l_x+8l_y-5)\log N$}~~& ~~ \multirow{2}{*}{$(6l_x+8l_y-5)\log N$}~~ & ~~ \multirow{2}{*}{$(6l_x+8l_y-5)\log N$}~~ & ~~ \multirow{2}{*}{$(6l_x+8l_y-5)\log N$}~~ \\
&&\\
\multirow{2}{*}{$\mathtt{16}$}& ~~\multirow{2}{*}{$(6l_x+8l_y-5)\log N$}~~& ~~ \multirow{2}{*}{$(6l_x+8l_y-5)\log N$}~~ & ~~ \multirow{2}{*}{$(6l_x+8l_y-5)\log N$}~~ & ~~ \multirow{2}{*}{$(6l_x+8l_y-5)\log N$}~~ \\
&&\\
\hline\hline

\end{tabular}
\caption{Entanglement entropies for different configurations. Rows correspond to the configurations in Fig.~\ref{fig:EBP}(b), ordered left-to-right and top-to-bottom.}
\label{tab:S}
\end{table*}

We employ the method developed in Refs.~\onlinecite{hamma_bipartite_2005, kim_unveiling_2025, ebisu_symmetric_2023} to explicitly evaluate the entanglement entropies. The procedure can be summarized as follows. When a pure state $\ket{\psi}$ can be written as
\begin{align}
|\psi\rangle =\frac{1}{|G|^{1/2}}\sum_{g \in G} g|{\bf 0}\rangle,    
\end{align}
with $g = g_X \otimes g_Y$, where $|{\bf 0}\rangle$ is a $+1$ eigenstate of every ${\sigma^z}$ operator and $X$ and $Y$ form a bipartition, the reduced density matrix $\rho_X$ obeys
\begin{equation}\label{eq: rho_X projector}
     \rho_X^2=\left(\frac{|G_X||G_Y|}{|G|}\right)\rho_X,
\end{equation}
and the entanglement entropy of region $X$ is given by
\begin{align}\label{eq: S stabilizer}
S(\rho_X) = \log \left( \frac{|G|}{|G_X||G_Y|} \right),
\end{align}
where $G_X = \{ g \in G \mid g = g_X \otimes I_Y \}$ and $G_Y = \{ g \in G \mid g = I_X \otimes g_Y \}$.

In the calculation, $G$ is generated by the independent $V_{\mathbf{r}}$ operators, and, without loss of generality, we assume that the dependent generators are supported in $(BC)^c$ for Fig.~\ref{fig:EBP}(a) and in $(BCD)^c$ for Fig.~\ref{fig:EBP}(b). Under this assumption, when computing, for example, $S(\rho_B)$, it is unnecessary to evaluate $|G|$ explicitly, since $|G_{B^c}|$ can be expressed as $|G|$ multiplied by an additional factor~\cite{kim_unveiling_2025}.

First, let us show that ${\bf A0}$ is satisfied. We begin by computing $S(\rho_{BC})$. Then,
\begin{align}
|G_{BC}|&=N^{(3l_x-2)(3l_y-1)},\nn
|G_{(BC)^c}|&=|G|/N^{(3l_x+2)(3l_y+1)-2} \times N^2,\nn
S(\rho_{BC})&=(6l_x+12l_y-4)\log N.\label{eq:SBC}
\end{align}
Similarly, we obtain $S(\rho_{C})=(2l_x+4l_y-4)\log N$. To calculate $S(\rho_B)$, we can show that
\begin{align}
|G_B|&=N^{(3l_x-2)(3l_y-1)-(l_x+2)(l_y+1)+2}\times N^2\nn
|G_{B^c}|&=|G|/N^{(3l_x+2)(3l_y+1)}\times N^2 \times N^{(l_x-2)(l_y-1)}\times N^2,
\end{align}
which leads to $S(\rho_B)=(8l_x+16l_y-8)\log N$. Therefore, ${\bf A0}$ is satisfied.

Similarly, we compute the entanglement entropies needed to verify ${\bf A1}$; the results are summarized in Table~\ref{tab:S}. We therefore conclude that ${\bf A1}$ is satisfied as well.

\section{Details of lattice examples}
\label{asec:new}

In this section, we consider Wen's plaquette model~\cite{wen_quantum_2003, you_projective_2012}, the anisotropic dipolar toric code~\cite{ebisu_anisotropic_2023}, and the rank-2 toric code~\cite{oh22a,pace-wen,oh22b,oh23,kim_unveiling_2025,kim_gauging_2026}. For each model, we analyze the anyon automorphisms induced by transparent domain walls and determine the number of extreme points of the ICSs on noncontractible annuli.

These lattice models exhibit position-dependent anyons~\cite{pace-wen, oh22b, oh23, kim_unveiling_2025}, and varying the lattice size leads to the emergence of transparent domain walls. To analyze the resulting automorphisms, we define the unit cell as the minimal set of positions containing all distinct position-dependent anyon types. After fixing a specific unit cell in the lattice, we determine the automorphisms by transporting an anyon in the unit cell along $l_x$ or $l_y$ and identifying the anyon type that returns to the unit cell.

\subsection{Wen's plaquette model}

The position dependence of the anyons in Wen's plaquette model can be summarized as follows: for a plaquette located at ${\bf p} = p_x \hat{x} + p_y \hat{y}$, the excitation is an $e$ anyon when $p_x + p_y$ is odd, and an $m$ anyon when $p_x + p_y$ is even. Then the unit cell consists of two nearest-neighbor plaquettes. One can then readily verify that when both $L_x$ and $L_y$ are even, only the trivial automorphism occurs, that is, there are no transparent domain walls. In contrast, when $L_x$ ($L_y$) is odd, transporting an anyon in the unit cell along $l_x$ ($l_y$) induces the nontrivial automorphism
\begin{align}
&\phi_{x(y)}(1)=1, &&\phi_{x(y)}(e)=m, \nn
&\phi_{x(y)}(m)=e, &&\phi_{x(y)}(f)=f.
\end{align}

The number of extreme points of the ICS on $X \in \mathtt{A}_{x(y)}$ is $4$ in the absence of transparent domain walls. When a transparent domain wall is present along $l_x$, the number of extreme points is $4$ for $X \in \mathtt{A}_x$ and $2$ for $X \in \mathtt{A}_y$. When transparent domain walls are present along both $l_x$ and $l_y$, the number of extreme points is $2$ for both $X \in \mathtt{A}_x$ and $X \in \mathtt{A}_y$.

\subsection{Anisotropic dipolar toric code}

The position dependence of the anyons in the anisotropic dipolar toric code can be summarized as follows: for a vertex located at ${\bf r} = r_x \hat{x} + r_y \hat{y}$, the excitation with eigenvalue $\omega^a$ can be expressed as
\begin{align}
[v]_{\bf r}^a=e_1^{a}\times m_2^{a r_x} .
\end{align}
Moreover, there are excitations at ${\bf r}-\hat{y}/2$ and the excitations with eigenvalue $\omega^a$ can be expressed as
\begin{align}
[p]_{{\bf r}-\hat{y}/2}^a=e_2^{a}\times m_1^{-a r_x}.
\end{align}
Then the unit cell consists of $N$ vertices and $N$ vertical edges, for a total of $2N$ sublattices. One can then readily verify that when $L_x$ is a multiple of $N$, only the trivial automorphism occurs. In general, transporting an anyon in the unit cell along $l_x$ ($l_y$) induces the following nontrivial automorphism
\begin{align}
\phi_{x}([v]_{\bf r}^a)&=[v]_{{\bf r}+(-L_x ~{\rm mod}~ N)\hat{x}}^a,\nn
\phi_{y}([v]_{\bf r}^a)&=[v]_{{\bf r}}^a,\nn
\phi_{x}([p]_{{\bf r}-\hat{y}/2}^a)&=[p]_{{\bf r}-\hat{y}/2+(-L_x ~{\rm mod}~ N)\hat{x}}^a\nn
\phi_{y}([p]_{{\bf r}-\hat{y}/2}^a)&=[p]_{{\bf r}-\hat{y}/2}^a,
\end{align}
so there is a nontrivial transparent domain wall along $l_y$ when $L_x$ is not an integer multiple of $N$.

The number of extreme points of the ICS for $X \in \mathtt{A}_{x(y)}$ can be obtained by considering all distinct anyons within a unit cell:
\begin{align}
&\prod_{k=0}^{N-1}\Bigl ( [v]_{{\bf r}+k\hat{x}}^{a_k} \times [p]_{{\bf r}-\hat{y}/2+k\hat{x}}^{a_k'} \Bigl)\nn
&=\prod_{k=0}^{N-1} \Bigl (e_1^{a_k}\times m_2^{a_k (r_x+k)} \times e_2^{a_k'}\times m_1^{-a_k' (r_x+k)} \Bigl )\nn
&=e_1^{A_1}\times m_2^{A_1 r_x+A_2} \times e_2^{A_1'}\times m_1^{-A_1' r_x-A_2'}
\end{align}
where
\begin{align}
A_1&=\sum_{k=0}^{N-1} a_k & A_2&=\sum_{k=0}^{N-1} a_kk\nn
A_1'&=\sum_{k=0}^{N-1} a_k'& A_2'&=\sum_{k=0}^{N-1} a_k'k.
\end{align}
For these quantities to be invariant under $\phi_x$, both $A_1$ and $A_1'$ must be integer multiples of $N/\gcd(L_x,N)$. It then follows that the number of extreme points of the ICS on $X \in \mathtt{A}_x$ is $N^2 \gcd(L_x,N)^2$, whereas for $X \in \mathtt{A}_y$ it is $N^4$.

\subsection{Rank-2 toric code}

The position dependence of the anyons in the rank-2 toric code can be summarized as follows: for a vertex located at ${\bf r} = r_x \hat{x} + r_y \hat{y}$, the excitation with eigenvalue $\omega^a$ can be expressed as
\begin{align}
[e]_{\bf r}^a=e_1^{a}\times e_2^{a r_x} \times e_3^{a r_y} .
\end{align}
Moreover, there are excitations at ${\bf r}+\hat{x}/2$ and ${\bf r}+\hat{y}/2$, and the excitations with eigenvalue $\omega^a$ can be expressed as
\begin{align}
{[m^x]}_{{\bf r}+\hat{x}/2}^a &=m_3^{-a} \times m_1^{a r_y}\nn
{[m^y]}_{{\bf r}+\hat{y}/2}^a &=m_2^{a} \times m_1^{-a r_x}.
\end{align}
Then the unit cell consists of an $N\times N$ vertex sublattice, an $N\times 1$ vertical-edge sublattice, and a $1\times N$ horizontal-edge sublattice. One can then readily verify that when both $L_x$ and $L_y$ are multiples of $N$, only the trivial automorphism occurs. In general, transporting an anyon in the unit cell along $l_x$ ($l_y$) induces the nontrivial automorphism
\begin{align}
\phi_{x}([e]_{\bf r}^a)&=[e]_{{\bf r}+(-L_x ~{\rm mod }~N)\hat{x}}^a,\nn
\phi_{y}([e]_{\bf r}^a)&=[e]_{{\bf r}+(-L_y ~{\rm mod }~N)\hat{y}}^a,\nn
\phi_{x}({[m^x]}_{{\bf r}+\hat{x}/2}^a)&={[m^x]}_{{\bf r}+\hat{x}/2}^a,\nn
\phi_{y}({[m^x]}_{{\bf r}+\hat{x}/2}^a)&={[m^x]}_{{\bf r}+\hat{x}/2+(-L_y ~{\rm mod}~N)\hat{y}}^a,\nn
\phi_{x}({[m^y]}_{{\bf r}+\hat{y}/2}^a)&={[m^y]}_{{\bf r}+\hat{y}/2+(-L_x ~{\rm mod}~N)\hat{x}}^a,\nn
\phi_{y}({[m^y]}_{{\bf r}+\hat{y}/2}^a)&={[m^y]}_{{\bf r}+\hat{y}/2}^a,
\end{align}
so a nontrivial transparent domain wall appears along $l_y$ when $L_x$ is not divisible by $N$, while one appears along $l_x$ when $L_y$ is not divisible by $N$.

The number of extreme points of the ICS for $X \in A_{x(y)}$ can be obtained by considering all distinct anyons within a unit cell~\cite{kim_unveiling_2025}:
\begin{align}
&\Bigl (\prod_{k_x,k_y=0}^{N-1} [e]_{{\bf r}+k_x \hat{x}+k_y\hat{y}}^{a_{k_x,k_y}}\Bigl ) \times \Bigl ( \prod_{k_y=0}^{N-1}{[m^x]}_{{\bf r}+\hat{x}/2+k_y \hat{y}}^{a_{k_y}'}\Bigl ) \times \Bigl ( \prod_{k_x=0}^{N-1}{[m^y]}_{{\bf r}+\hat{y}/2+k_x \hat{x}}^{a_{k_x}''} \Bigl)\nn
&=e_1^{B_1}\times e_2^{B_1 r_x+B_2} \times e_3^{B_1 r_y+B_3} \times m_1^{B_4 r_y-B_5 r_x +B_6} \times m_2^{B_5}\times m_3^{-B_4}
\end{align}
where
\begin{align}
B_1&=\sum_{k_x,k_y} a_{k_x,k_y} & B_2 &= \sum_{k_x,k_y} a_{k_x,k_y} k_x & B_3&=\sum_{k_x,k_y} a_{k_x,k_y} k_y\nn
B_4&=\sum_{k_y} a_{k_y}'& B_5&=\sum_{k_x} a_{k_x}'' & B_6&=\sum_{k_y}a_{k_y}'k_y-\sum_{k_x} a_{k_x}''k_x.
\end{align}
For these quantities to be invariant under $\phi_x$, both $B_1$ and $B_5$ must be integer multiples of $N/\gcd(L_x,N)$. It then follows that the number of extreme points of the ICS on $X \in \mathtt{A}_x$ is $N^4 \gcd(L_x,N)^2$, whereas for $X \in \mathtt{A}_y$ it is $N^4 \gcd(L_y,N)^2$.

\section{ Entanglement asymmetry for the anisotropic dipolar toric code} 
\label{asec:C}

In this section, we calculate the entanglement asymmetry for an MES of the anisotropic dipolar toric code. We again employ the methods of Refs.~\onlinecite{hamma_bipartite_2005,kim_unveiling_2025,ebisu_symmetric_2023}, summarized in Appendix~\ref{asec:adtc}, to calculate the entanglement entropy of bipartite regions, taking the stabilizer structure of the ground state as the starting point. To illustrate the method and its application to the entanglement asymmetry, we first consider the simpler example of the toric code. For an MES reduced to a noncontractible region on the torus, the TEE and the entanglement asymmetry are known to be $\log 2$ and $2\log 2$, respectively.

\subsection{Warm-up: toric code}
The toric code Hamiltonian is given by
\begin{equation}\label{eq:ham_tc}
H = -\sum_{v} \, A_v \,
- \sum_p\, B_p,
\end{equation}
where
\begin{align}
&A_v = \begin{tikzpicture}[scale = 0.5, baseline = {([yshift=-.5ex]current bounding box.center)}]
    \draw[black] (1.5,0) -- (-1.5, 0);
    \draw[black] (0,1.5) -- (0,-1.5);
    \node at (0.9, 0) {\normalsize $\sigma^z$};
    \node at (-0.9, 0) {\normalsize $\sigma^z$};
    \node at (0, 0.9) {\normalsize $\sigma^z$};
    \node at (0, -0.9) {\normalsize $\sigma^z$};
    \node at (0.2, 0.2) {\small $\textcolor{black}{v}$};
\end{tikzpicture},
&B_p = \begin{tikzpicture}[scale = 0.5, baseline = {([yshift=-.5ex]current bounding box.center)}]
        \draw[color = black] (-1, -1) -- (-1, 1) -- (1, 1) -- (1, -1) -- cycle;
        \node at (-1, 0) {\small $\sigma^x$};
        \node at (1, 0) {\small $\sigma^x$};
        \node at (0, -1) {\small $\sigma^x$};
        \node at (0, 1) {\small $\sigma^x$};
        \node at (0, 0) {\small $\textcolor{black}{p}$};
    \end{tikzpicture},
\end{align}
and $\sigma^z$ and $\sigma^x$ denote the Pauli-$Z$ and Pauli-$X$ operators acting on the links of the square lattice, respectively. Here, we consider a torus of size $L_x \times L_y$. The electric Wilson line operator and the magnetic 't Hooft line operator are defined, respectively, as follows:
\begin{align}\label{eq:ops}
&U_\gamma = \prod_{l \in \gamma} \sigma^x_l, &  T_{\tilde{\gamma}} = \prod_{l \perp \tilde{\gamma}} \sigma^z_l,
\end{align}
where $\gamma$ denotes a path defined on the links of the lattice, and $\tilde{\gamma}$ denotes a path on the dual lattice.

We consider a subregion $X \in \mathtt{A}_y$ to compute the entanglement asymmetry. We calculate it for the MES associated with $X$, meaning that it has maximum TEE $\gamma_{{\rm LW},y}^{\rm max}=\log2$ \cite{zhang_quasiparticle_2012}, and consequently minimum entanglement entropy for region $X$. This state can be obtained from a reference state $|
{\bf 0}\rangle\equiv|+\rangle^{\otimes 2L_xL_y}$ by applying all elements of $G$ generated by the star operators $A_v$ centered on each vertex $v$ of the lattice, and the noncontractible 't Hooft line parallel to the entanglement cut, $T_y$. Then we can write our MES as 
\begin{equation} \label{eq: reference MES}
    |\text{MES}\rangle=\frac{1}{|G|^{1/2}}\sum_{g\in G} g|\bf 0\rangle.
\end{equation}
This MES is a symmetry-breaking ground state in the sense that applying either $U_x$ or $T_x$ transforms it into another MES of the toric code.

Using \eqref{eq: S stabilizer}, we can calculate the entanglement entropy for the MES associated with $X$, but we must first determine $|G|$, $|G_X|$, and $|G_Y|$. In the absence of the generator $T_y$, we have $|G| = 2^{L_x L_y - 1}$. The factor of $2^{-1}$ arises from the constraint $\prod_v A_v = 1$. When the generator $T_y$ is taken into account, we obtain $|G| = 2^{L_xL_y}$. Now let $L_xL_y=N_v=N_{v,X}+N_{v,Y}+N_{\partial X}$, where $N_{v,X(Y)}$ is the number of vertices completely inside region $X$ ($Y$), and $N_{\partial X}=2L_y=|\partial X|$ is the number of vertices at the boundary of $X$. Then $|G_X|=2^{N_{v,X}+1}$ and $|G_Y|=2^{N_{v,Y}+1}$, where the extra factor of $2$ comes from the generator $T_y$ being applied either inside subregion $X$ or $Y$. Then \eqref{eq: S stabilizer} gives
\begin{align}
    S(\rho_X)
    &= \log\left(\frac{2^{N_v}}{2^{N_{v,X}+1}2^{N_{v,Y}+1}}\right) = \log\left({2^{N_{\partial X}-2}}\right)\\
    &= |\partial X|\log 2-2\log 2,
\end{align}
where $|\partial X|$ is the length of the boundary of region $X$. This is the expected result, showing a correction of $-2\log2$ to the area law term of the entanglement entropy and $\gamma_{{\rm LW},y}^{\rm max}=\frac{1}{2}2\log 2=\log2$.

We can calculate the entanglement asymmetry using similar techniques as above. We start by showing that the symmetrized reduced density matrix $\rho_X^\text{sym}$ is also proportional to a projector operator. In the models considered in this section, the symmetry operators are all tensor products of local operators, and can be written as $U=U_{X}\otimes U_{\mathtt{R} \setminus X}$, where the subscripts indicate that the support is only on region $X$ or its complement $\mathtt{R} \setminus X$. Then the symmetrized reduced density matrix, defined in Definition~\ref{def:entanglement-asymmetry}, takes the simple form
\begin{equation}
\rho_{X}^{\rm sym} = \frac{1}{|F_{x(y)}|} \sum_{g \in F_{x(y)}} U_{g|X}  \rho_X U_{g|X}^\dagger.
\end{equation}

For a noncontractible region $X\in A_y$ wrapping around the $y$-cycle of the torus, the MES transforms under the Wilson and 't Hooft loops in the $x$-direction of the torus, $U_x$ and $T_x$. Then,
\begin{equation}\label{eq: appendix rho sym}
    \rho_X^\text{sym}= \frac{1}{4}\left(\rho_X+U_{x|X}\rho_XU^\dagger_{x|X}+T_{x|X}\rho_XT^\dagger_{x|X}+T_{x|X}U_{x|X}\rho_XU^\dagger_{x|X}T^\dagger_{x|X}\right),
\end{equation}
where $U_{x|X}$ and $T_{x|X}$ are the symmetry operators restricted to region $X$. Note that the symmetry operators transform the $\rho_X$ of an MES into the reduced density matrix of another MES, as discussed in the main text. Indeed, in the absence of a transparent domain wall, $\rho_X$ is an extreme point of the ICS of the region $X\in A_y$ for a reference state in the toric code fixed point. Then \eqref{eq: appendix rho sym} corresponds to \eqref{eq:rho_sym_MES} for $\rho_X^\text{sym}$, with the four terms being distinct extreme points of the ICS, and $|\mathfrak{C}_{y}^{\tau_x}|=|\mathbb{Z}_2\times \mathbb{Z}_2|=4$. 

To continue, we consider $(\rho_X^\text{sym})^2$. Since the extreme points (the reduced density matrices of different MESs) are orthogonal to each other, the cross-terms vanish, and we get 
\begin{align}
    (\rho_X^\text{sym})^2&=\frac{1}{16}\left(\rho^2_X+U_{x|X}\rho^2_XU^\dagger_{x|X}+T_{x|X}\rho^2_XT^\dagger_{x|X}+T_{x|X}U_{x|X}\rho^2_XU^\dagger_{x|X}T^\dagger_{x|X}\right)\\
    &=\left(\frac{|G_X||G_Y|}{|\mathbb{Z}_2\times \mathbb{Z}_2||G|}\right)\rho_X^\text{sym},
\end{align}
where we used \eqref{eq: rho_X projector} to go from the first line to the second. Following the same steps that lead to \eqref{eq: S stabilizer}, we obtain the entanglement entropy of $\rho_X^\text{sym}$ as
\begin{equation}
    S(\rho_X^\text{sym})=\log\left(\frac{|\mathbb{Z}_2\times \mathbb{Z}_2||G|}{|G_X||G_Y|}\right)=\log(|\mathbb{Z}_2\times \mathbb{Z}_2|)+S(\rho_X).
\end{equation}
Then the entanglement asymmetry is 
\begin{equation}
    \Delta S_X^{\rm max} = \log(|\mathbb{Z}_2\times \mathbb{Z}_2|)=2\log 2,
\end{equation}
which matches the result in Ref. \onlinecite{lamas_higher-form_2025}.

Note that we did not have to calculate $S(\rho_X)$, $|G_X|$, or $|G_Y|$ explicitly to find $\Delta S_X^{\rm max}$ using this formalism. This makes it feasible to calculate $\Delta S_X^{\rm max}$ for more complicated models even when finding the entanglement entropy or the TEE alone is challenging. Next, we employ the same techniques to calculate the entanglement asymmetry of the anisotropic dipolar toric code.

\subsection{Anisotropic dipolar toric code}

In this section, we calculate the  entanglement asymmetry associated with the anisotropic dipolar toric code using the same method as in the previous section. We consider the asymmetry for two distinct noncontractible regions, each wrapping around a different cycle of the torus, and their respective MESs.

\subsubsection{\texorpdfstring{$X\in \mathtt{A}_y$}{Ay}}

The MES for $X\in \mathtt{A}_y$ is 
\begin{equation}
    |{\rm MES}_y\rangle=\frac{1}{|G|^{1/2}}\sum_{g \in G} g|\textbf{0}\rangle,
\end{equation}
where the group $G$ is generated by all $V_{\bf r}$'s, $V_{\sigma^x}$, and $\bar{V}_{\sigma^x}$, and we take the reference state to be in the $\sigma^z$ basis, $|\textbf{0}\rangle\equiv |0\rangle^{\otimes 2L_xL_y}$, for $Z|0\rangle=|0\rangle$.
Our symmetry operators are $H_{\sigma^x}$, $\bar{H}_{\sigma^x}$, $H_{\sigma^z}$, $\bar{H}_{\sigma^z}$, so the symmetrized reduced density matrix is
\begin{equation}
    \rho_X^\text{sym}=\frac{1}{N^2\gcd(N,L_x)^2}\sum_{n,p=1}^{N}\sum_{m,q=1}^{\gcd{(N,L_x)}}H^n_{\sigma^z|X}\bar{H}_{\sigma^z|X}^m H_{\sigma^x|X}^p\bar{H}_{\sigma^x|X}^q \rho_{X} \bar{H}_{\sigma^x|X}^{-q} H_{\sigma^x|X}^{-p} \bar{H}_{\sigma^z|X}^{-m} H^{-n}_{\sigma^z|X}.
\end{equation}

We want to show that $(\rho_X^\text{sym})^2\propto\rho_X^\text{sym}$. To do that, we first show that the cross-terms of $(\rho_X^\text{sym})^2$ resulting from products of different symmetry operators vanish. First, we consider the cross terms for different powers of the same operator $H_{\sigma^z|X}$, having the form 
\begin{equation}
    H_{\sigma^z|X}^n \rho_X H_{\sigma^z|X}^{-n+m} \rho_X H_{\sigma^z|X}^{-m}.
\end{equation}
When $m\neq n$, we get 
\begin{align}
    H_{\sigma^z|X}^n \rho_X H_{\sigma^z|X}^{-n+m} \rho_X H_{\sigma^z|X}^{-m} &= |G_Y|^2 |G|^{-2}\sum_{\substack{g,h\in G/G_Y\\\tilde g_X,\tilde h _X\in G_X}} H_{\sigma^z|X}^n g_X\ket{0_X}\bra{0_X}\tilde{g}_X^\dagger g_X^\dagger H_{\sigma^z|X}^{-n+m} h_X\ket{0_X}\bra{0_X} \tilde{h}_X^\dagger h_X^\dagger H_{\sigma^z|X}^{-m}\nn
    &= |G_Y|^2 |G|^{-2} \sum_{k,l=1}^N \sum_{\substack{g,h\in G/G_Y\\\tilde g_X \in S_X^{(k)},\tilde h _X\in S_X^{(l)}}} \omega^{nk+ml} g_X\ket{0_X}\bra{0_X}\tilde{g}_X^\dagger g_X^\dagger h_X\ket{0_X}\bra{0_X} \tilde{h}_X^\dagger h_X^\dagger\nn
    &= |G_Y|^2 |G|^{-2} \sum_{k,l=1}^{N} \sum_{\substack{g\in G/G_Y\\\tilde g_X \in S_X^{(k)},\tilde h _X\in S_X^{(l)}}} \omega^{nk+ml} g_X\ket{0_X}\bra{0_X} g_X^\dagger \tilde{g}_X^\dagger \tilde{h}_X^\dagger\nn
    &= |G_Y|^2 |G|^{-2} \sum_{p} \omega^{mp}\sum_{\substack{g\in G/G_Y\\\tilde o_X \in S_X^{(p)}}} \sum_{k=1}^N \omega^{(n-m)k} g_X\ket{0_X}\bra{0_X} g_X^\dagger \tilde{o}_X^\dagger\nn
    &=0,
\end{align}
where $S_{X}^{(k)}=\{g_X\in G_X|g_X H_{\sigma^z|X}=\omega^k H_{\sigma^z|X} g_X\}$. Then these cross terms vanish. The same can be shown for the other cross terms which are powers of the same operator, such as 
\begin{align}
    \bar{H}_{\sigma^z|X}^n \rho_X \bar{H}_{\sigma^z|X}^{-n+m} \rho_X \bar{H}_{\sigma^z|X}^{-m} =0,
\end{align}
and the equivalent expressions for $\bar{H}_{\sigma^x|X}$.

Now we need to check the cross terms between different symmetry operators. First, note that $H_{\sigma^x}$ and $\bar{H}_{\sigma^x}$ commute with all elements of $G$, and with the other symmetry generators $H_{\sigma^z}$ and $\bar{H}_{\sigma^z}$ (all the $H$ operators wrap around the $x$-cycle of the torus, hence they commute even if they are composed of different Pauli operators). Then the only nontrivial action of $H_{\sigma^x}$ and $\bar{H}_{\sigma^x}$ is to transform the reference state $|\textbf{0}\rangle$ into an orthogonal state $|-\rangle\equiv H_{\sigma^x}|\textbf{0}\rangle$ having a noncontractible loop of $|1\rangle^{\otimes L_x}$ around the $x$-cycle of the torus. Note that the elements of the group $G$, $V_{\bf r}$, $V_{\sigma^x}$ and $\bar{V}_{\sigma^x}$ cannot generate such a noncontractible loop starting from the reference state $|\textbf{0}\rangle$, and $H_{\sigma^z}$ and $\bar{H}_{\sigma^z}$ act merely by measuring the charge as in \eqref{eq:ZX}. Then any cross terms between $H_{\sigma^x}$ or $\bar{H}_{\sigma^x}$ and the other generators $H_{\sigma^z}$ or $\bar{H}_{\sigma^z}$ must vanish.

The only cross-terms left to investigate are those between $H_{\sigma^z|X}$ and $\bar{H}_{\sigma^z|X}$.
\begin{align*}
        H_{\sigma^z|X}^n \rho_X H_{\sigma^z|X}^{-n}\bar{H}_{\sigma^z|X}^m \rho_X \bar H_{\sigma^z|X}^{-m} = |G_Y|^2 |G|^{-2}\sum_{\substack{g,h\in G/G_Y\\\tilde g_X,\tilde h _X\in G_X}} H_{\sigma^z|X}^n g_X\ket{0_X}\bra{0_X}\tilde{g}_X^\dagger g_X^\dagger H_{\sigma^z|X}^{-n} \bar H_{\sigma^z|X}^m h_X\ket{0_X}\bra{0_X} \tilde{h}_X^\dagger h_X^\dagger \bar H_{\sigma^z|X}^{-m}
\end{align*}
\begin{align}
    &= |G_Y|^2 |G|^{-2} \sum_{k=1}^N\sum_{l=1}^{\gcd(L_x,N)}  \sum_{\substack{g,h\in G/G_Y\\\tilde g_X \in S_X^{(k)},\tilde h _X\in \bar S_X^{(l)}}} \omega^{nk+mlN/\gcd(L_x,N)} g_X\ket{0_X}\bra{0_X}\tilde{g}_X^\dagger g_X^\dagger h_X\ket{0_X}\bra{0_X} \tilde{h}_X^\dagger h_X^\dagger\nn
    &= |G_Y|^2 |G|^{-2} \sum_{k=1}^N\sum_{l=1}^{\gcd(L_x,N)} \sum_{\substack{g\in G/G_Y\\\tilde g_X \in S_X^{(k)},\tilde h _X\in \bar S_X^{(l)}}} \omega^{nk+mlN/\gcd(L_x,N)} g_X\ket{0_X}\bra{0_X} g_X^\dagger \tilde{g}_X^\dagger \tilde{h}_X^\dagger\nn
    &= |G_Y|^2 |G|^{-2} \sum_{k,r=1}^N \sum_{l,s=1}^{\gcd{(L_x,N)}} \sum_{\substack{g\in G/G_Y\\\tilde o_X \in S_X^{(k+r)}\cap \bar S_X^{(l+s)}}}  \omega^{nk+mlN/\gcd(L_x,N)}  g_X\ket{0_X}\bra{0_X} g_X^\dagger \tilde{o}_X^\dagger\nn
    &=|G_Y|^2 |G|^{-2} \sum_{p,r=1}^N \sum_{q,s=1}^{\gcd{(L_x,N)}} \sum_{\substack{g\in G/G_Y\\\tilde o_X \in S_X^{(p)}\cap \bar S_X^{(q)}}}  \omega^{n(p-r)+m(q-s)N/\gcd(L_x,N)}  g_X\ket{0_X}\bra{0_X} g_X^\dagger \tilde{o}_X^\dagger \nonumber\\
    &={|G_Y|^2}{ |G|^{-2}} \sum_{p,q}  \sum_{\substack{g\in G/G_Y\\\tilde o_X \in S_X^{(p)}\cap \bar S_X^{(q)}}} \left(\sum_{r=1}^N\omega^{-nr}\right)\left(\sum_{s=1}^{\gcd{(L_x,N)}}\omega^{-msN/\gcd (L_x,N)}\right) \omega^{np+mqN/\gcd(L_x,N)}  g_X\ket{0_X}\bra{0_X} g_X^\dagger \tilde{o}_X^\dagger \nonumber \\&=0,
\end{align}
where $S_{X}^{(k)}=\{g_X\in G_X~|~ g_X H_{\sigma^z|X}=\omega^k H_{\sigma^z|X}g_X\}$ and $\bar S_{X}^{(l)}=\{g_X\in G_X|g_X \bar H_{\sigma^z|X}=\omega^{lN/\gcd (N,L_x)} \bar H_{\sigma^z|X}g_X\}$. A similar cancellation happens for the cross terms between $H_{\sigma^x|X}$ and $\bar{H}_{\sigma^x|X}$. These cancellations are expected, because the different logical operators transform $\rho_X$ into the reduced density matrices of different MESs, which have orthogonal support in subregion $X$. 

Then the total symmetrization yields
\begin{align}
    (\rho_X^\text{sym})^2 &= \frac{1}{\left[N^2\gcd(N,L_x)^2\right]^2}
    \sum_{n,m,p,q}H^n_{\sigma^z|X}\bar{H}_{\sigma^z|X}^m H_{\sigma^x|X}^p\bar{H}_{\sigma^x|X}^q \rho_X^2 \bar{H}_{\sigma^x|X}^{-q} H_{\sigma^x|X}^{-p} \bar{H}_{\sigma^z|X}^{-m} H^{-n}_{\sigma^z|X}\nn
    &= \left(\frac{|G_X||G_Y|}{N^2\gcd(N,L_x)^2|G|}\right)\rho_X^{\rm sym}.
\end{align}
Finally, the entanglement asymmetry is
\begin{equation}
    \Delta S_{X_y}^{\rm max} = 2\log [N\gcd(N,L_x)].
\end{equation}

\subsubsection{\texorpdfstring{$X\in A_x$}{Ax}}

The MES for $X\in A_x$ is 
\begin{equation}
    |{\rm MES}_x\rangle=\frac{1}{|G|^{1/2}}\sum_{g \in G} g|\textbf{0}\rangle,
\end{equation}
where the group $G$ is generated by all $V_{\bf r}$'s, $H_{\sigma^x}$, and $\bar{H}_{\sigma^x}$, and the reference state is the same as in the example above, $|\textbf{0}\rangle\equiv |0\rangle^{\otimes 2L_xL_y}$, for $Z|0\rangle=|0\rangle$. Our symmetry operators are $V_{\sigma^x}$, $\bar{V}_{\sigma^x}$, $V_{\sigma^z}$ and $\bar{V}_{\sigma^z}$, so the symmetrized reduced density matrix has the form
\begin{equation}
    \rho_X^\text{sym}=\frac{1}{N^2\gcd(N,L_x)^2}\sum_{n,m,p,q}V^n_{\sigma^z|X}\bar{V}_{\sigma^z|X}^m V_{\sigma^x}^p\bar{V}_{\sigma^x}^q \rho_X \bar{V}_{\sigma^x}^{-q} V_{\sigma^x}^{-p} \bar{V}_{\sigma^z|X}^{-m} V^{-n}_{\sigma^z|X}.
\end{equation}
Since the symmetry operators and the stabilizers in $G$ obey the same algebra as in the previous example (see \eqref{eq:ZX}), the previous calculations and arguments transfer directly here. Then the cross terms of $( \rho_X^\text{sym})^2$ vanish, and we get
\begin{align}
        (\rho_X^\text{sym})^2 &= \left(\frac{|G_X||G_Y|}{N^2\gcd(N,L_x)^2|G|}\right)\rho_X^{\rm sym},\\
         \Delta S_{X_x}^{\rm max} &= 2\log [N\gcd(N,L_x)].
\end{align}
 Note that this is the same entanglement asymmetry as in the previous example, where we considered a perpendicular noncontractible region and its associated MES, despite this model being anisotropic in the $x$ and $y$ directions of the torus. This differs from what we find for the TEE, which is calculated in Appendix \ref{asec:D}.

\section{Topological entanglement entropy}
\label{asec:D}

In this section, we calculate the TEE on a noncontractible annulus for the MESs associated with $X\in \mathtt{A}_{x(y)}$ of the anisotropic dipolar toric code on an $L_x \times L_y$ square lattice.

\subsection{\texorpdfstring{$\gamma_{\rm{LW},y}^{\rm max}$}{gamma (LW,y)}}

For the calculation of $\gamma_{{\rm LW},y}^{\rm max}$, we will use the partition illustrated in Fig.~\ref{fig:TEE_caly}. The MES associated with $X \in \mathtt{A}_y$ used in our calculation is
\begin{equation}
    |{\rm MES}_y\rangle=\frac{1}{|G|^{1/2}}\sum_{g \in G} g|0\rangle,
\end{equation}
where the group $G$ is generated by all $V_{\bf r}$'s, $V_{\sigma^x}$, and $\bar{V}_{\sigma^x}$ in \eqref{eq:lo}. Here, $S(\rho_{AB})$, $S(\rho_{BC})$, and $S(\rho_{B})$ are readily obtained from \eqref{eq:SBC}, since the entanglement entropy of a disk is identical for all distinct ground states. They are given by
\begin{align}
S(\rho_{AB})&= (2l_x+4l_A+4l_B-4)\log N,\nn
S(\rho_{BC})&= (2l_x+4L_y-4l_A-4)\log N,\nn
S(\rho_{B})&= (4l_x+4l_B-8)\log N.
\end{align}
It therefore remains only to determine $S(\rho_{ABC})$.

\begin{figure}[ht]
    \centering
    \includegraphics[width=0.4\linewidth]{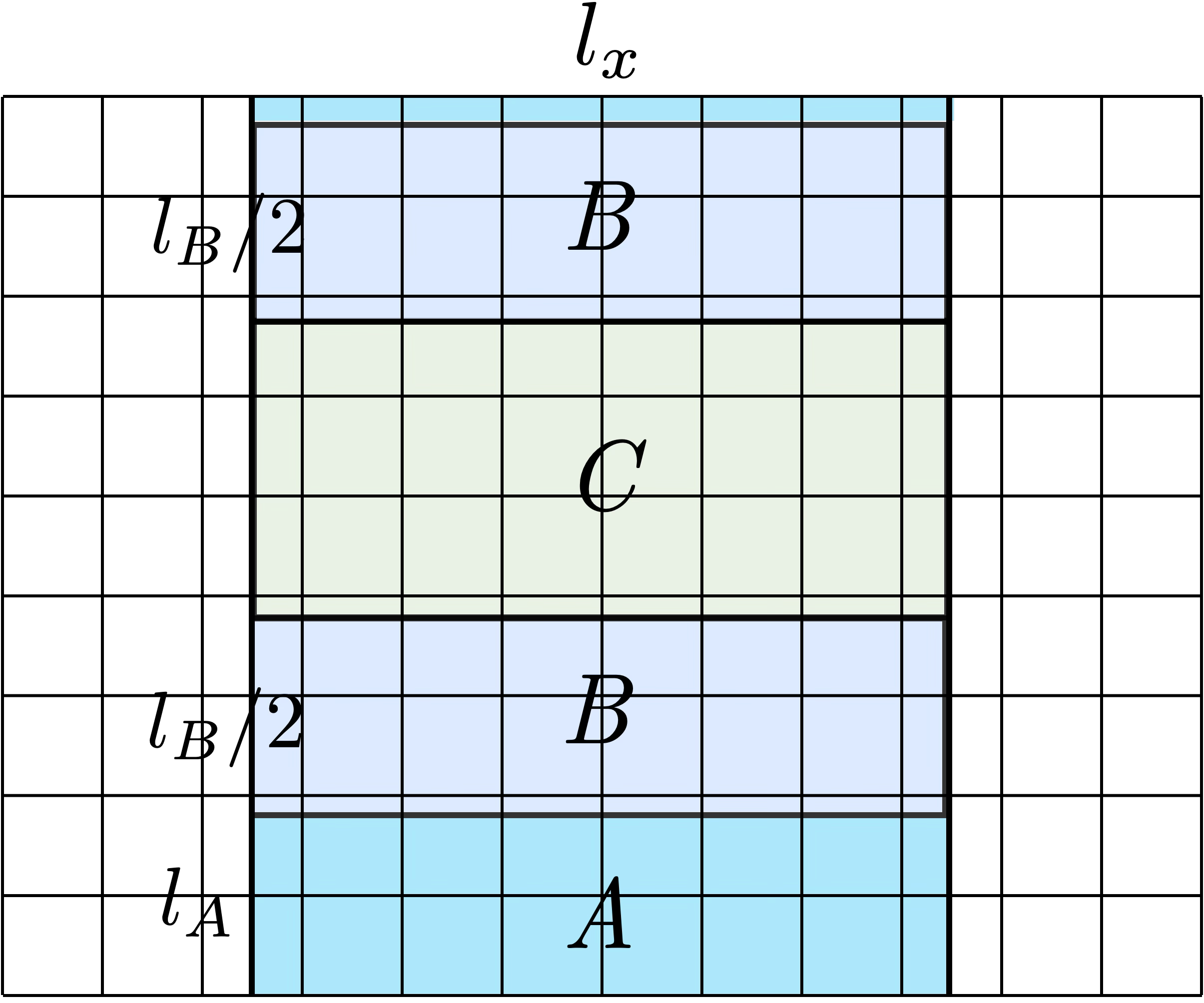}
    \caption{Partition used for the calculation of $\gamma_{{\rm LW},y}^{\rm max}$}
    \label{fig:TEE_caly}
\end{figure}

We first determine $|G_{ABC}|$. There are $N^{(l_x-2)L_y}$ ways to arrange the $V_{\bf r}$ operators completely inside of $ABC$. We get another factor of $N$ for the $V_{\sigma^x}$ operators that may multiply all the $V_{\bf r}$ operators that are inside of $ABC$. One can do the same with the $\bar V_{\sigma^x}$ operators, which gives another factor of $N$. Then $|G_{ABC}|=N^{(l_x-2)L_y+2}$.

To find $|G_{(ABC)^c}|$, we can simply divide $|G|$ by the number of operators that are not completely inside of $(ABC)^c$. We need to account for the number $(l_x-2)L_y$ of $V_{\bf r}$ operators that are completely inside of $ABC$, and for the number $4L_y$ of $V_{\bf r}$ operators that are at the boundary ($L_y$ on each side of each boundary). Here, we assume that $V_{\sigma^x}$ and $\bar V_{\sigma^x}$ are already defined completely inside of $ABC$, so there are no further constraints from that. So far, we have $|G_{(ABC)^c}|\propto |G|/N^{(l_x-2)L_y+4L_y}$. However, we must also account for the operators
\begin{align}
&M_1=\prod_{{\bf r}\in ABC}V_{\bf r}, &M_2=\prod_{{\bf r}\in ABC} V_{\bf r}^{r_x},
\end{align}
which are completely inside of $(ABC)^c$ once the multiplication of the $V_{\bf r}$ inside of $ABC$ is carried through. Therefore, the final answer for $|G_{(ABC)^c}|$ is
\begin{equation}
    |G_{(ABC)^c}| = \frac{|G|N^2}{N^{(l_x-2)L_y+4L_y}}.
\end{equation}
Then the entanglement entropy for $ABC$ is
\begin{align}
    S(\rho_{ABC}) &= \log\left(\frac{|G|}{|G_{ABC}||G_{(ABC)^c}|}\right)=\log\left(\frac{N^{(l_x-2)L_y+4L_y}}{N^{(l_x-2)L_y+2}N^2}\right)\\
    &=\log\left(N^{4L_y-4}\right) = 4L_y\log N - 4\log N,
\end{align} 
from which we see that $\gamma_{{\rm LW},y}^{\rm max}=2\log N$.

\subsection{\texorpdfstring{$\gamma_{{\rm LW},x}^{\rm max}$}{gamma_{LW,x}}}
For the calculation of $\gamma_{{\rm LW},x}$, we will use the partition illustrated in Fig.~\ref{fig:TEE_calx}. The MES associated with $X \in \mathtt{A}_x$ used in our calculation is
\begin{equation}
    |{\rm MES}_x\rangle=\frac{1}{|G|^{1/2}}\sum_{g \in G} g|0\rangle,
\end{equation}
where the group $G$ is generated by all $V_{\bf r}$'s, $H_{\sigma^x}$, and $\bar{H}_{\sigma^x}$ in \eqref{eq:lo}. Here, $S(\rho_{AB})$, $S(\rho_{BC})$, and $S(\rho_{B})$ are readily obtained from \eqref{eq:SBC}, since the entanglement entropy of a disk is identical for all distinct ground states. They are given by
\begin{align}
S(\rho_{AB})&= (2l_A+2l_B+4l_y-4)\log N,\nn
S(\rho_{BC})&= (2L_x-2l_A+4l_y-4)\log N,\nn
S(\rho_{B})&= (2l_B+8l_y-8)\log N.
\end{align}
It therefore remains only to determine $S(\rho_{ABC})$.

\begin{figure}[ht]
    \centering
    \includegraphics[width=0.33\linewidth]{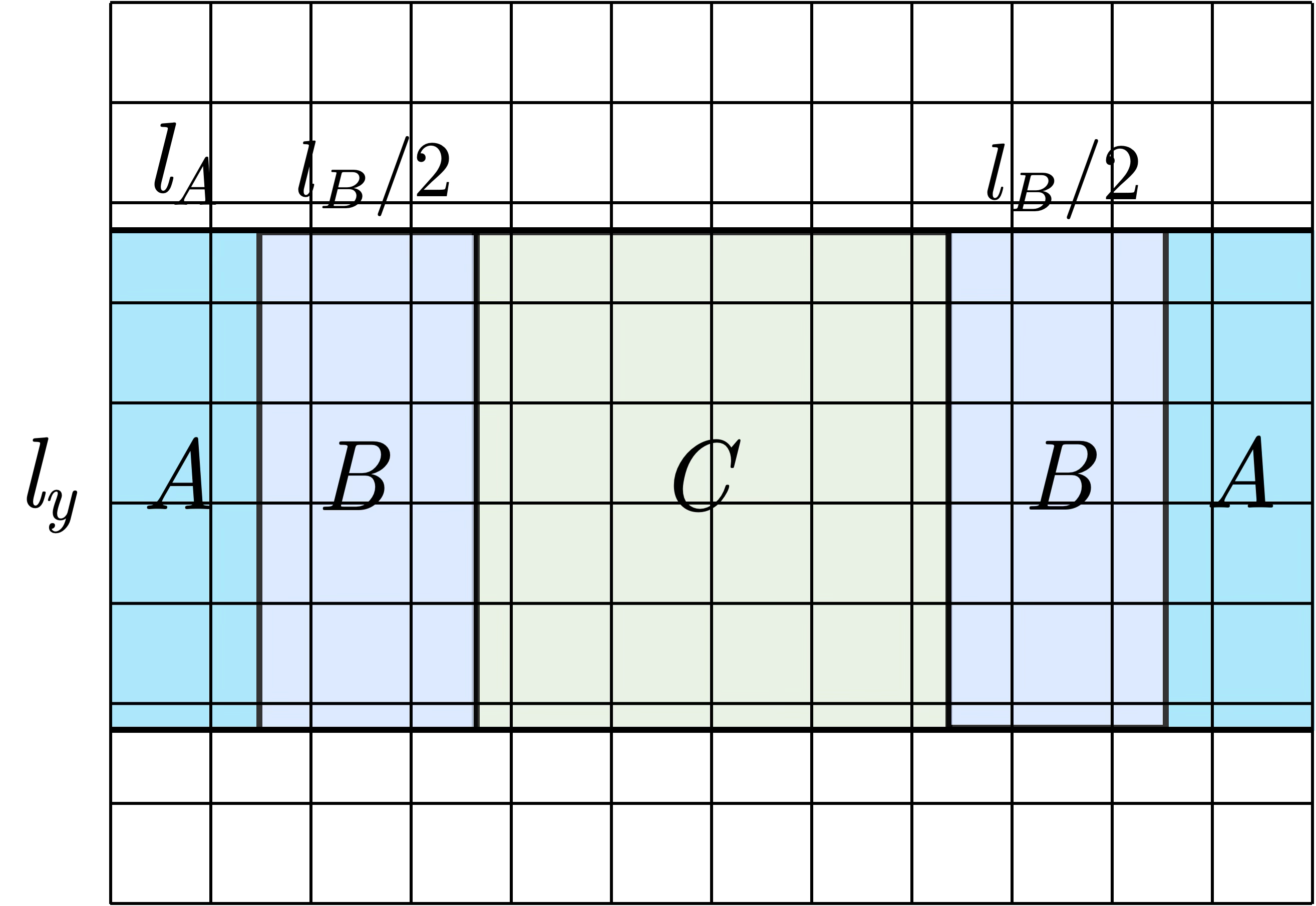}
    \caption{Partition used for the calculation of $\gamma_{{\rm LW},x}^{\rm max}$}
    \label{fig:TEE_calx}
\end{figure}

There are $N^{(l_y-1)L_x}$ elements of $G_{ABC}$ generated exclusively by $V_{\bf r}$ operators which are completely inside of region $ABC$. Furthermore, we get an extra factor of $N$ by acting on those operators inside of $ABC$ with $H_{\sigma^x}$. The other operator, $\bar{H}_{\sigma^x}$, requires more care. Because $\bar{H}_{\sigma^x}$ is a product of operators of the form $(\sigma^x_{\bf r})^{Nr_x/\gcd(N,L_x)}$ going around the $x$-cycle of the torus, there are only $\gcd(L_x,N)$ independent powers of that operator which are completely inside of $ABC$. Then $|G_{ABC}|=N^{(l_y-1)L_x+1}\gcd(L_x,N)$.

We now need to calculate $|G_{(ABC)^c}|$. Accounting just for the $V_{\bf r}$ completely inside of $ABC$ and the $V_{\bf r}$ at the boundaries gives $|G|/N^{(l_y-1)L_x+2L_x}$. Here we also must account for operators that have support only on $(ABC)^c$ but can be obtained from operators in $ABC$. These are 
\begin{align}
    &M_3=\prod_{{\bf r}\in ABC}V_{\bf r},&M_4=\prod_{{\bf r}\in ABC}V_{\bf r}^{{Nr_x}/{\gcd{(N,L_x)}}}.
\end{align}
While $M_3$ contributes another factor of $N$ to $|G_{(ABC)^c}|$, $M_4$ contributes only a factor of $\gcd(N,L_x)$, as there are only $\gcd(N,L_x)$ independent powers of $M_4$.

Indeed, only those powers of $M_4$ for which the factors $V_{\bf r}^{Nr_x/\gcd(N,L_x)}$ cancel upon winding around the $x$-cycle of the torus within $ABC$ can be supported entirely inside $(ABC)^c$, and there are $\gcd(N,L_x)$ independent operators of this form. Finally, we obtain 
\begin{equation}
    |G_{(ABC)^c}|=N\gcd(N,L_x)\frac{|G|}{N^{(l_y-1)L_x+2L_x}}.
\end{equation}

The entanglement entropy $S(\rho_{ABC})$ is then given by
\begin{align}
    S(\rho_{ABC}) &= \log\left(\frac{|G|}{|G_{ABC}||G_{(ABC)^c}|}\right)=\log\left(\frac{N^{(l_y-1)L_x+2L_x}}{N^{(l_y-1)L_x}N^2\gcd^2(N,L_x)}\right)\\
    &=2L_x\log N-2\log [N \gcd (N,L_x)].
\end{align} 
Therefore, $\gamma_{{\rm LW},x}^{\rm max}=\log [N \gcd (N,L_x)]$.

\twocolumngrid

\bibliography{references}

\end{document}